\documentclass[onefignum,onetabnum]{siamonline181217}

\usepackage{lipsum}
\usepackage{amsfonts}
\usepackage{graphicx}
\usepackage{epstopdf}
\usepackage{algorithmic}
\ifpdf
  \DeclareGraphicsExtensions{.eps,.pdf,.png,.jpg}
\else
  \DeclareGraphicsExtensions{.eps}
\fi

\usepackage{enumitem}
\setlist[enumerate]{leftmargin=.5in}
\setlist[itemize]{leftmargin=.5in}

\newsiamremark{remark}{Remark}
\newsiamremark{hypothesis}{Hypothesis}
\crefname{hypothesis}{Hypothesis}{Hypotheses}
\newsiamthm{claim}{Claim}

\headers{The benefits of acting locally}{B. Adcock, S. Brugiapaglia, and M. King-Roskamp}

\title{The benefits of acting locally: Reconstruction algorithms for sparse in levels signals with stable and robust recovery guarantees\thanks{Submitted to the editors on \today\funding{BA and MKR acknowledge the support of the PIMS CRG “High-dimensional Data Analysis”, SFU’s Big Data Initiative “Next Big Question” Fund and by NSERC through grant R611675. SB acknowledges NSERC through grant RGPIN-2020-06766 and the Faculty of Arts and Science of Concordia University.}}}

\author{Ben Adcock\thanks{Simon Fraser University, Burnaby, BC, Canada. 
  (\email{ben\_adcock@sfu.ca})}
\and Simone Brugiapaglia\thanks{Concordia University, Montr\'eal, QC, Canada. 
  (\email{simone.brugiapaglia@concordia.ca}) }
\and Matthew King-Roskamp\thanks{Simon Fraser University, Burnaby, BC, Canada. 
  (\email{mkingros@sfu.ca}})}

\usepackage{amsopn}

\ifpdf
\hypersetup{
  pdftitle={The benefits of acting locally},
  pdfauthor={B. Adcock, S. Brugiapaglia, and M. King-Roskamp}
}
\fi

\usepackage{tcolorbox,bm,amssymb}

\DeclareMathOperator*{\argmin}{\arg\!\min}
\DeclareMathOperator*{\argmax}{\arg\!\max}

\usepackage{ulem}

\begin{document}

\maketitle

\begin{abstract}
The sparsity in levels model recently inspired a new generation of effective acquisition and reconstruction modalities for compressive imaging. Moreover, it naturally arises in various areas of signal processing such as parallel acquisition, radar, and the sparse corruptions problem. Reconstruction strategies for sparse in levels signals usually rely on a suitable convex optimization program. Notably, although iterative and greedy algorithms can outperform convex optimization in terms of computational efficiency and have been studied extensively in the case of standard sparsity, little is known about their generalizations to the sparse in levels setting. In this paper, we bridge this gap by showing new stable and robust uniform recovery guarantees for sparse in level variants of the iterative hard thresholding and the CoSaMP algorithms. Our theoretical analysis generalizes recovery guarantees currently available in the case of standard sparsity and favorably compare to sparse in levels guarantees for weighted $\ell^1$ minimization. In addition, we also propose and numerically test an extension of the orthogonal matching pursuit algorithm for sparse in levels signals.
\end{abstract}



\setcounter{page}{1}

\section{Introduction}

The model of \textit{sparse vectors} has proven exceptionally useful in a wide range of mathematical and engineering  applications. This underlying low-dimensional structure can be taken advantage of by compressed sensing techniques to recover an $s$-sparse vector $x \in \mathbb{C}^{N}$ from noisy linear measurements $y = Ax+e \in \mathbb{C}^{m}$. Furthermore, many applications exhibit structure beyond classical sparsity. Hence, there has been study on more complex \textit{structured sparsity} models such as group or block sparsity, joint sparsity, weighted sparsity, connected tree sparsity and numerous others. In fact, many of these more sophisticated models can lead to boosted practical performance  \cite{baraniuk2010model,DuarteEldarStructuredCS,TraonmilinGribonvalRIP}. 

This paper focuses on the so-called \textit{sparsity in levels} model, which has been shown to provide significant theoretical and practical gains over the standard sparsity model \cite{AHPRBreaking, BastounisHansen}. Sparse in levels vectors exhibit a local sparsity pattern, specified by a vector $(s_{1}, \ldots, s_{r})$, as opposed to a single sparsity $s$. This simple generalization leads to a rich theory of compressed sensing extending naturally from the classical setting \cite{AHPRBreaking,LiAdcockRIP}. To date, the sparse in levels model has been exploited using convex optimization-based decoders. This paper generalizes the corresponding theory of recovery algorithms for sparse vectors, and provides theoretical guarantees of stability and robustness for the recovery of sparse in levels vectors. We focus on  three standard iterative and greedy algorithms: Iterative Hard Thresholding (IHT), Orthogonal Matching Pursuit (OMP) and Compressive Sampling Matching Pursuit (CoSaMP), which are natural first algorithms of interest \cite{BlumensathDavies2008, NeedellTropp2008}. For each, we derive and study a suitable generalization to the sparse in levels model.

\subsection{Motivations}

The sparsity in levels model arises naturally in various contexts. For instance, it can be used to model so-called \textit{sparse and distributed} or \textit{sparse and balanced} vectors, which occur in parallel acquisition problems \cite{ChunAdcock16ITW,AdcockChunParallel} and radar \cite{Dorsch2016}.  The specific case of two levels also arises in the \textit{sparse corruptions} problem \cite{BAEtAlCorruptions,LiCorruptionsConstrApprox}, in which, rather than standard Gaussian or uniformly bounded noise, a small fraction of the measurements of a signal is substantially corrupted. Another natural context of interest is the problem of compressive imaging, where sparse in levels vectors model the wavelet coefficients of natural images. This model allows one to design better sampling strategies over those optimized for standard sparse models, which leads to enhanced recovery performance \cite{AHPRBreaking,OptimalSamplingQuest,AsymptoticCS}. As noted, sparsity in levels has hitherto been exploited using optimization-based decoders. Yet, it is well known that such decoders have some limitations. For instance, they can be computationally intensive. Moreover, a decoder based on minimizing a convex optimization problem is not a method \textit{per se}, as it requires a secondary algorithm to actually compute a solution. Therefore, as noted in \cite{BASBMKRCSwavelet}, there is a gap between compressed sensing theory based on minimizers of optimization problems and its practical performance. With this in mind, the  primary motivation for this work is to derive algorithms for the sparsity in levels model in compressed sensing which are provably stable and robust, and which are also provably polynomial time in $m$ and $N$.

\subsection{Contributions}

 The main contributions of this work are the robustness and stability guarantees for the levels-based algorithms IHTL and CoSaMPL. These are presented in Theorem~\ref{thm:IHTFinalTheorem} and Theorem~\ref{thm:CoSamPlMainTheorem} respectively. These results determine an error bound in certain weighted $\ell^1$-norms depending on the approximate sparsity in levels, and the noise level.
They directly generalize known results for the sparse case, and  require no stricter assumptions on the corresponding restricted isometry constant. Interestingly, in the equivalent result for the optimization-based Quadratically-Constrained Basis Pursuit (QCBP) decoder with the sparsity in levels model (Theorem~\ref{thm:simplerQCBPresult}), the condition on the restricted isometry constant scales with the number of levels -- a phenomenon which does not occur in either IHTL or CoSaMPL.
We also propose a generalization of OMP to the levels setting, and examine the numerical performance of these iterative algorithms. Generally, we find that the levels based generalizations IHTL, CoSaMPL improve over their non-local counterparts, whereas OMPL shows situational improvement.

\subsection{Previous work}

\label{s:previous_work}

The IHT and CoSaMP algorithms were introduced to compressed sensing in \cite{BlumensathDavies2008} and \cite{NeedellTropp2008} respectively. Their theoretical analysis can be found, for instance, in \cite{FoucartRauhutCSbook}. The \textit{Iterated Shrkinage} methods~\cite{elad2007wide} served as a precursor for IHT, which was introduced in the context of compressed sensing in the late 2000s~\cite{BlumensathDavies2008,BlumensathDavies2009}. Accelerating IHT using variable stepsize was examined later~\cite{Blumensath2012,BlumensathDavies2010}. The IHTL and CoSaMPL algorithms have been previously examined numerically in \cite{adcock2019iterative} by the authors, wherein these algorithms were first introduced. However, this previous work contained no theoretical analysis, and did not consider OMPL. 
The sparsity in levels model was first introduced in \cite{BastounisHansen}. Nonuniform recovery guarantees for the $\ell^{1}$ minimization decoder were proven first in \cite{AHPRBreaking}, with uniform guarantees later in \cite{LiAdcockRIP} and \cite{AAHWalshWavelet}. This analysis is in a form that is not immediately comparable to our results for IHTL or CoSaMPL. Hence we also provide a slightly modified analysis herein.

As noted, sparsity in levels is a particular type of structured sparsity model. General structured sparsity models are typically formalized within the setting of a \textit{union of subspaces} model \cite{baraniuk2010model,lu2008theory,blumensath2009sampling,eldar2009robust,baraniuk2010model,duarte2011structured}. Arguably starting with the work of \cite{baraniuk2010model}, there has been a substantial amount of research on extending compressed sensing to such general models. Theoretical guarantees for general models can be found in \cite{TraonmilinGribonvalRIP,dirksen2016dimensionality,junge2017generalized} and references therein. Practically, such structure is often promoted through either convex decoders \cite{TraonmilinGribonvalRIP,bach2012structured,micchelli2013regularizers,yu2011solving} or structured iterative or greedy algorithms \cite{baraniuk2010model}. In particular, generalizations of IHT and CoSaMP to the union of subspaces model were introduced and analyzed in \cite{baraniuk2010model,HedgeIndykSchmidt2015}.

Sparsity in levels is a particular type of union of subspaces model.
Note that the IHTL and CoSaMPL algorithms we consider in this paper are special cases of those introduced in \cite{baraniuk2010model} corresponding to the sparsity in levels model. However, the analysis we conduct in this paper, being specific to the sparsity in levels model, is simpler, in that it avoids additional conditions on the measurement matrix required for more general models. We also note in passing that the sparsity in levels models is explicitly parameterized by a vector of local sparsities $(s_1,\ldots,s_r)$. This renders it different, at least semantically, from other structured sparsity models that use only one parameter to quantify sparsity. We further note that sparsity in levels is different from the well-known block sparsity and group sparsity models. These models separate a vector into a finite number of (potentially overlapping) groups, of which at most $s$ may be nonzero.

Most works in structured sparsity derive uniform recovery guarantees with an \textit{a priori} imposed model. Another line of work, initiated by \cite{BoyerBlockStructured} (see also \cite{AdcockChunParallel}) strives to derive recovery guarantees that are local to the support of the underlying vector, i.e.\ independent of any structured sparsity model. This work is done in the context of convex regularizers. It is not clear whether it can be applied to iterative and greedy algorithms, in which the proofs heavily rely on uniform recovery tools, such as structured Restricted Isometry Properties (RIPs).

Recently, the IHTL algorithm has also been independently introduced in \cite{donato2020structured}. An analysis based on the matrix coherence was performed, as well as applications to inverse source problems and off-the-grid recovery. Note that in this paper we provide an analysis based on a type of RIP. It is known that coherence-based analyses generically lead to weaker recovery guarantees than RIP-based approaches (the so-called quadratic bottleneck in standard compressed sensing \cite{FoucartRauhutCSbook}).

\subsection{Outline}

We begin in Section \ref{sec:SparseCase} by reviewing the theory of iterative and greedy methods in the sparse case. This should serve as an introduction, and parallel the ideas in Section \ref{sec:LevelesCase}, in which the sparsity in levels model is discussed. Section \ref{sec:Algs} defines the algorithms of interest for this work, IHTL, CoSaMPL and OMPL. Section \ref{sec:MainResults} contains the statements of the main results, and some discussion on useful special cases. Some numerical demonstrations follow in Section \ref{sec:numerics}. Section \ref{sec:proofs} contains the proofs of these results, and begins itself by outlining the strategy used. Finally, we summarize and state directions of future work in Section \ref{sec:conclusion}.

\subsection{Notation} If $X \leq C Y$, where $C>0$ is a constant independent of any quantity involved in $X$ and $Y$, we write $X \lesssim Y$.  Given a vector $z \in \mathbb{C}^N$, we write $\|x\|_{\ell^p} = (\sum_{i = 1}^N{|x_i|^p})^{1/p}$ for the $\ell^p$-norm of $z$, for any $p>0$. If $w \in \mathbb{R}^N$, the relations $w>0$ or $w = 1$ are read componentwise (i.e., $w_i > 0$ or $w_i = 1$, for every $i = 1, \ldots, N$, respectively). Given a vector of weights $w \in \mathbb{R}^N$ such that $w >0$, we refer to $\|x\|_{\ell_w^1} = \sum_{i = 1}^N w_i |x_i|$ as the weighted $\ell^1$-norm of $z$. We denote the standard inner product of $\mathbb{C}^N$ by $\langle x, y \rangle = \sum_{i = 1}^N x_i y_i^{*}$, for any $x, y \in \mathbb{C}^N$.  $A^*$ denotes the conjugate transpose of a matrix $A$.

\section{Preliminaries: Sparse Case\label{sec:SparseCase}}

Before developing the levels-based techniques, we recall the theory of compressed sensing for the sparse case. This will serve as an important special case and comparison point for the levels theory. A vector $x = (x_i)^{N}_{i=1} \in \mathbb{C}^N$ is \textit{$s$-sparse} if it has at most $1 \leq s \leq N$ nonzero entries: that is,
\begin{equation*}
| \mathrm{supp}(x) | \leq s,    
\end{equation*}
where $\mathrm{supp}(x) = \{ i : x_i \neq 0 \}$ is the \textit{support} of $x$.  Classical compressed sensing concerns the recovery of a sparse vector $x$ from $s \leq m \leq N$ noisy linear measurements
\begin{equation*}
y = A x + e \in \mathbb{C}^m,    
\end{equation*}
where $A \in \mathbb{C}^{m\times N}$ is the measurement matrix and $e \in \mathbb{C}^m$ is an unknown noise vector.
The \textit{best $s$-term approximation error} of $x \in \mathbb{C}^N$ (with respect to the $\ell^1$-norm) is defined as: 
\begin{equation*}
    \sigma_{s}(x)_{\ell^1} = \inf_{z \in \mathbb{C}^{N}} \{ \Vert x- z \Vert_{\ell^1} : \mbox{$z$ is $s$-sparse} \}.
\end{equation*}

A compressed sensing recovery procedure seeks to approximate the true solution $x$ with a vector $\hat{x}$ that is as close as possible to $x$. A standard theoretical property used to assure such recovery is the \textit{Restricted Isometry Property}:

\begin{definition}
Let $1 \leq s \leq N$.  The $s$-th \textit{Restricted Isometry Constant (RIC)} $\delta_s$ of a matrix $A \in \mathbb{C}^{m \times N}$ is the smallest $\delta \geq 0$ such that
\begin{equation}
\label{RIP}
(1-\delta) \| x \|^2_{\ell^2} \leq \| A x \|^2_{\ell^2} \leq (1+\delta) \| x \|^2_{\ell^2},\quad \mbox{for all $s$-sparse $x$}.
\end{equation}
If $0 < \delta_s < 1$ then $A$ is said to have the \textit{Restricted Isometry Property (RIP)} of order $s$.
\end{definition}

Matrices that satisfy the RIP have been well studied, with a classical example being that of Gaussian random matrices (See e.g.~\cite[Theorem 9.27]{FoucartRauhutCSbook}). 
It is well know that other large classes of random matrices satisfy the RIP with high probability, such as subgaussian matrices, Bernoulli matrices, and subsampled bounded orthonormal systems \cite{FoucartRauhutCSbook}.

\subsection{IHT and CoSaMP}

For a vector $x \in \mathbb{C}^{N}$ (not necessarily sparse), let $L_{s}(x)$ be the index set of its $s$ largest entries in absolute value.  The \textit{hard thresholding} operator $H_{s} : \mathbb{bbC}^N \rightarrow \mathbb{C}^N$ is, for $x = (x_i)^{N}_{i=1} \in \mathbb{C}^{N}$,  defined by
\begin{equation*}
H_{s}(x) = (H_{s}(x)_{i})_{i=1}^{N},\qquad    H_{s}(x)_{i} = \begin{cases} x_{i} &i \in L_{s}(x) \\
    0 &\text{otherwise}
        \end{cases}.
\end{equation*}
That is, $H_{s}(x)$ is the vector of the $s$ largest entries of $x$ with all other entries set to zero.  The classical \textit{Iterative Hard Thresholding (IHT)} algorithm is now defined as follows:

\begin{tcolorbox}
Function $\hat{x}=\mathrm{IHT}(A,y,s)$ 
\\
\noindent \textbf{Inputs:} $A \in \mathbb{C}^{m\times N}$, $y \in \mathbb{C}^m$, sparsity $s$ 
\\
\noindent \textbf{Initialization:} $x^{(0)} \in \mathbb{C}^N$ (e.g.\ $x^{(0)} = 0$)
\\
\noindent \textbf{Iterate:} Until some stopping criterion is met at $n = \overline{n}$, set 
\begin{equation*}
    x^{(n+1)}= H_{s}(x^{(n)} + A^{*}(y - Ax^{(n)}))
\end{equation*}

\noindent \textbf{Output:} $\hat{x} = x^{(\overline{n})}$
\end{tcolorbox}

The \textit{Compressive Sampling Matching Pursuit (CoSaMP)} algorithm is:

\begin{tcolorbox}
Function $\hat{x}=\mathrm{CoSaMP}(A,y,s)$ \\
\noindent \textbf{Inputs:} $A \in \mathbb{C}^{m\times N}$, $y \in \mathbb{C}^m$, sparsity $s$ \\
\noindent \textbf{Initialization:} $x^{(0)} \in \mathbb{C}^N$ (e.g.\ $x^{(0)} = 0$) 
\\
\noindent \textbf{Iterate:} Until some stopping criterion is met at $n = \overline{n}$, set 
\begin{align*}
    U^{(n+1)} &= \text{supp}(x^{(n)}) \cup L_{2s}(A^{*}(y-Ax^{(n)})) \\
    u^{(n+1)} &\in \argmin_{z \in \mathbb{C}^N} \{ \Vert y -Az \Vert_{\ell^{2}} \: : \: \text{supp}(z) \subset U^{(n+1)} \} \\
    x^{(n+1)} &= H_{s}(u^{(n+1)})
\end{align*}

\noindent \textbf{Output:} $\hat{x} = x^{(\overline{n})}$
\end{tcolorbox}

 As stated, the RIP is a sufficient condition for these algorithms to recover a sparse solution. Generalizing the following two results to the sparse in levels setting will be the overall goal of this work.

\begin{theorem} (E.g.\ \cite[Theorem 6.21]{FoucartRauhutCSbook})
Suppose that the $6s$-th RIC constant of $A \in \mathbb{C}^{m \times N}$ satisfies $\delta_{6s} < \frac{1}{\sqrt{3}}$. Then, for all $x \in \mathbb{C}^{N}$ and $e \in \mathbb{C}^{m}$, the sequence $(x^{(n)})_{n \geq 0}$ defined by $\mathrm{IHT}(A,y,2s)$ with $y = A x + e$ and $x^{(0)} = 0$ satisfies, for any $n \geq 0$,
\begin{align*}
 \Vert x - x^{(n)} \Vert_{\ell^{1}} &\leq C\sigma_{s}(x)_{\ell^{1}} + D \sqrt{s} \Vert e \Vert_{\ell^{2}} + 2 \sqrt{s} \rho^{n} \Vert x \Vert_{\ell^{2}}, \\
\Vert x - x^{(n)} \Vert_{\ell^{2}} &\leq \frac{C}{\sqrt{s}}\sigma_{s}(x)_{\ell^{1}} + D \Vert e \Vert_{\ell^{2}} + \rho^{n} \Vert x \Vert_{\ell^{2}},
\end{align*}
where $\rho = \sqrt{3}\delta_{6s} < 1$, and $C,D >0$ are constants only depending on $\delta_{6s}$.
\end{theorem}

\begin{theorem} (E.g.\ \cite[Theorem 6.28]{FoucartRauhutCSbook})  Suppose that the $8s$-th RIC constant of $A$ satisfies
\begin{equation*}
    \delta_{8s} < \tfrac{\sqrt{\frac{11}{3}} - 1}{4} \approx 0.478.
\end{equation*}
Then, for all $x \in \mathbb{C}^{N}$ and $e \in \mathbb{C}^{m}$ the sequence $(x^{(n)})_{n \geq 0}$ defined by $\mathrm{CoSaMP}(A,y,2s)$ with $y=Ax + e$ and $x^{(0)} = 0$, satisfies for any $n \geq 0$,
\begin{align*}
    \Vert x - x^{(n)} \Vert_{\ell^{1}} &\leq C\sigma_{s}(x)_{\ell^1} + D \sqrt{s} \Vert e \Vert_{\ell^{2}} + 2 \sqrt{s} \rho^{n} \Vert x \Vert_{\ell^{2}}, \\
\Vert x - x^{(n)} \Vert_{\ell^{2}} &\leq \frac{C}{\sqrt{s}}\sigma_{s}(x)_{\ell^{1}} + D\Vert e \Vert_{\ell^{2}} + 2 \rho^{n} \Vert x \Vert_{\ell^{2}},
\end{align*}
where  $\rho = \sqrt{\frac{2\delta_{8s}^{2}(1+3\delta_{8s}^{2})}{1-\delta_{4s}^{2}}} < 1$ and $C,D >0$ are constants only depending on $\delta_{8s}$.
\end{theorem}

We note the precise value of $\rho$ is not stated in the original results, but is derived within their proofs. We include this here, as the later main results share precisely the same constants. Finally, we state the formulation of the Orthogonal Matching Pursuit (OMP) algorithm, introduced in \cite{mallatzhang1993matching}. In this sparse setting, a similar theorem to the above holds. While we propose a generalization of OMP to the sparse-in-levels setting, we do not prove a theoretical result for this proposed algorithm, and provide simply some numerical evidence of good performance.

\begin{center}
\begin{tcolorbox}
Function $\hat{x}=\mathrm{OMP}(A,y,s)$ \\
\noindent \textbf{Inputs:} $A \in \mathbb{C}^{m\times N}$, $y \in \mathbb{C}^m$, sparsity $s$ \\
\noindent \textbf{Initialization:} $x^{(0)} \in \mathbb{C}^N$ (e.g.\ $x^{(0)} = 0$) , $S^{(0)} = \emptyset$
\\
\noindent \textbf{Iterate:} For each $k = 1, \ldots ,s$, set
         \begin{equation*}
             j_{k} \in \argmax_{j=1, \ldots, N} \vert (A^{*}(y- A x^{(k-1)}))_{j} \vert
         \end{equation*}
        Update $S^{(k)} = S^{(k-1)} \cup  \{ j_{k} \}$ \\
        Set $x^{(k)} \in \displaystyle\argmin_{z \in \mathbb{C}^N: \, \text{supp}(z) \subseteq S^{(k)}} \left\Vert y - A z \right\Vert_{\ell^{2}} $
     
\noindent \textbf{Output:} $\hat{x} = x^{(s)}$
\end{tcolorbox}
\end{center}

\section{Compressed sensing for sparse in levels vectors} \label{sec:LevelesCase}

With the sparse case summarized, we now move to the levels case. We first must define an appropriate local version of sparsity, recall an RIP-type property, and define the algorithms of interest.

\begin{definition}
Let $r \geq 1$, $\bm{M} = (M_1,\ldots,M_r)$, where $1 \leq M_1 < M_2 < \ldots < M_r = N$ and $\bm{s} = (s_1,\ldots,s_r)$, where $s_k \leq M_k - M_{k-1}$ for $k=1,\ldots,r$, with $M_0 = 0$.  A vector $x = (x_i)^{N}_{i=1} \in \mathbb{C}^N$ is \textit{$(\bm{s},\bm{M})$-sparse} if
\begin{equation*}
\left | \text{supp}(x) \cap \{ M_{k-1}+1,\ldots,M_k \} \right | \leq s_k,\quad k = 1,\ldots,r.
\end{equation*}
We write $\Sigma_{\bm{s},\bm{M}} \subseteq \mathbb{C}^N$ for the set of $(\bm{s},\bm{M})$-sparse vectors.
\end{definition}
 This model was first introduced in \cite{BastounisHansen}. We refer to $s = s_1+\ldots+s_r$ for the total sparsity, and we denote by $D_{\bm{s},\bm{M}} \subset \{ 1, \ldots M \}$ the set of all $(\bm{s},\bm{M})$-sparse index sets. We refer to $\bm{M}$ as \textit{sparsity levels} and $\bm{s}$ as \textit{local sparsities}. Moreover, any index set of the form $\{M_{k-1}+1,\ldots,M_k\}$ for some $k = 1,\ldots,r$ is said to be a \textit{level}. Of further use are the projection operators onto some index set and onto a level. Given $x \in \mathbb{C}^{N}$, these are defined as
\begin{align*}
    (P_{\Delta}x)_{i} & = \begin{cases} x_{i} & i \in \Delta \\ 0 & \text{otherwise}\end{cases}, \\
    (P^{M_{k}}_{M_{k+1}}x)_{i} & = \begin{cases} x_{i} & i \in \{ M_{k}+1, \ldots M_{k+1} \} \\ 0 & \text{otherwise}\end{cases}.
\end{align*}
While the latter is obviously a special case of the former, it is used with enough frequency to warrant special notation. 

\begin{definition}
Given a vector of weights $w \in \mathbb{R}^N$ with $w>0$, the \textit{best $(\bm{s},\bm{M})$-term approximation error} of $x \in \mathbb{C}^N$ (with respect to the weighted $\ell^1$-norm) is defined as
\begin{equation*}
    \sigma_{s}(x)_{\ell^{1}_{w}} = \inf_{z \in \mathbb{C}^{N}} \{ \Vert x- z \Vert_{\ell^{1}_{w}} : \mbox{$z$ is $(\bm{s},\bm{M})$-sparse} \}.
\end{equation*}
\end{definition}
This definition in the unweighted case was introduced alongside the sparsity in levels model in \cite{BastounisHansen}, and later extended to the weighted case in \cite{AAHWalshWavelet}.
Past works on convex optimization-based decoders for the sparsity in levels model have found that better uniform recovery guarantees can be obtained by replacing the $\ell^1$-norm with a suitable weighted $\ell^1$-norm \cite{AAHWalshWavelet,TraonmilinGribonvalRIP}. We shall find a similar phenomenon occurs in the case of iterative and greedy methods in Section \ref{sec:MainResults}. As in these previous works, we shall suppose that the weights are constant on each level
\begin{equation}
    w_i = w^{(k)},\quad M_{k-1} < i \leq M_k,\ 1 \leq k \leq r,
    \label{constantweights}
\end{equation}
for some $w^{(k)} > 0$. In particular, we shall typically make the choice
\begin{equation}\label{optimalweights}
    w^{(k)} = \sqrt{s/s_k},
\end{equation}
where $s_k$ is the $k$th local sparsity and $s$ is the total sparsity. Note that these weights are not involved in the definition of the levels-based algorithms proposed in Section~\ref{sec:Algs}. The introduction of this weighted setting is instead aimed at improving the recovery guarantees presented in Section~\ref{sec:MainResults}.

Much like the sparse setting, the main tool used to prove recovery guarantees is a restricted isometry property, here in levels:

\begin{definition}
Let $\bm{M} = (M_1,\ldots,M_r)$ be sparsity levels and $\bm{s} = (s_1,\ldots,s_r)$ be local sparsities.  The $(\bm{s},\bm{M})$-th \textit{Restricted Isometry Constant in Levels (RICL)} $\delta_{\bm{s},\bm{M}}$ of a matrix $A \in \mathbb{C}^{m \times N}$ is the smallest $\delta \geq 0$ such that
\begin{equation}
(1-\delta) \Vert x \Vert^2_{\ell^2} \leq \Vert A x \Vert^2_{\ell^2} \leq (1+\delta) \Vert  x \Vert^2_{\ell^2},\quad \forall x \in \Sigma_{\bm{s},\bm{M}}.
\end{equation}
If $0 < \delta_{\bm{s},\bm{M}} < 1$ then the matrix is said to have the \textit{Restricted Isometry Property in Levels (RIPL)} of order $(\bm{s},\bm{M})$.
\end{definition}

As expected, recovery results will require assumptions on the order of RIPL. While constructing matrices that satisfy the RIPL is not the purpose of this paper, we mention in passing that there are many examples in literature.
For example,
a random matrix with independent normal entries having mean zero and variance $1/m$ has the RIPL of order $\delta_{\bm{s},\bm{M}} \leq \delta$ with probability at least $1-\epsilon$, provided
\begin{equation}m \geq C \delta^{-2} \left ( \sum^{r}_{k=1} s_k \log \left ( \tfrac{\mathrm{e} (M_k - M_{k-1}) }{s_k} \right ) + \log(\epsilon^{-1}) \right ), 
\label{gaussmeasurmentCond} 
\end{equation}
which follows from \cite{dirksen2016dimensionality} as noted in \cite{LiAdcockRIP}.
Analogously to the sparse case, other wide classes of random matrices can be shown to satisfy the RIPL with high probability. For example, subsampled unitary matrices~\cite{LiAdcockRIP} and the particular case of binary sampling with the Walsh--Hadamard transform \cite{AAHWalshWavelet}. Designing matrices that take advantage of the local sparsity in levels structure of the vector $x$ leads to significant benefits~\cite{AHPRBreaking,AsymptoticCS}. 

\section{Levels-based algorithm definitions}
\label{sec:Algs}

Closely examining IHT and CoSaMP, the iterative step in either case is designed to promote that $(x^{(n)})_{n\geq0}$ converges to a solution to the linear system $Ax=y$. In particular the sparsity assumption is \textit{only} enforced via the thresholding operator. With this observation it is natural that a generalization of these algorithm to the sparse in levels setting should only alter the thresholding operation.

Let $\bm{M} = (M_1,\ldots,M_r)$ and $\bm{s} = (s_1,\ldots,s_r)$ be sparsity levels and local sparsities, respectively. Then any vector $x \in \mathbb{C}^N$ can be written uniquely as $x = \sum^{r}_{k=1} P_{M_{k}}^{M_{k-1}} x$.  For arbitrary $x \in \mathbb{C}^{N}$, we write $L_{\bm{s},\bm{M}}(x)$ for the set
\begin{equation*}
    L_{\bm{s},\bm{M}}(x) = \bigcup_{k=1}^{r} L_{s_{k}}\left(P^{M_{k-1}}_{M_{k}}x\right).
\end{equation*}
In other words this is the index set consisting of, in each level $\{ M_{k-1}+1, \ldots , M_{k} \}$, the largest absolute $s_{k}$ entries of $x$ in that level.
This allows the definition of the \textit{hard thresholding in levels} operator $H_{\bm{s},\bm{M}} : \mathbb{C}^N \rightarrow \mathbb{C}^N$ by $H_{\bm{s},\bm{M}}(x) = (H_{\bm{s},\bm{M}}(x)_{i})_{i=1}^{N}$ with
\begin{equation*}
H_{\bm{s},\bm{M}}(x)_{i} = \begin{cases} x_{i} &i \in L_{\bm{s},\bm{M}}(x) \\
    0 &\text{otherwise}
        \end{cases},\qquad x \in \mathbb{C}^{N}.  
\end{equation*}
 
With these preliminaries we can define the \textit{IHT in Levels (IHTL)} algorithm as

\begin{tcolorbox}
Function $\hat{x}=\mathrm{IHTL}(A,y,\bm{s},\bm{M})$ \\
\noindent \textbf{Inputs:} $A \in \mathbb{C}^{m\times N}$, $y \in \mathbb{C}^m$, local sparsities $\bm{s}$, sparsity levels $\bm{M}$
\\
\noindent \textbf{Initialization:} $x^{(0)} \in \mathbb{C}^N$ (e.g.\ $x^{(0)} = 0$) \\
\noindent \textbf{Iterate:} Until some stopping criterion is met at $n = \overline{n}$, set 
\begin{equation*}
    x^{(n+1)}= H_{\bm{s},\bm{M}}(x^{(n)} + A^{*}(y - Ax^{(n)}))
\end{equation*}

\noindent \textbf{Output:} $\hat{x} = x^{(\overline{n})}$

\end{tcolorbox}

and \textit{CoSaMP in Levels (CoSaMPL)} is defined by

\begin{tcolorbox}

Function $\hat{x}=\mathrm{CoSaMPL}(A,y,\bm{s},\bm{M})$ \\
\noindent \textbf{Inputs:} $A \in \mathbb{C}^{m\times N}$, $y \in \mathbb{C}^m$, local sparsities $\bm{s}$, sparsity levels $\bm{M}$
\\
\noindent \textbf{Initialization:} $x^{(0)} \in \mathbb{C}^N$ (e.g.\ $x^{(0)} = 0$) \\
\noindent \textbf{Iterate:} Until some stopping criterion is met at $n = \overline{n}$, set 
\begin{align*}
    U^{(n+1)} &= \text{supp}(x^{(n)}) \cup L_{2\bm{s},\bm{M}}(A^{*}(y-Ax^{(n)})) \\
    u^{(n+1)} & \in \argmin_{z \in \mathbb{C}^N} \{ \Vert y -Az \Vert_{\ell^{2}} \: : \: \text{supp}(z) \subset U^{(n+1)} \} \\
    x^{(n+1)} &= H_{\bm{s},\bm{M}}(u^{(n+1)})
\end{align*}

\noindent \textbf{Output:} $\hat{x} = x^{(\overline{n})}$

\end{tcolorbox}
We again emphasize here that these differ from the non-levels based versions only in the threshold operator and the index set $L_{2\bm{s},\bm{M}}$, and do not change the main iteration steps at all. Thus much of the analysis and intuition of these algorithms in the sparse case may still be applied, albeit with care.

As already mentioned in Section~\ref{s:previous_work}, IHTL and CoSaMPL are specific instances of model-based IHT and CoSaMP \cite{baraniuk2010model}. Note that projections onto the model sets, which may be hard to compute for general structured sparsity models (e.g., tree-based models), are straightforward in the case of sparsity in levels, being defined just in terms of the operators $H_{\bm{s},\bm{M}}$ and $L_{\bm{s},\bm{M}}$, which are easy to compute.

The final algorithm of interest is greedy Orthogonal Matching Pursuit (OMP), first described in \cite{tropp2007signal}. OMP also admits a generalization to this new setting. A notable benefit of OMP in contrast to IHT and CoSaMP is that the algorithm terminates after a fixed number $s$ of iterations given by the total sparsity. Furthermore, the intermediate least squares problems never exceed size $m \times s$. These together can, in certain cases (e.g., when the target sparsity $s$ is very small), save significant computational time. Generalizing OMP to the levels setting is motivated by these desirable features. However, the usual operation of greedy index selection becomes more subtle, as it is not immediately obvious how one should select the ``best'' indices: there is a choice of whether to proceed level by level in parallel, or select indices in a sequential fashion. In the following formulation, we propose a greedy index selection over all levels not yet saturated in the approximation. This is the sequential approach mentioned above, and seeks to reduce the approximation error at each step as much as possible. This formulation of OMPL performs well numerically in our experiments, but is not examined within this work from a theoretical perspective. 

\newpage

\begin{center}
\begin{tcolorbox}
Function $\hat{x}=\mathrm{OMPL}(A,y,\bm{s},\bm{M})$ \\
\noindent \textbf{Inputs:} $A \in \mathbb{C}^{m\times N}$, $y \in \mathbb{C}^m$, local sparsities $\bm{s}$, sparsity levels $\bm{M}$ (with $M_r = N$)
\\
\noindent \textbf{Initialization:} Choose initial $x^{(0)}$, and set $S^{(0)} = \emptyset$, $\bm{s}^{(0)} = \bm{0}$ and  $\mathcal{L} = \emptyset$.
\\
\noindent \textbf{Iterate:} For each $k = 1, \ldots ,s$, set
         \begin{equation*}
             j_{k} \in \argmax_{\substack{j=1, \ldots, N \\ j \notin \{ M_{p-1}+1, \ldots M_{p} \}, \, \forall p \in \mathcal{L}}} \vert (A^{*}(y- A x^{(k-1)}))_{j} \vert
         \end{equation*}
      and denote $l$ as the level such that $j_{k} \in \{ M_{l-1}+1, \ldots M_{l} \}.$\\
      
        Update $\bm{s}^{(k)} = \bm{s}^{(k-1)} + \bm{e}_{l}$, where $\bm{e}_l$ is the $l$-th standard unit vector\\
        Update $S^{(k)} = S^{(k-1)} \cup  \{ j_{k} \}$ \\
        \qquad If $\bm{s}^{(k)}_{l} = \bm{s}_{l}$, then update the set of saturated levels $\mathcal{L} = \mathcal{L} \cup \{ l \}$ \\ \\
        Set $x^{(k)} \in \displaystyle\argmin_{z \in \mathbb{C}^N : \, \text{supp}(z) \subseteq S^{(k)}} \left\Vert y - A z \right\Vert_{\ell^{2}} $
     
\noindent \textbf{Output:} $\hat{x} = x^{(s)}$
\end{tcolorbox}
\end{center}
\section{Main results\label{sec:MainResults}}

We are now in a position to present our two main theorems. They state that the RIPL of suitable order is sufficient to guarantee stable and robust uniform recovery for the IHTL and CoSaMPL algorithms in the sparse in levels case. The proofs of these results can be found in Section~\ref{sec:proofs}. Before stating these results, we recall that the recovery error of IHTL and CoSaMPL is compared to the target accuracy achieved by the best $(\bm{s},\bm{M})$-term approximation error with respect to the weighted $\ell^1$-norm, where the weights $w$ are assumed to be constant on each level as in \eqref{constantweights}. Moreover, we define the following two key quantities:
\begin{equation}
    \zeta = \sum_{i=1}^{r} (w^{(i)})^{2}s_{i}, \qquad \xi = \min_{i= 1,\ldots,r} (w^{(i)})^{2}s_{i}.\label{zetaxidefinition}
\end{equation}
It is worth stressing that the weights $w$ are not employed in the IHTL and CoSaMPL algorithms, but are only used to prove the corresponding recovery guarantees. 

We state our first result, concerning IHTL.
\begin{theorem}\label{thm:IHTFinalTheorem}
Suppose that the $(6\bm{s},\bm{M})$-th RICL constant of $A \in \mathbb{C}^{m \times N}$ satisfies $\delta_{6\bm{s},\bm{M}} < \frac{1}{\sqrt{3}}$, and let $w \in \mathbb{R}^N$, with $w>0$, be a set of weights constant in each level, i.e.\ as in (\ref{constantweights}). Then, for all $x \in \mathbb{C}^{N}$ and $e \in \mathbb{C}^{m}$, the sequence $(x^{(n)})_{n \geq 0}$ defined by $\mathrm{IHTL}(A,y,2\bm{s},\bm{M})$ with $y=Ax+e$ and $x^{(0)} = 0$ satisfies, for any $n \geq 0$,
\begin{align*}
 \Vert x - x^{(n)} \Vert_{\ell^{1}_{w}} &\leq C\frac{\sqrt{\zeta}}{\sqrt{\xi}}\sigma_{\bm{s,M}}(x)_{\ell_{w}^{1}} + D \sqrt{\zeta} \Vert e \Vert_{\ell^{2}} + 2 \sqrt{\zeta} \rho^{n} \Vert x \Vert_{\ell^{2}}, \\
\Vert x - x^{(n)} \Vert_{\ell^{2}} &\leq \frac{E}{\sqrt{\xi}}\sigma_{\bm{s,\bm{M}}}(x)_{\ell^{1}_{w}} + F \Vert e \Vert_{\ell^{2}} + \rho^{n} \Vert x \Vert_{\ell^{2}},
\end{align*}
where $\rho = \sqrt{3} \delta_{6\bm{s},\bm{M}}< 1$ and $C,D,E,F >0$ only depend on $\delta_{6\bm{s},\bm{M}}$, and $\zeta,\xi$ are as in (\ref{zetaxidefinition}).

\end{theorem}

An analogous result holds for CoSaMPL.
\begin{theorem} \label{thm:CoSamPlMainTheorem} Suppose that the $(8\bm{s},\bm{M})$-th RICL constant of $A \in \mathbb{C}^{m \times N}$ satisfies
\begin{equation*}
    \delta_{8\bm{s},\bm{M}} < \tfrac{\sqrt{\frac{11}{3}} - 1}{4} \approx 0.478,
\end{equation*}
and let $w \in \mathbb{R}^N$, with $w>0$, be a set of weights constant in each level as in Theorem~\ref{thm:IHTFinalTheorem}. Then, for all $x \in \mathbb{C}^{N}$ and $e \in \mathbb{C}^{m}$ the sequence $(x^{(n)})_{n\geq 0}$ constructed by $\mathrm{CoSaMPL}(A,y,2\bm{s},\bm{M})$ with $y=Ax+e$ and $x^{(0)} = 0$, satisfies for any $n \geq 0$,
\begin{align*}
    \Vert x - x^{(n)} \Vert_{\ell^{1}_{w}} 
    &\leq C\frac{\sqrt{\zeta}}{\sqrt{\xi}}\sigma_{\bm{s,M}}(x)_{\ell_{w}^{1}} 
    + D \sqrt{\zeta} \Vert e \Vert_{\ell^{2}} 
    + 2 \sqrt{\zeta} \rho^{n} \Vert x \Vert_{\ell^{2}}, \\
\Vert x - x^{(n)} \Vert_{\ell^{2}} &\leq \frac{E}{\sqrt{\xi}}\sigma_{\bm{s,\bm{M}}}(x)_{\ell^{1}_{w}} + F\Vert e \Vert_{\ell^{2}} + \rho^{n} \Vert x \Vert_{\ell^{2}},
\end{align*}
where $\zeta$ and $\xi$ are as in Theorem~\ref{thm:IHTFinalTheorem}
and $\rho = \sqrt{ \frac{2\delta_{8\bm{s},\bm{M}}^{2}(1 +3\delta_{8\bm{s},\bm{M}}^{2})}{1-\delta_{8\bm{s}}^{2}}} $ and $C,D,E,F >0$ only depend on $\delta_{8\bm{s},\bm{M}}$.

\end{theorem}

Note that several recovery guarantees for the model-based CoSaMP algorithm were given in \cite{baraniuk2010model}. However, due to the more general setting, these involve additional conditions on the measurement matrix $A$ and an explicit compressibility assumption on the vector $x$. Theorem \ref{thm:CoSamPlMainTheorem}, being specific to the sparsity in levels model, yields a simpler recovery guarantee that requires only the RIPL for $A$ and no assumptions on $x$. It also provides bounds in any weighted $\ell^1_{w}$-norm, with weights as in (\ref{constantweights}).

\subsection{Discussion}

\label{s:discussion}

We now make several remarks on Theorems~\ref{thm:IHTFinalTheorem} and \ref{thm:CoSamPlMainTheorem}. First, we analyze some interesting special cases based on different choices of weights. Then, we compare the recovery guarantees for IHTL and CoSaMPL with those of a decoder based on weighted $\ell^1$ minimization.

If we have one level, i.e. $r = 1$, $\bm{M} = (1,N)$ and weights $w^{(1)} = \ldots = w^{(r)} = 1$, we recover \textit{exactly} the result from the sparse case, with identical assumptions on the RIP constant $\delta_{s} = \delta_{\bm{s},\bm{M}}$, and resulting in the same values of $\rho$. 

Now suppose we have arbitrary numbers of levels and local sparsities, but constant weights $w^{(1)} = \ldots = w^{(r)} = 1$. Then, for sufficiently large $n$, our results yield error bounds of the form 
    \begin{align*}
         \Vert x - x^{(n)} \Vert_{\ell^{1}} &\lesssim \frac{\sqrt{s}}{\sqrt{\min_{i} s_{i}}}\sigma_{\bm{s,M}}(x)_{\ell^{1}} + \sqrt{s} \Vert e \Vert_{\ell^{2}}, \\
         \Vert x - x^{(n)} \Vert_{\ell^{2}} &\lesssim \frac{1}{\sqrt{\min_{i} s_{i}}}\sigma_{\bm{s,\bm{M}}}(x)_{\ell^{1}} + \Vert e \Vert_{\ell^{2}},
    \end{align*}
where $s$ is the total sparsity. This leads to large factors multiplying the best $(\bm{s},\bm{M})$-ap\-prox\-i\-ma\-tion error if the minimum local sparsity $\min_i s_{i}$ is small in comparison to the maximum local sparsity $\max_i s_i$. These factors have instead moderate size if $0<\min_{i} s_{i}  \approx \max_{i} s_{i}$.

Accordingly, a good choice of weights is realized by making $\zeta/\xi$ order one, which results in the error bound in the $\ell^{1}_{w}$-norm being optimal up to a constant. 

Finally, if the weights are chosen as in \eqref{optimalweights}, we obtain $\zeta = rs$ and $\xi = s$. This choice yields error bounds where the constant factors only depend on the number of levels $r$ and the total sparsity $s$. Namely, for $n$ large enough, we have a dependence scaling with the number of levels, approximately
\begin{align*}
    \Vert x - x^{(n)} \Vert_{\ell^{1}_{w}} & \lesssim \sqrt{r} \sigma_{\bm{s,M}}(x)_{\ell^{1}_{w}} + \sqrt{r s} \Vert e \Vert_{\ell^{2}}, \\
    \Vert x - x^{(n)} \Vert_{\ell^{2}} &\lesssim \frac{1}{\sqrt{s}} \sigma_{\bm{s},\bm{M}}(x)_{\ell^{1}_{w}} + \Vert e \Vert_{\ell^{2}}.
\end{align*}

In order to further understand the theoretical estimates obtained for IHTL and CoSaMPL, we compare them with recovery guarantees based on convex optimization via weighted $\ell^1$ minimization. As mentioned, previous work \cite{AAHWalshWavelet} on the sparsity in levels model has focused on the weighted Quadratically Constrained Basis Pursuit (QCBP) decoder
\begin{equation}
\label{QCBP}
\min_{z \in \mathbb{C}^N} \Vert z \Vert_{\ell^1_w}\ \mbox{subject to $\Vert A z  -y \Vert_{\ell^2} \leq \eta$},  
\end{equation}
with weights as in \eqref{constantweights}. The following is an analogous result to Theorems \ref{thm:IHTFinalTheorem} and \ref{thm:CoSamPlMainTheorem} for the weighted QCBP decoder:

\begin{theorem}
\label{thm:simplerQCBPresult}
Suppose that the $(2\bm{s},\bm{M})$-th RICL constant of $A \in \mathbb{C}^{m \times N}$ satisfies
\begin{equation}
\label{QCBPRIPLcond}
 \delta_{2 \bm{s} , \bm{M}} <  \frac{1}{\sqrt{2 \zeta / \xi} + 1},  
\end{equation}
where $\zeta$ and $\xi$ are as in \eqref{zetaxidefinition}, and let $\eta \geq 0$ and $w \in \mathbb{R}^N$, with $w >0$, be a set of weights  as in \eqref{constantweights}. Then, for all $x \in \mathbb{C}^{N}$ and $e \in \mathbb{C}^{m}$ with $\Vert e \Vert_{\ell^2} \leq \eta$ any minimizer of $\hat{x}$ of \eqref{QCBP} with $y = A x + e$ satisfies
\begin{equation}
 \label{QCBPresult}
\begin{split}
    \Vert x - \hat{x} \Vert_{\ell^{1}_{w}} 
    &\leq C \sigma_{\bm{s,M}}(x)_{\ell_{w}^{1}} 
    + D \sqrt{\zeta} \eta,
    \\
\Vert x - \hat{x} \Vert_{\ell^{2}} &\leq \left ( 1 + (\zeta/\xi)^{1/4} \right ) \left ( \frac{E}{\sqrt{\xi}}\sigma_{\bm{s,\bm{M}}}(x)_{\ell^{1}_{w}} + F \eta \right ),
\end{split}
\end{equation}
where $C,D,E,F$ only depend on $\delta_{2\bm{s},\bm{M}}$ only. 
\end{theorem}

Note that the analysis of QCBP via the RICL was previously addressed in \cite[Thm.\ 3.5]{AAHWalshWavelet}, but with the order of the RIPL was given in terms of the weights and the quantity $\zeta$ (this approach was better suited to the purposes of \cite{AAHWalshWavelet}, which focused on Walsh--Hadamard sampling). Theorem \ref{thm:simplerQCBPresult} is a more direct generalization of a standard result for stable and robust sparse vectors with QCBP. Indeed, when $r = 1$ and $w^{(1)} = \ldots = w^{(r)} = 1$ it reduces to
\begin{align*}
    \Vert x - \hat{x} \Vert_{\ell^{1}} 
    &\leq C\sigma_{s}(x)_{\ell^{1}} 
    + D \sqrt{s} \eta, 
    \\
\Vert x - \hat{x} \Vert_{\ell^{2}} &\leq \frac{E}{\sqrt{s}}\sigma_{s}(x)_{\ell^{1}} + F \eta,
\end{align*}
under the condition that $A$ has the RIP of order $2s$ with constant $\delta_{2s} < \sqrt{2}-1$, which is a classical result due to Cand\`es \cite{candes2008restricted}.
Since, to the best of authors' knowledge, it has not appeared previously, we give a short proof of Theorem \ref{thm:simplerQCBPresult} in Section \ref{ss:simplerQCBPresultproof}. 

Several remarks are in order. First, associating $\Vert e \Vert_{\ell^2}$ with $\eta$, we observe that the error bounds in \eqref{QCBPresult} are similar to those in Theorems~\ref{thm:IHTFinalTheorem} and \ref{thm:CoSamPlMainTheorem} (for $n$ large enough), up to the scaling with respect to $\zeta$ and $\xi$. In particular, for the $\ell^1_{w}$-norm error, QCBP has a better dependence on $\sigma_{\bm{s},\bm{M}}(x)_{\ell^1_{w}}$ by a factor of $\sqrt{\zeta/\xi}$. 
 The dependence on the noise is the same. 
We mention in passing that it is not clear whether or not the scaling $\sqrt{\zeta/\xi}$ for IHTL and CoSaMPL is necessary, or simply an artefact of the proof.
 Conversely, the $\ell^2$-norm error bound is better for the IHTL and CoSaMPL decoders, by a factor of $(\zeta/\xi)^{1/4}$. Overall, the factor $\zeta / \xi$ is seemingly ubiquitous. In particular, for all decoders, choosing the weights as in \eqref{optimalweights} acts to minimize this factor, thus giving the best recovery guarantees.

Second, the condition on the RICL depends on $\zeta$ and $\xi$ in the case of QCBP, but is independent of them in the case of IHTL and CoSaMPL. 
In particular, in the unweighted case the condition \eqref{QCBPRIPLcond} for QCBP becomes
\begin{equation}
\label{RIPLl1sparselevels}
\delta_{2 \bm{s},\bm{M}} < \frac{1}{\sqrt{2 s / \min_{i} \{ s_i \} } + 1},    
\end{equation}
which depends on the ratio of the total sparsity $s$ and the minimal local sparsity. Conversely, if the weights are chosen as in \eqref{optimalweights}, the condition \eqref{QCBPRIPLcond} becomes
\begin{equation}
 \label{measurementCondDelta}
\delta_{2 \bm{s},\bm{M}} < \frac{1}{\sqrt{2r} +1} .
\end{equation}
Observe from \eqref{gaussmeasurmentCond} that the number of measurements that guarantees an RIP generally scales like $\delta^{-2}$. Combining this observation with condition \eqref{measurementCondDelta} suggests that $m$ should scale linearly in $r$ for QCBP to ensure stable and robust recovery, whereas for IHTL and CoSaMPL the corresponding condition on $m$ would be independent of $r$. Note that while $r$ may be small in some applications, in others it may grow with $N$. For example, $r = \mathcal{O}(\log(N))$ when the levels delineate wavelet scales, as in the setups of \cite{AAHWalshWavelet, BastounisHansen,LiAdcockRIP}, thus making the QCBP measurement condition effectively worse by one log factor. 

\begin{remark}
It is natural to ask whether or not the condition \eqref{QCBPRIPLcond} is sharp. Certainly, the constant factor $\sqrt{2}$ can likely be improved, much as how the constant in the classical RIP condition $\delta_{2s} < \sqrt{2} - 1$, to which \eqref{RIPLl1sparselevels} and \eqref{measurementCondDelta} both reduce in the case of $r = 1$ levels, can be improved. Indeed, the optimal condition is known to be $\delta_{2s} < 1/\sqrt{2}$ \cite{CaiZhangRIP2}. On the other hand, whether or not the dependence on $\zeta$ and $\xi$ can be improved is unknown. On the other hand, the work \cite{BastounisHansen} analyzed sufficient conditions for recovery of QCBP. It was shown that the condition \eqref{RIPLl1sparselevels} for unweighted QCBP essentially cannot be improved, except possibly for the constants $2$ and $1$. Further, it is readily seen from the proof of \cite[Thm.\ 4.6]{BastounisHansen} that for weighted QCBP with weights as in \eqref{optimalweights} (or, in general, any weights for which $w^{(k)} = g(s^{(k)})$ for some positive function $g$), the scaling $1/\sqrt{r}$ is necessary. This demonstrates that the aforementioned distinction between QCBP and IHTL and CoSaMPL is fundamental.
\end{remark}



It is worth mentioning several other differences between QCBP and the proposed decoders. First, in QCBP the weights are a part of the decoder itself. Whereas for IHTL and CoSaMPL they appear solely in the theoretical analysis.
Second, the guarantees for QCBP rely on an \textit{a priori} control of the noise level. The is typical of QCBP approaches, but such a bound on the noise is unlikely to be known in many applications~\cite{brugiapaglia2018robustness}. The results for CoSaMPL and IHTL do not require any such bound on the noise, as is evident from the discussion above. We do note however, that the recovery results for QCBP can be extended to the weighted square-root LASSO decoder without assumptions on noise~\cite{BASBMKRCSwavelet}.


\section{Numerics} \label{sec:numerics}

Finally, we include some numerics to support the claim that OMP also generalizes well to this new setting. The experiments performed are analogous to those in \cite{adcock2019iterative}, which gives numerical results for IHTL and CoSAMPL. 

All numerical experiments share the following setup. For each fixed total sparsity $s$ and number of measurements $m$, we generate an $(\bm{s},\bm{M})$-sparse in levels random vector $x$ of length $N$. This is implemented by generating a random vector $v \in \mathbb{R}^N$ with independent entries distributed according to a centered Gaussian distribution with unit variance and by setting $x = H_{\bm{s},\bm{M}}(v)$. The local sparsity pattern $\bm{s}$ depends on the experiment, as outlined below. Then we compute an approximation $\hat{x}$ to the vector $x$ using a measurement matrix $A$ that is a Gaussian random matrix (independent, normally distributed entries with mean zero and variance $1/\sqrt{m}$), and record the relative error $\Vert x - \hat{x} \Vert_{\ell^2} / \Vert x \Vert_{\ell^2}$. Over 100 trials, we compute the success probability with the success criterion that the relative error be less than $10^{-2}$. For IHTL and CoSaMPL, we have the additional stopping criterion that the algorithms terminate either when $\|x^{(n+1)}-x^{(n)}\|_{\ell^2}/\|x^{(n+1)}\|_{\ell^2}$ is less than a tolerance $10^{-4}$, or if the algorithms exceeds 1000 iterations. For OMPL, we simply run $s$ iterations. Moreover, we operate the following normalizations in order to improve the practical performance of the algorithms considered. For OMP and OMPL, we normalize the columns of $A$ with respect to the 2-norm. For IHT and IHTL, we rescale $A$ so that it has unit spectral norm (i.e., $\|A\|_2 = 1$), as suggested in \cite{BlumensathDavies2010}.
    
With the aim of boosting the numerical performance of IHT and IHTL, we also consider Normalized IHT (NIHT) (see \cite{BlumensathDavies2010}) and its levels based version NIHTL obtained by replacing the hard thresholding operator with its in-levels version. We implement NIHTL as summarized in \cite[Section III]{BlumensathDavies2010} using hyperparameters $c = 0.1$ and $k = 1.1 / (1-c)$, defining the stepsize update rule (see \cite{BlumensathDavies2010} for further details).

In our implementation of CoSaMPL and OMPL we use the Matlab backslash operator to compute least-squares projections. This is possible since we only consider problems of small size in this paper. For large scale problems, the least-squares projection step can become a considerable computational bottleneck and it should be replaced by an approximate projection computed via a few iterations of an iterative algorithm such as Richardson's or the conjugate gradient method (see \cite{NeedellTropp2008}). We also note that, whereas the number of iterations of OMPL is by construction larger than or equal to the sparsity $s$, CoSaMPL does not have this limitation. Thus, the latter is a considerably more efficient algorithm than the former for large values of $s$.

The first experiment, whose results are shown in Figure~\ref{fig:4levelunif}, compares the performance of IHTL, NIHTL, CoSaMPL, and OMPL with different input sparsity levels. Given a fixed total sparsity $s$, we approximate a vector that is $(\bm{s},\bm{M})$-sparse with $\bm{M} = (N/4,N/2,3N/4,N)$ and either $\bm{s}_1 = (3 s/8 , s/8 , 3 s /8 , s/8)$ or $\bm{s}_2 = (s/2 , 0 , s/2 , 0)$. We let $N = 128$ and consider $s = 16$ and $s = 32$. We then run IHTL,CoSaMPL, and OMPL, with $1,2$ or $4$ levels each. The results show that more levels, closer to the true sparsity in levels structure of the underlying solution, result in better recovery, regardless of the algorithm used.
\newcommand{\figsize}{4.5cm}
\begin{figure}[ht!]
\begin{center}
\begin{tabular}{cc} 
\includegraphics[width=\figsize]{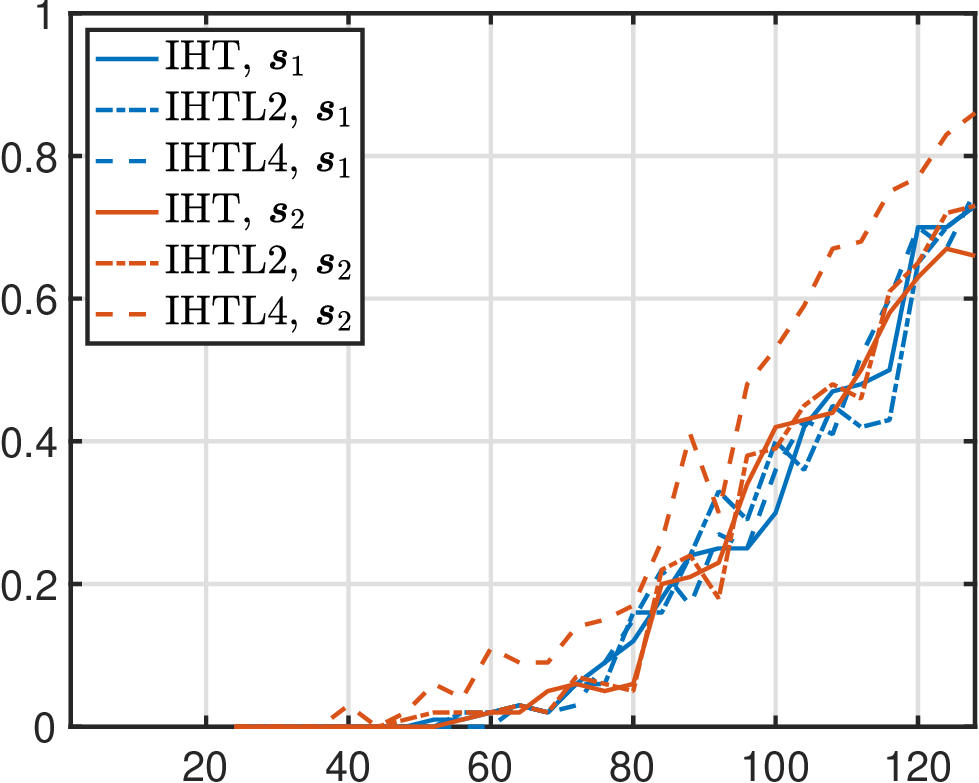} &
\includegraphics[width=\figsize]{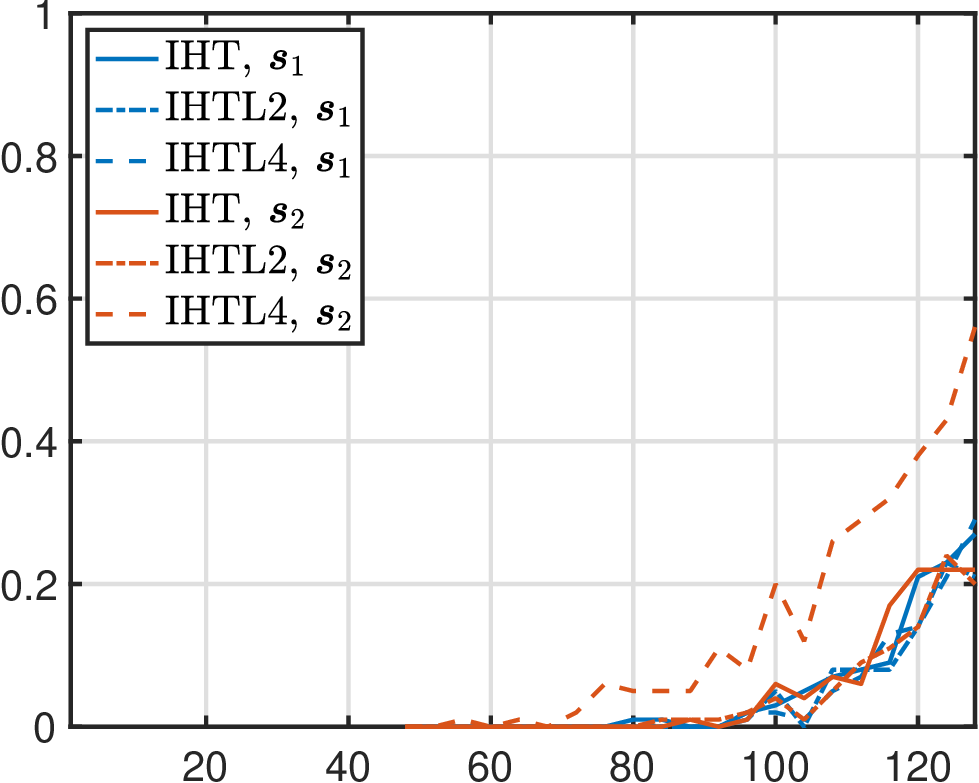} \\
\includegraphics[width=\figsize]{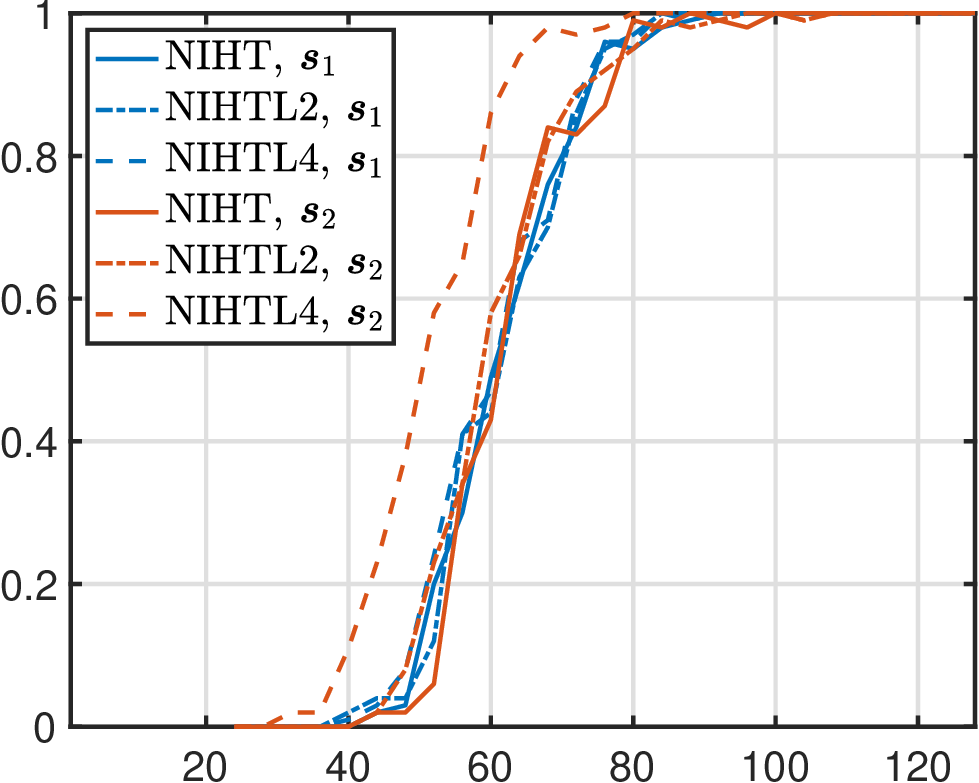} &
\includegraphics[width=\figsize]{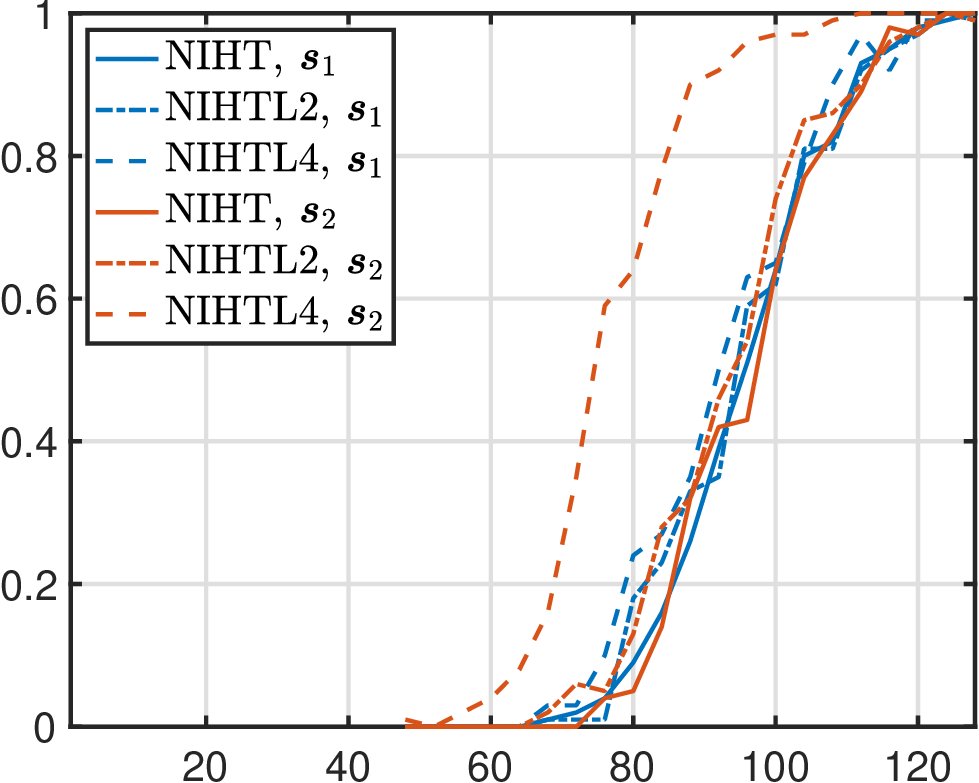} \\
\includegraphics[width=\figsize]{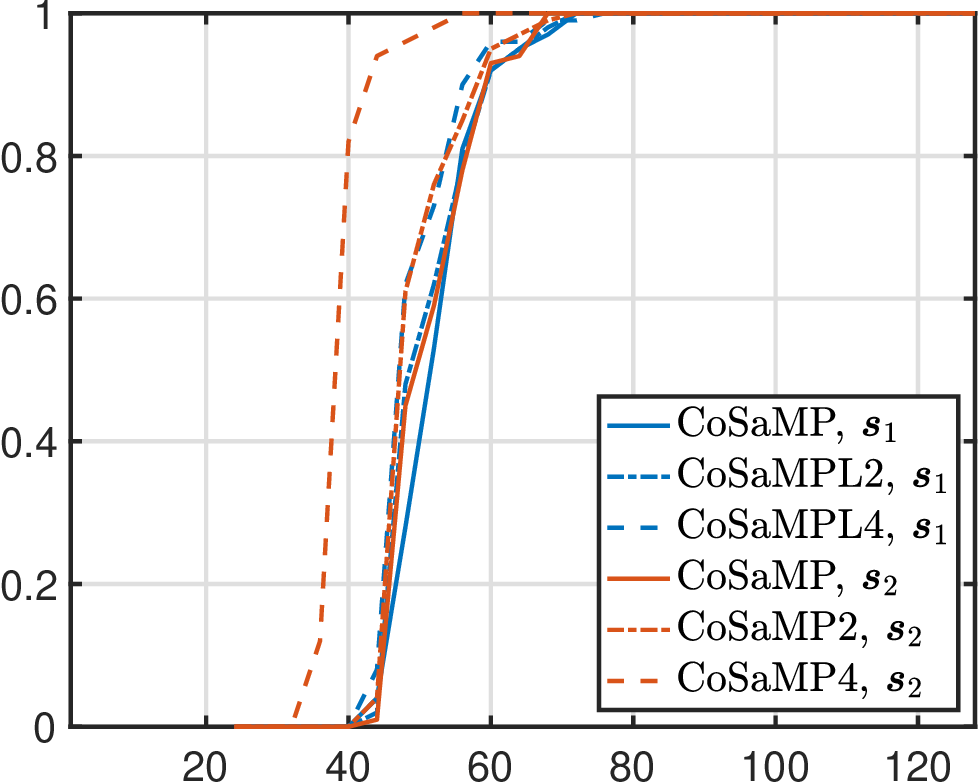} &
\includegraphics[width=\figsize]{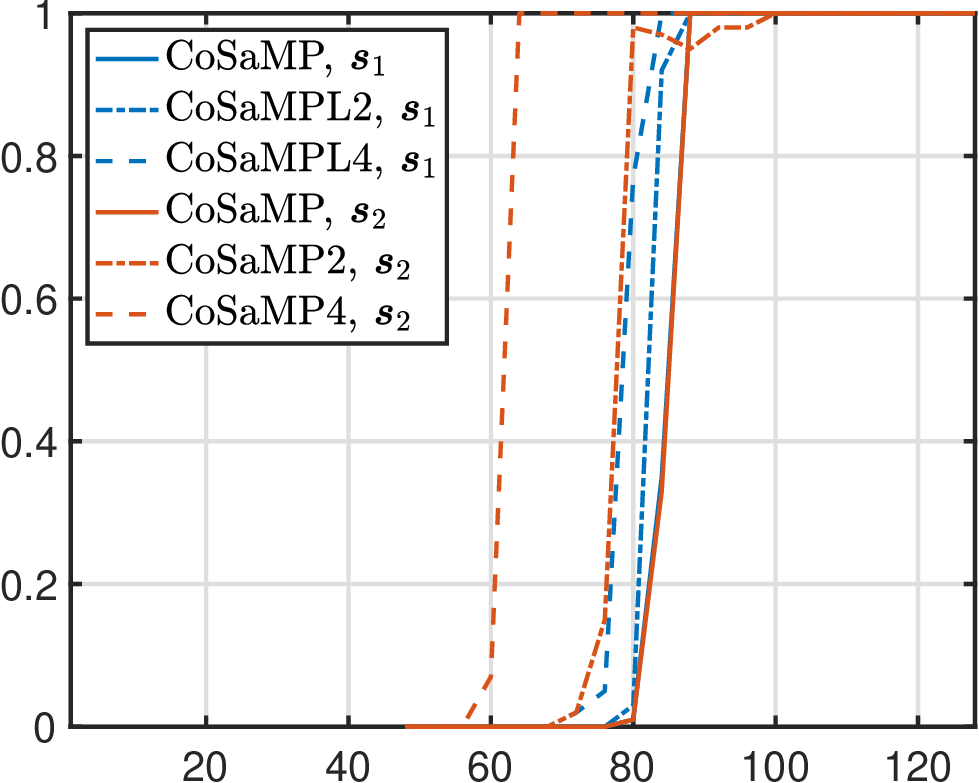} \\
\includegraphics[width=\figsize]{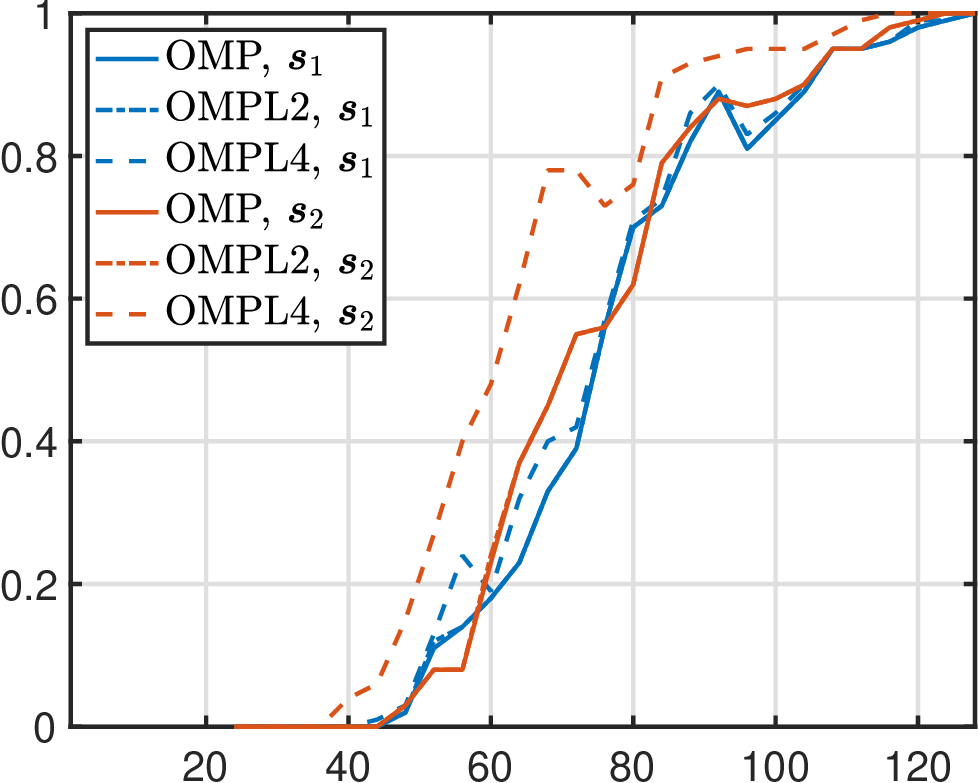} &
\includegraphics[width=\figsize]{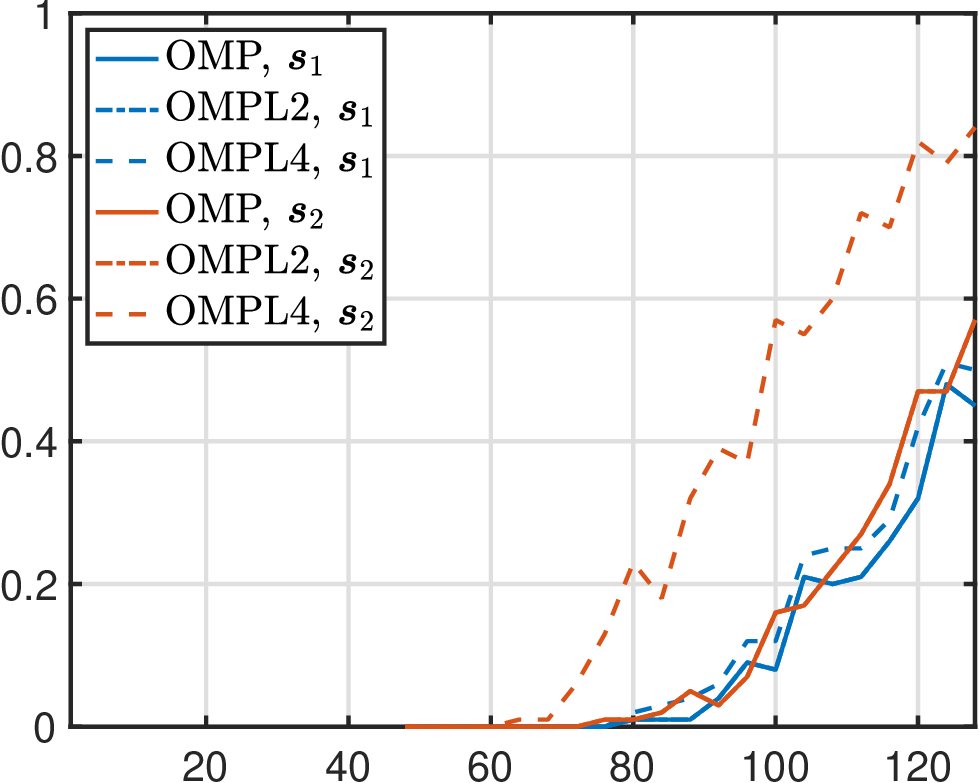} \\
$s = 16$ & $s = 32$
\end{tabular}
\end{center}
\caption[example] 
{ \label{fig:4levelunif} 
Horizontal phase transition line showing success probability versus $m$ for various fixed total sparsities $s$.  Four-level sparsity with $\bm{M} = (N/4,N/2,3N/4,N)$.  The local sparsities are $\bm{s}_1 = (3 s/8 , s/8 , 3 s /8 , s/8)$ and $\bm{s}_2 = (s/2 , 0 , s/2 , 0)$. In the levels case we consider two-level algorithms based on $\bm{M} = (N/4,N)$ and $\bm{s} = (s/2,s/2)$ and four-level algorithms based on $\bm{M} = (N/4,N/2,3N/4,N)$ and $\bm{s} = \bm{s}_1$ or $\bm{s} = \bm{s}_2$. Row one contains IHT and IHTL, row two NIHT and NIHTL, row three CoSaMP and CoSaMPL, and the final row OMP and OMPL. 
}
\end{figure}

The second experiment is similar to the first one. However, we consider a larger number of levels and a higher dimension $N = 512$. Specifically, we generate $(\bm{s},\bm{M})$-sparse vectors with levels $\bm{M}=(N/8, N/4, 3N/8, N/2, 5N/8, 3N/4, 7N/8, N)$ and local sparsities $\bm{s} = (s/4, 0, s/4, 0, s/4, 0, s/4, 0)$, with $s = 128$ and $s = 256$. We compare IHT, NIHT, CoSaMP, and OMP with three levels-based versions of them. The first two of them are four-level algorithms with uniform and nonuniform level splitting, and the third one is an eight-level algorithm. The results are shown in Figure~\ref{fig:8_levels} and, similarly to Figure 1, the levels-based strategies that are closest to the true structure of the underlying solution consistently outperform all the other strategies. In addition, we  observe that the four-level strategies, where local sparsities are constant across all levels, usually provide little or no benefit with respect to the standard versions of the algorithms, regardless of the type of levels considered (i.e., of uniform or nonuniform sizes). 

\begin{figure}[ht!]
\centering
\begin{tabular}{ccc}
\includegraphics[width = \figsize]{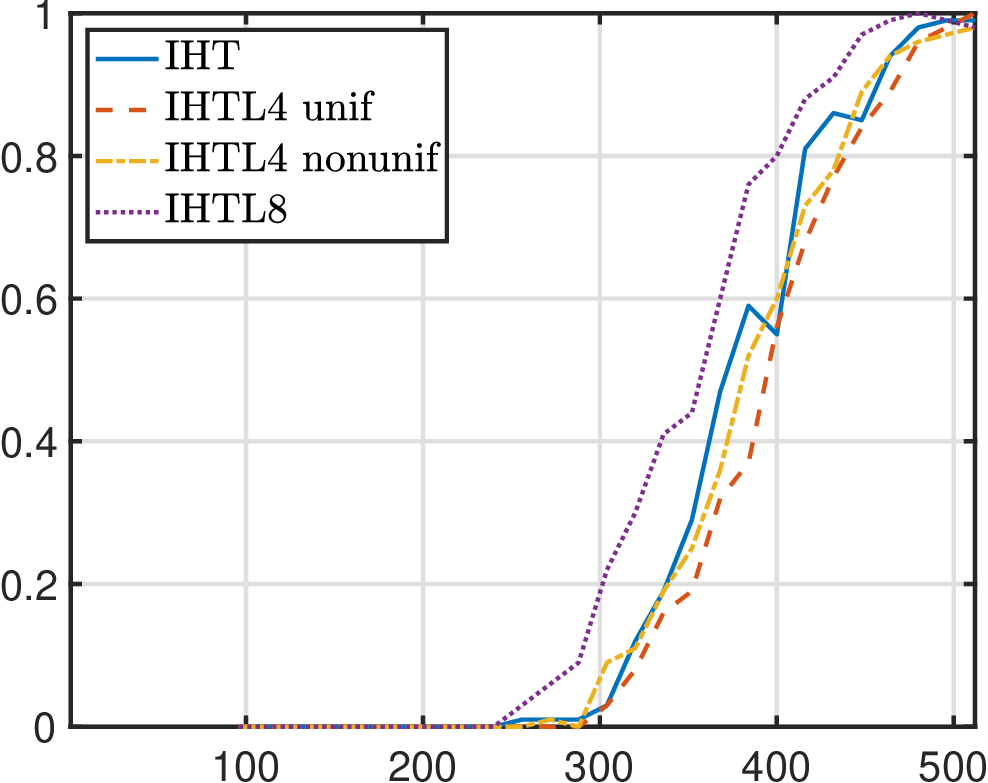}&
\includegraphics[width = \figsize]{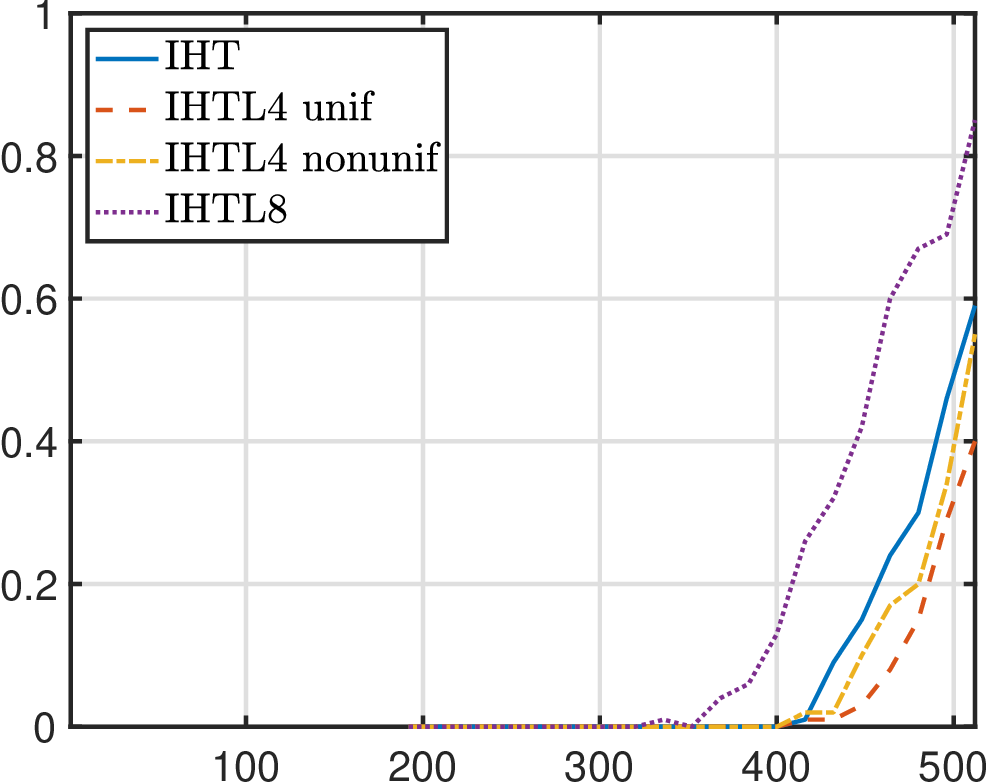}\\
\includegraphics[width = \figsize]{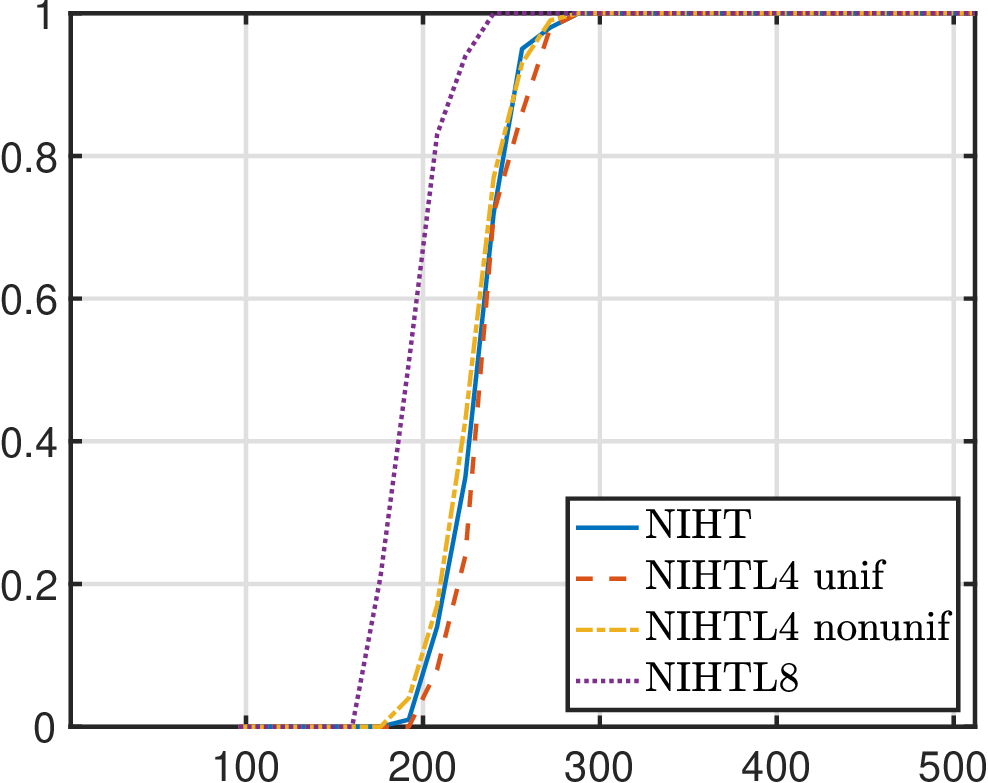}&
\includegraphics[width = \figsize]{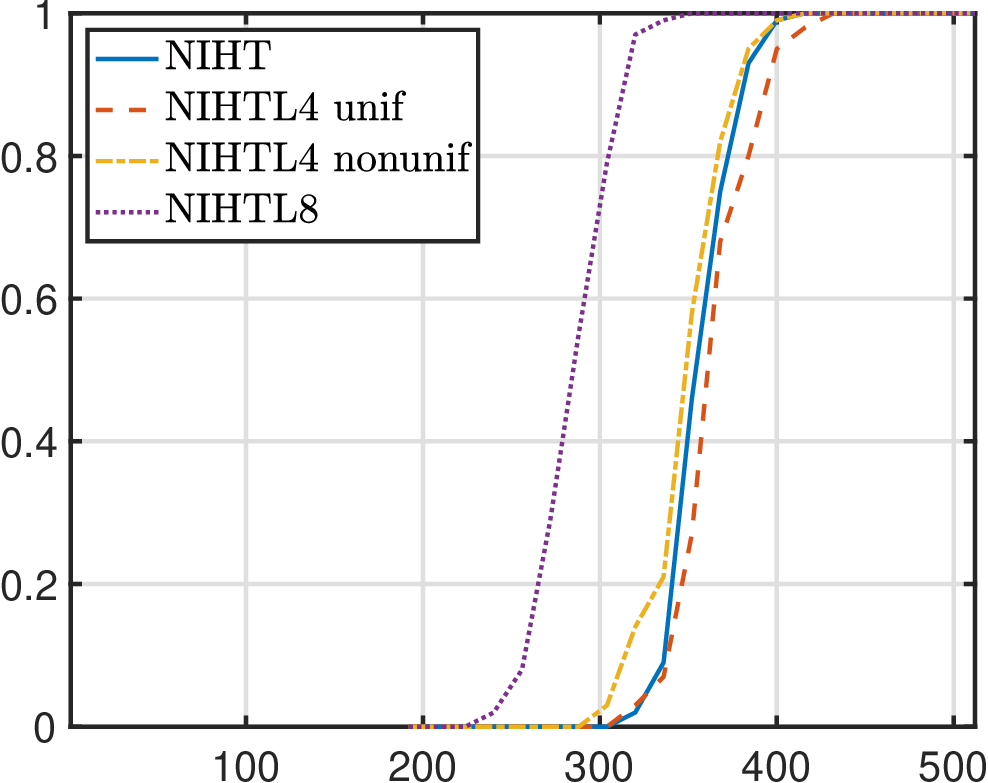}\\
\includegraphics[width = \figsize]{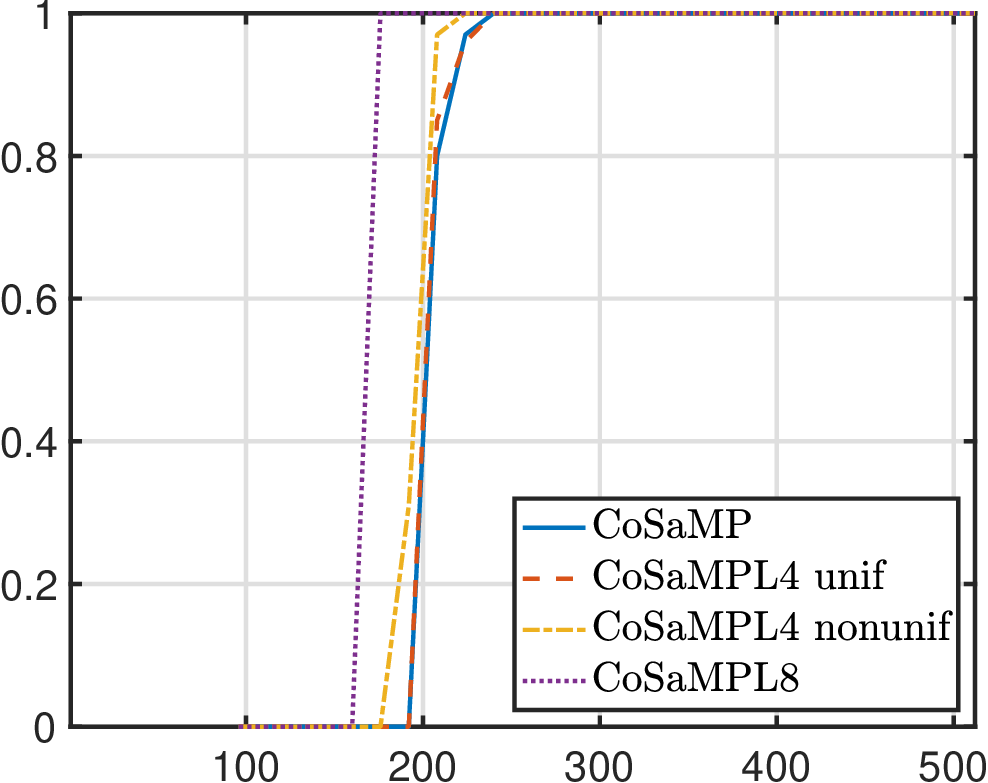}&
\includegraphics[width = \figsize]{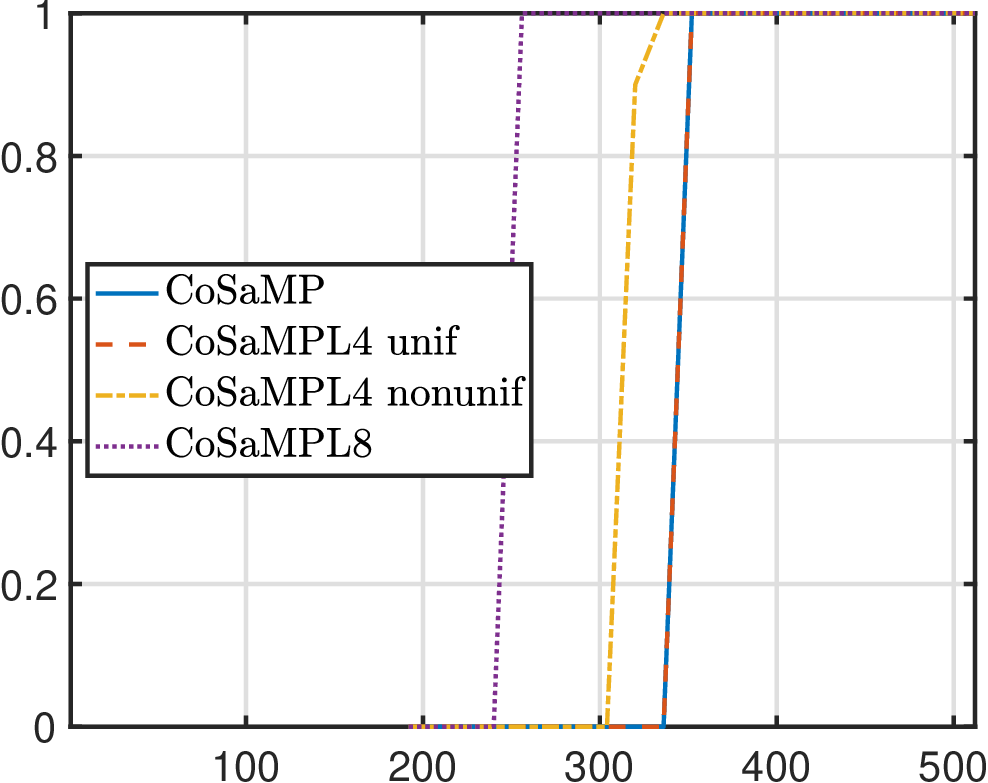}\\
\includegraphics[width = \figsize]{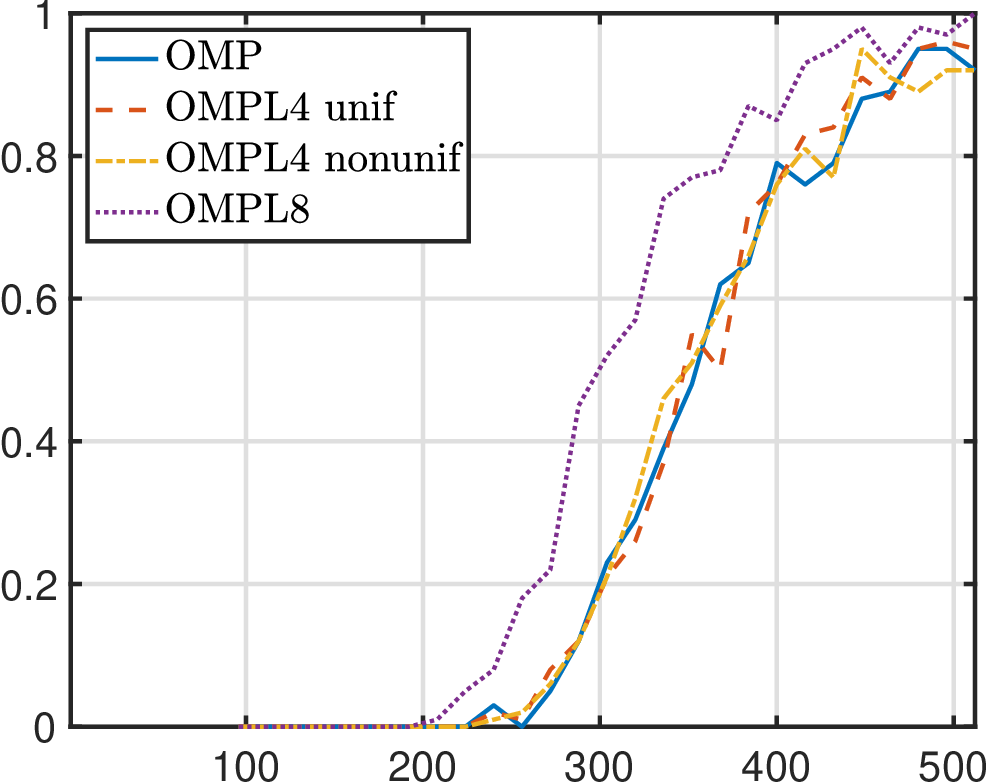}&
\includegraphics[width = \figsize]{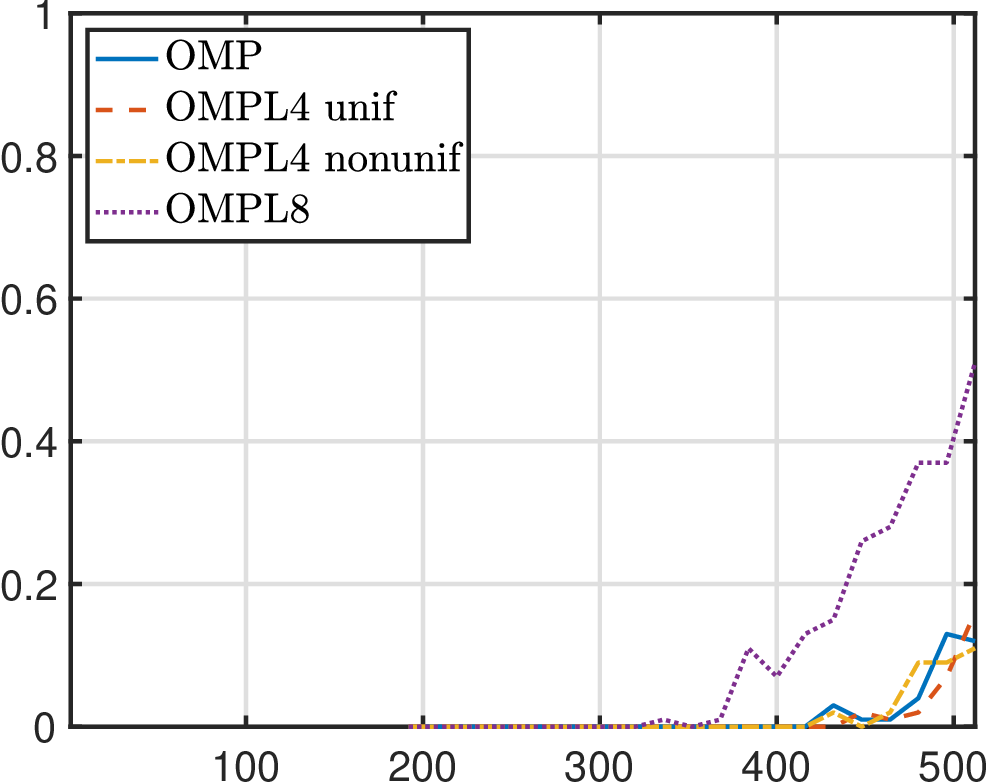}\\
$s = 128$ & $s = 256$
\end{tabular}

\caption{\label{fig:8_levels}Horizontal phase transition line showing success probability versus $m$ for various fixed total sparsities $s$.  Eight-level sparsity with $\bm{M}=(N/8, N/4, 3N/8, N/2, 5N/8, 3N/4, 7N/8, N)$ and local sparsities $\bm{s} = (s/4, 0, s/4, 0, s/4, 0, s/4, 0)$. In the levels case we consider two four-level algorithms based on uniform levels $\bm{M} = (N/4, N/2, 3N/4, N)$ and nonuniform levels $\bm{M} = (N/8 N/2 5N/8 N)$, with local sparsities $\bm{s} = (s/4, s/4, s/4, s/4)$, and eight-level algorithms where $\bm{M}=(N/8, N/4, 3N/8, N/2, 5N/8, 3N/4, 7N/8, N)$ and local sparsities $\bm{s} = (s/4, 0, s/4, 0, s/4, 0, s/4, 0)$. Row one contains IHT and IHTL, row two NIHT and NIHTL, row three CoSaMP and CoSaMPL, and the final row OMP and OMPL.}
\end{figure}

Next in Figure \ref{fig:PTPFull4level} we give phase transition plots for each algorithm, with the local sparsity pattern $\bm{s}=(s/2,0,s/2,0)$ in levels $\bm{M}=(N/4,N/2,3N/4,N)$. Note that this sparsity pattern is only sensible up to $s = N/2$, as thereafter we are fully saturated in the first and the third levels. Thus the experiments below only plot to a maximum of $m=s=N/2$.  We compare the standard sparse decoders of IHT, CoSaMP, and OMP against the levels-based versions, and see uniform improvement by moving to the levels setting.
\begin{figure} [ht]
\begin{center}
\begin{tabular}{ccc} 
\includegraphics[width=\figsize]{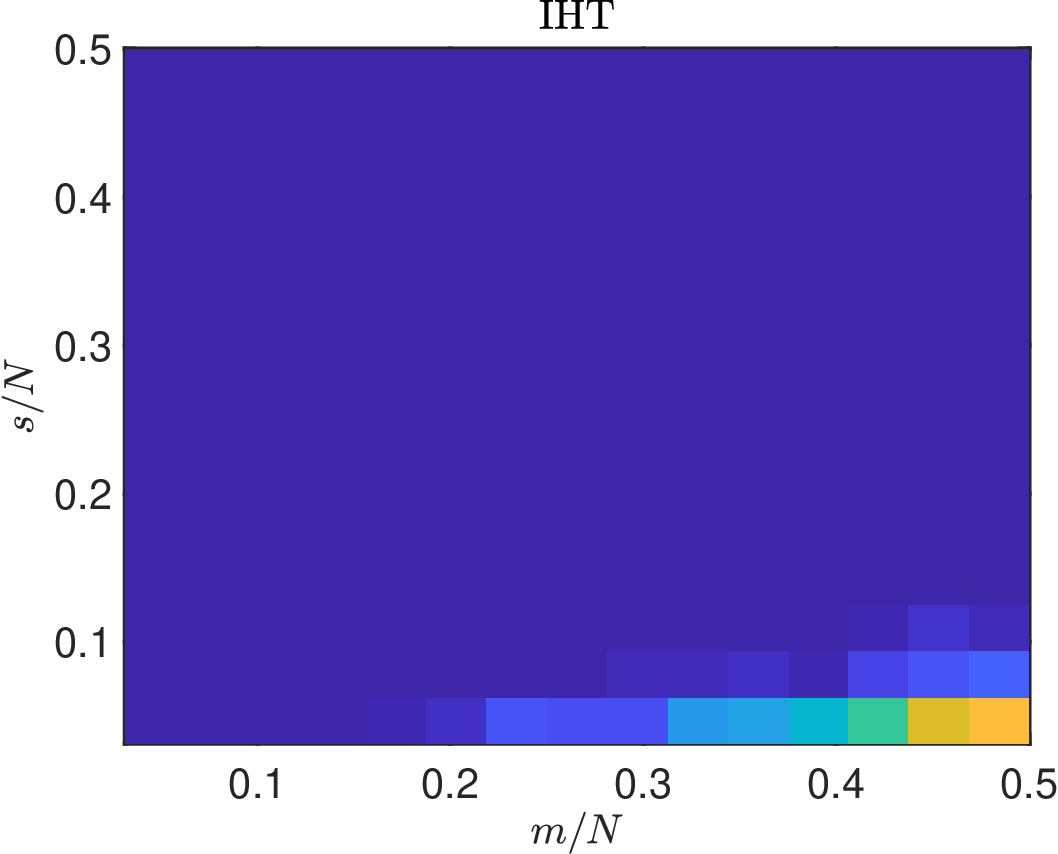} &
\includegraphics[width=\figsize]{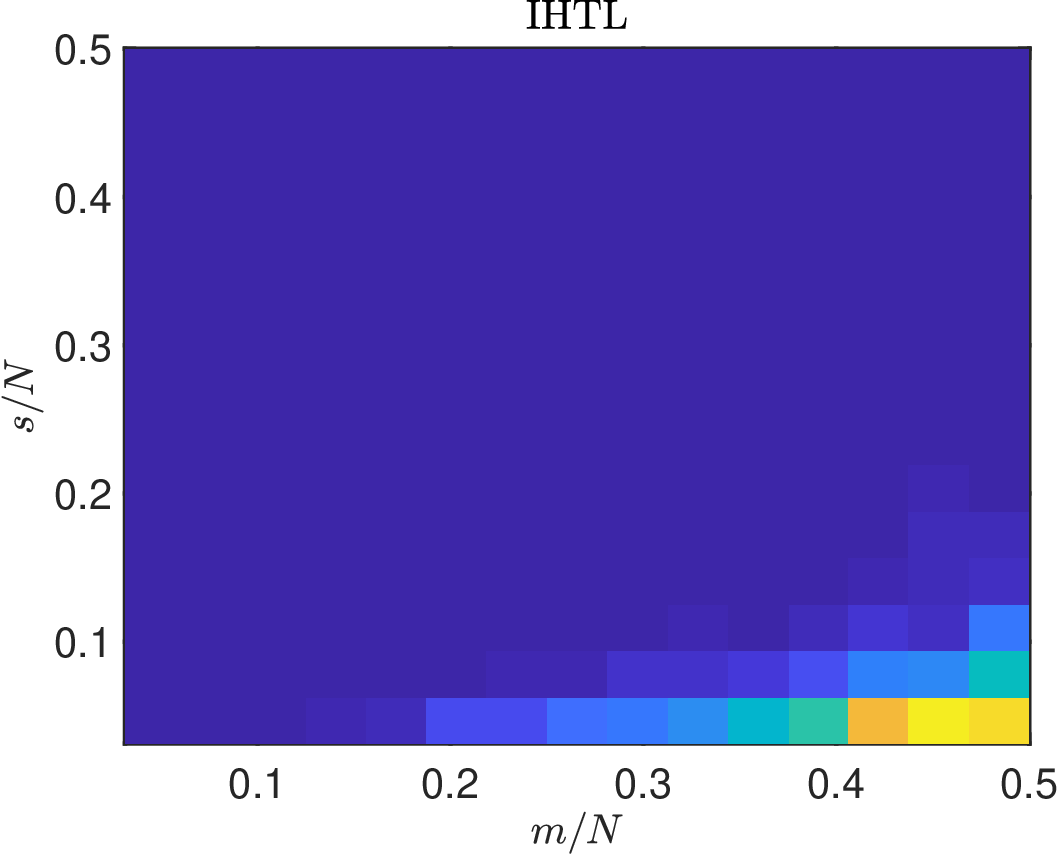} \\
\includegraphics[width=\figsize]{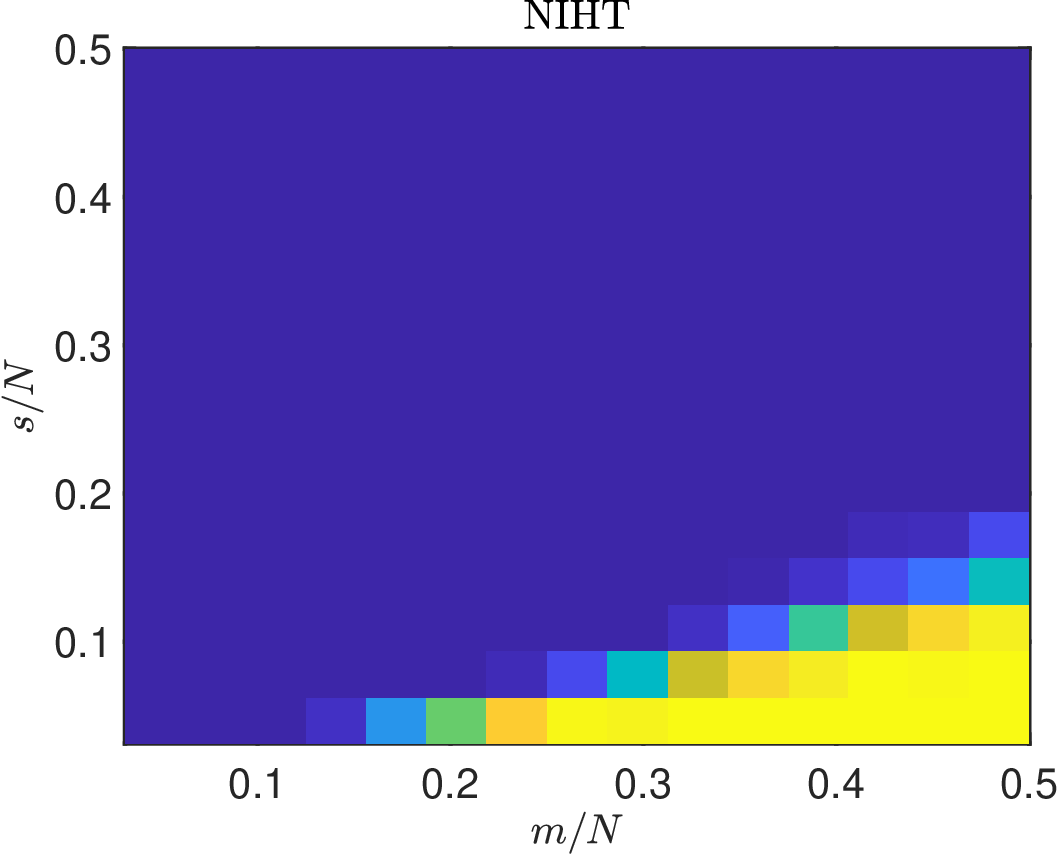} &
\includegraphics[width=\figsize]{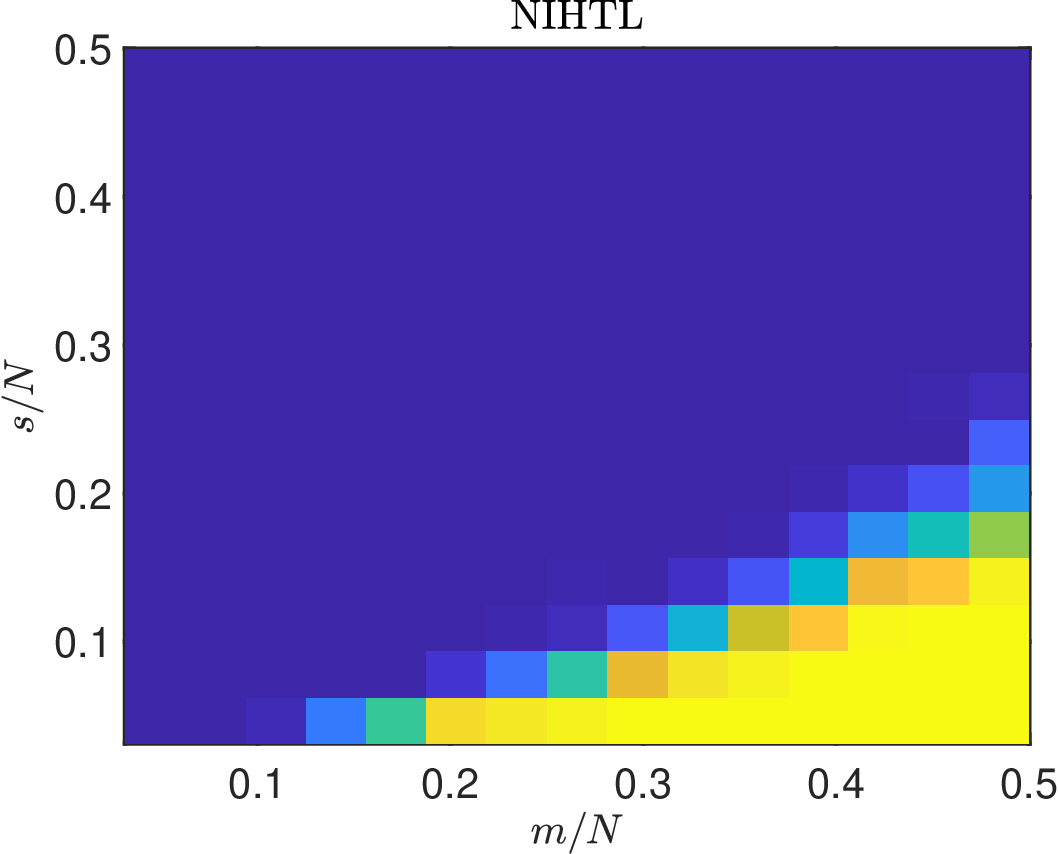} \\
\includegraphics[width=\figsize]{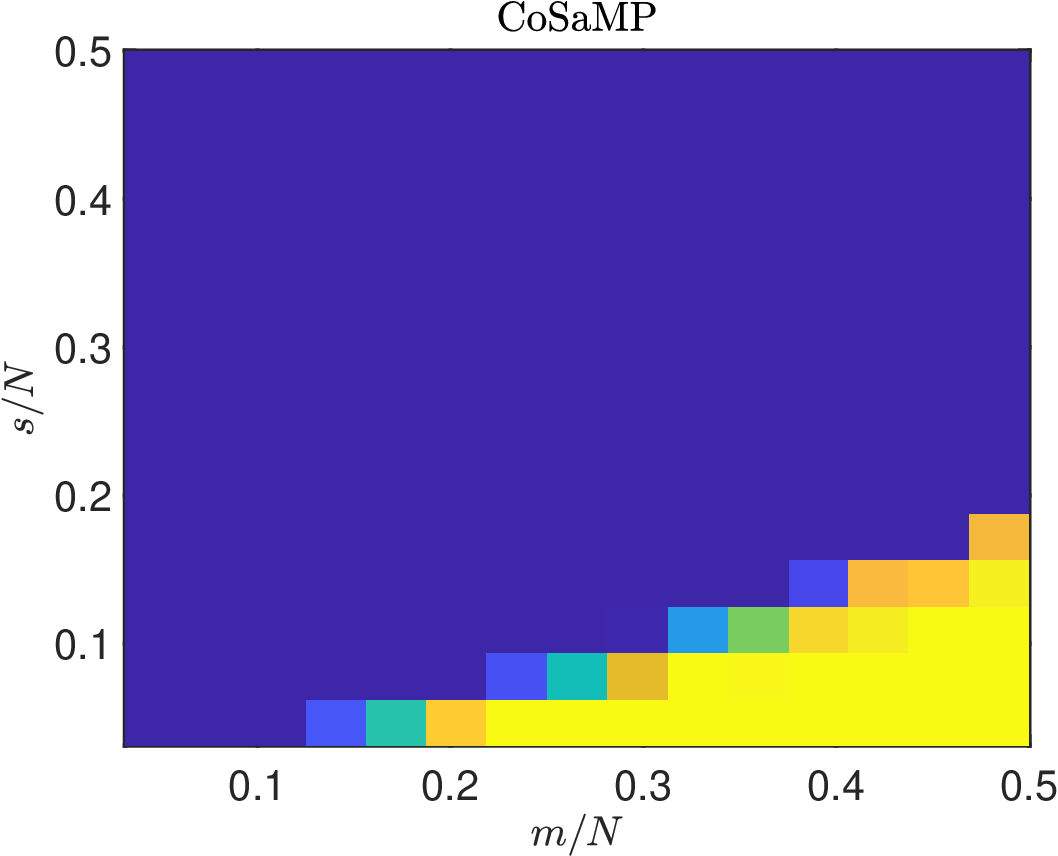} &
\includegraphics[width=\figsize]{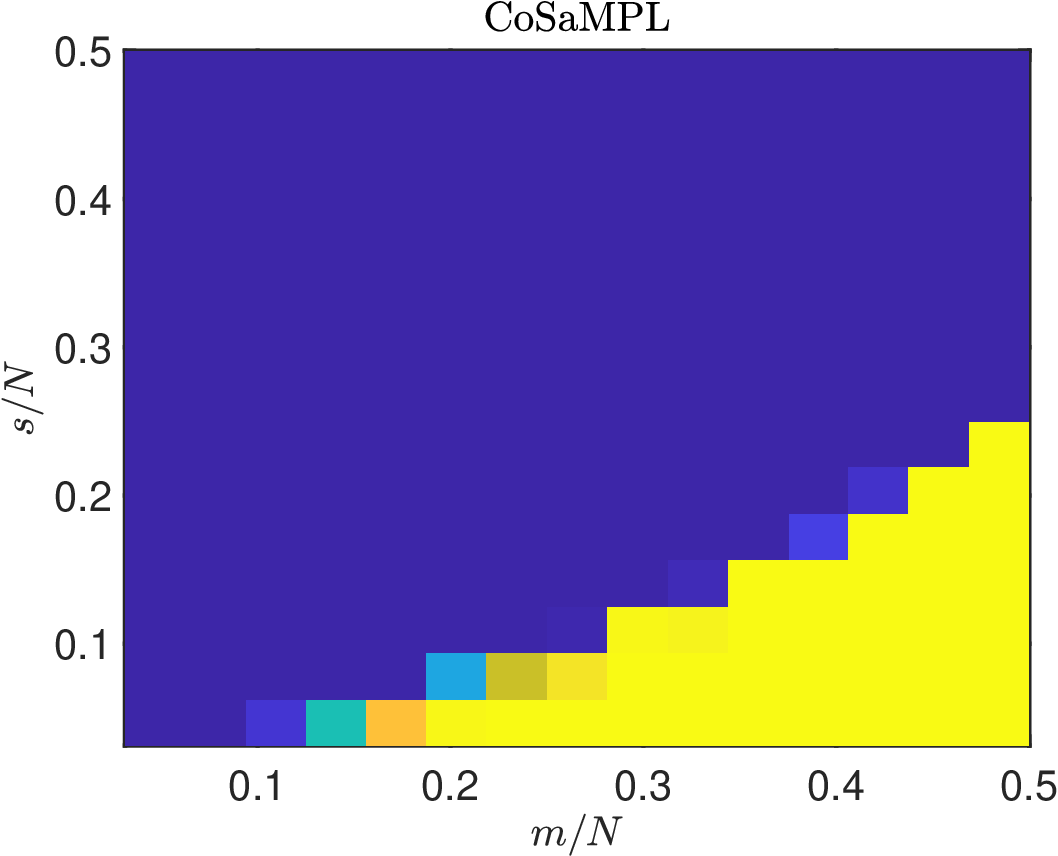} \\
\includegraphics[width=\figsize]{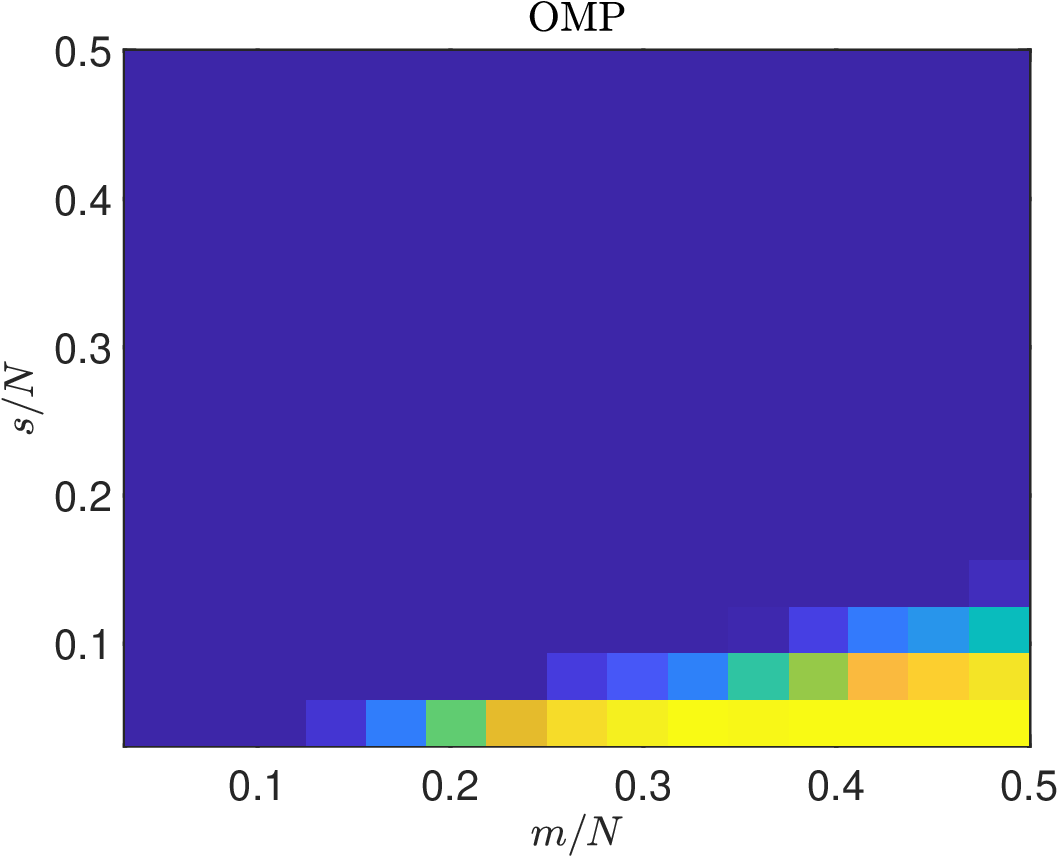} &
\includegraphics[width=\figsize]{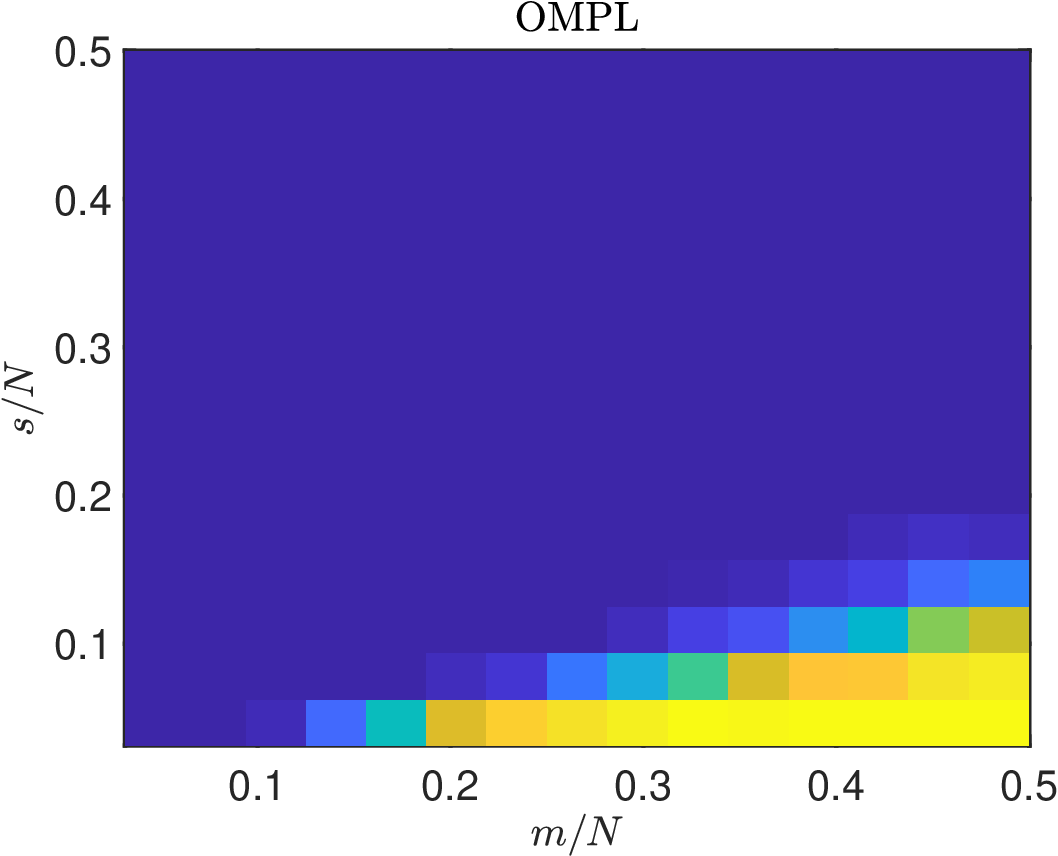} 
\end{tabular}
\end{center}
\caption[example] 
{ \label{fig:PTPFull4level}
Phase transition plots comparing the standard sparse decoders of IHT, CoSaMP, and OMP against the levels-based generalizations for $N = 256$. Here, the underlying vector is $\bm{s} = (s/2 , 0 , s/2 , 0)$ sparse in four levels $\bm{M} = (N/4,N/2,3N/4,N)$. Row one contains IHT, IHTL, row two NIHT, and NIHTL, row three CoSaMP and CoSaMPL, and the final row OMP and OMPL.
}
\end{figure}

The final numerics are contained in Figure \ref{fig:2levelunif}. This performs full phase transitions for a sparsity pattern $\bm{s} = (3s/4, s/4)$ in two levels $\bm{M} = (3s/4,N)$. This serves as a surrogate for the function approximation case. In fact, in the case of piecewise regular (specifically, piecewise $\alpha$-H\"{o}lder) functions, best $s$-term approximation rates of order $O(s^{-\alpha})$ are realized by  sparse-in-levels approximations to the vector of wavelet coefficients where the coarsest levels are fully saturated  (see \cite{BASBMKRCSwavelet} for further details). While NIHTL and CoSaMPL show improvement over NIHT and CoSaMP, IHTL and OMPL only improve over IHT and OMP in the low total sparsity regime -- unlike our other experiments.

\begin{figure} [ht]
\begin{center}
\begin{tabular}{ccc} 
\includegraphics[width=\figsize]{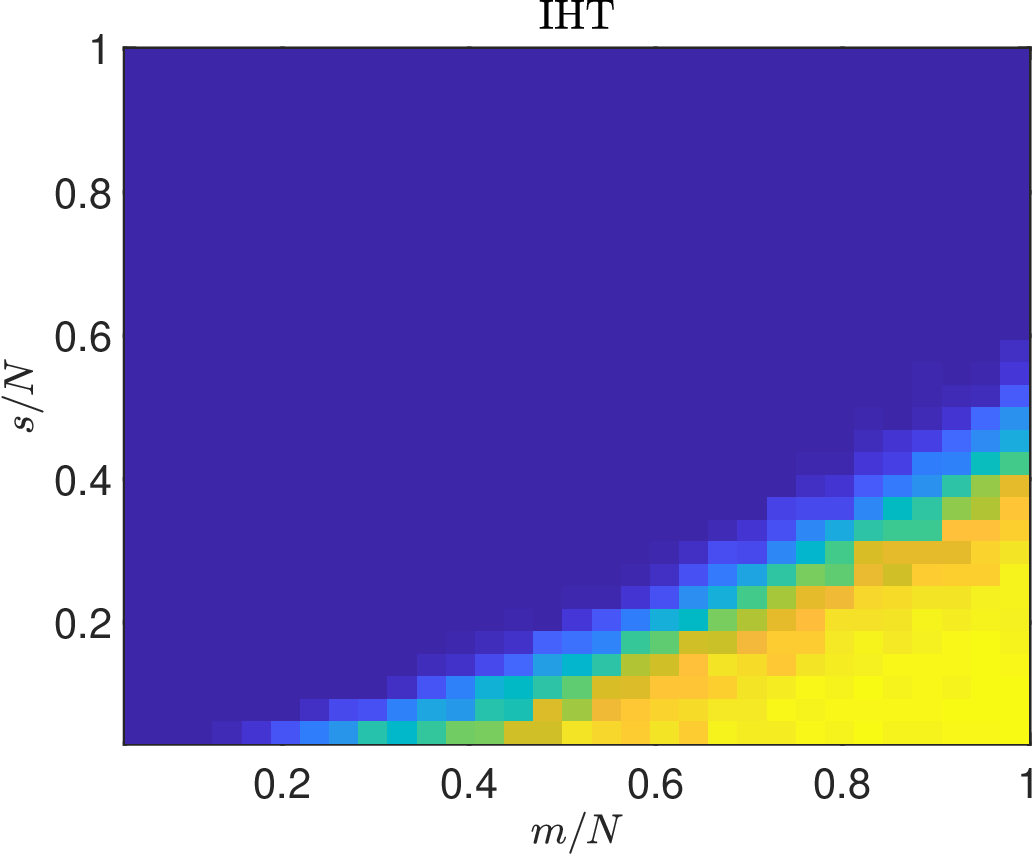} &
\includegraphics[width=\figsize]{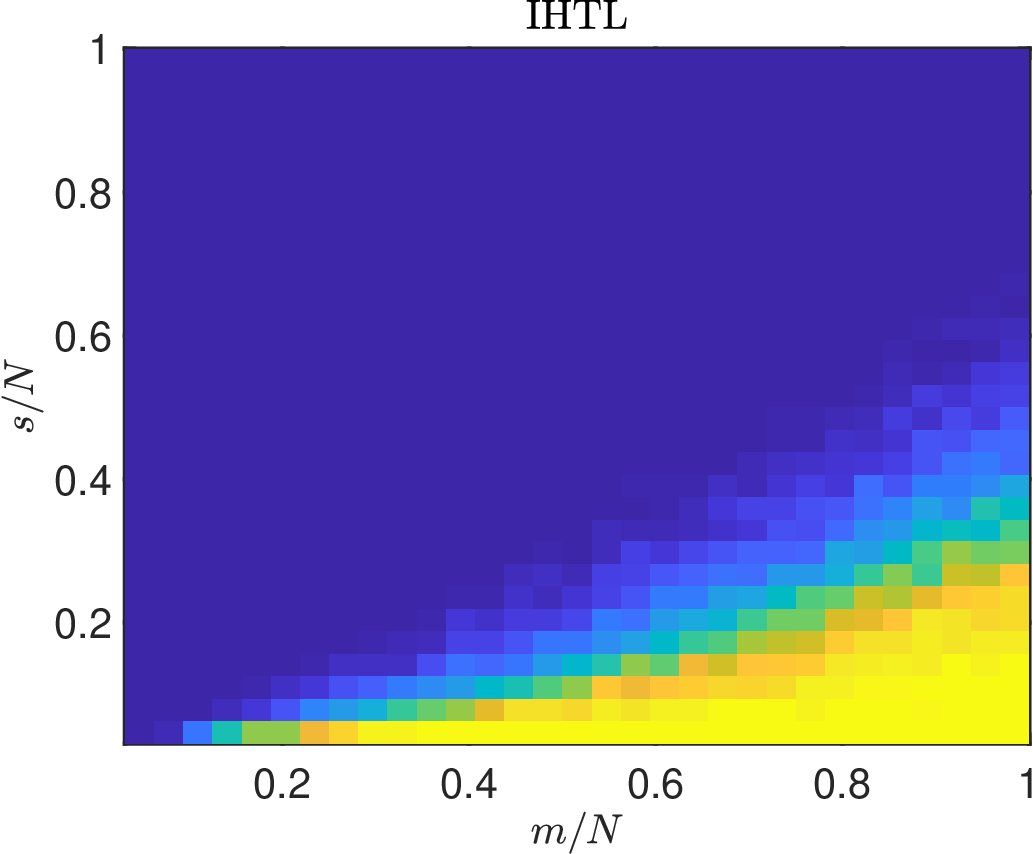}\\
\includegraphics[width=\figsize]{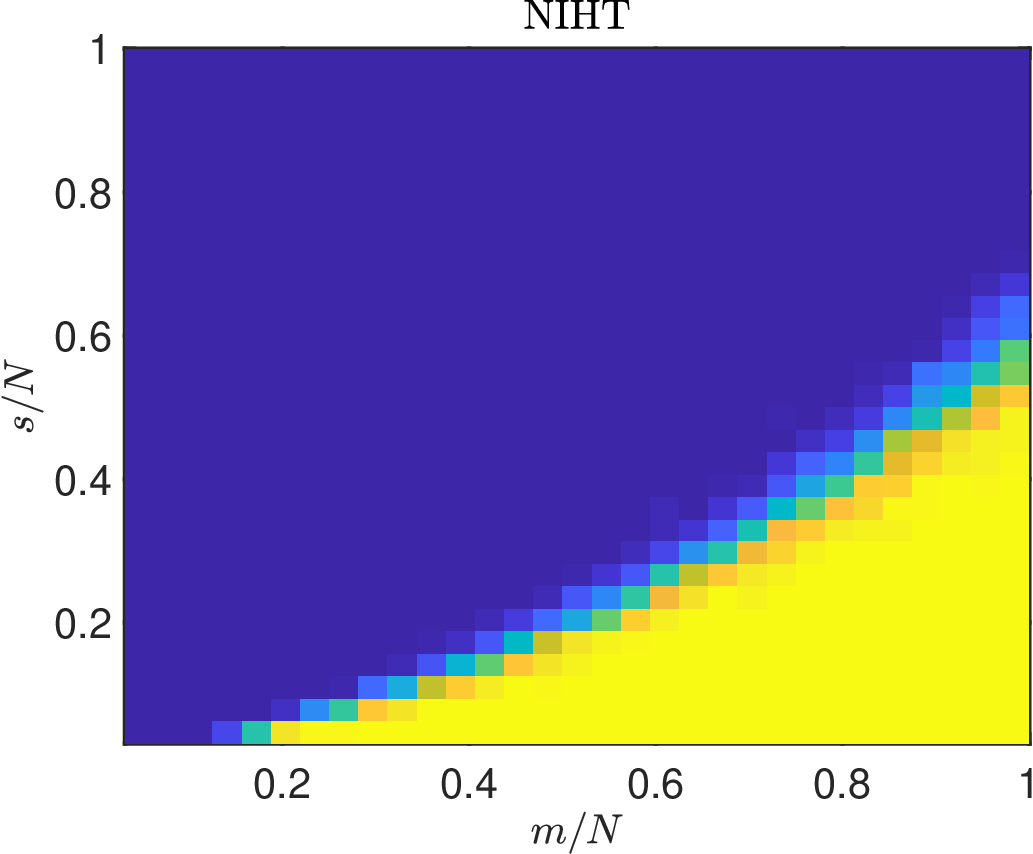} &
\includegraphics[width=\figsize]{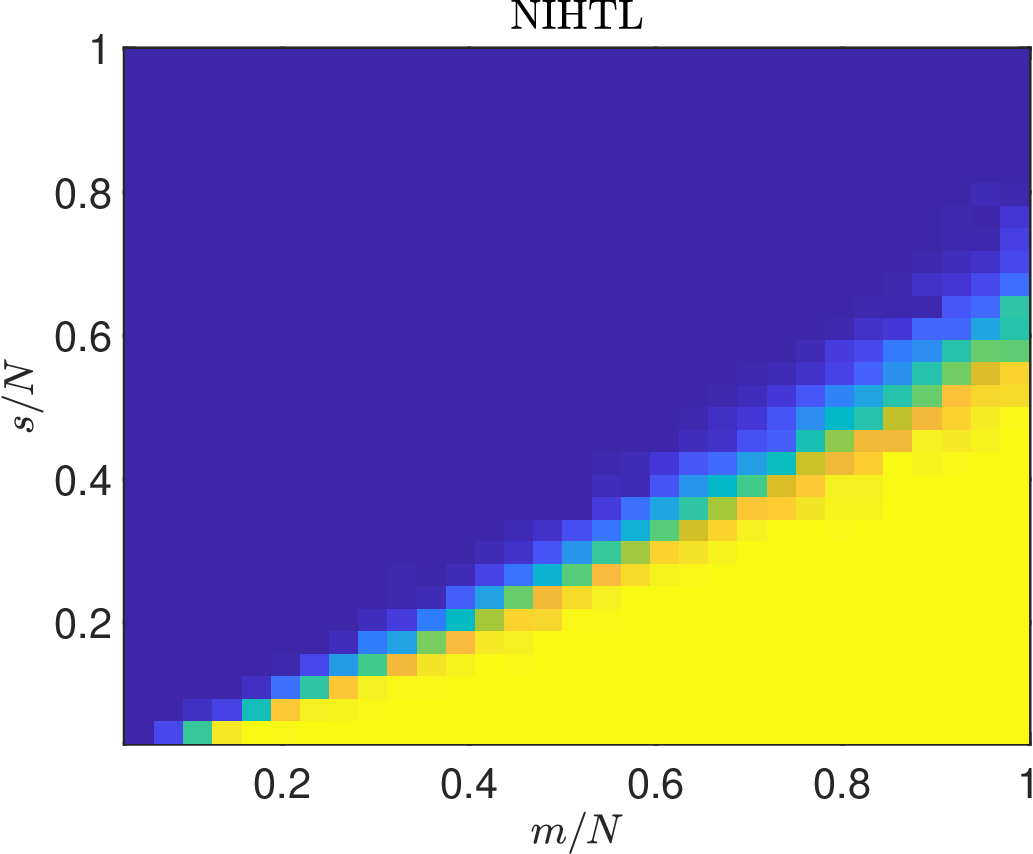}\\
\includegraphics[width=\figsize]{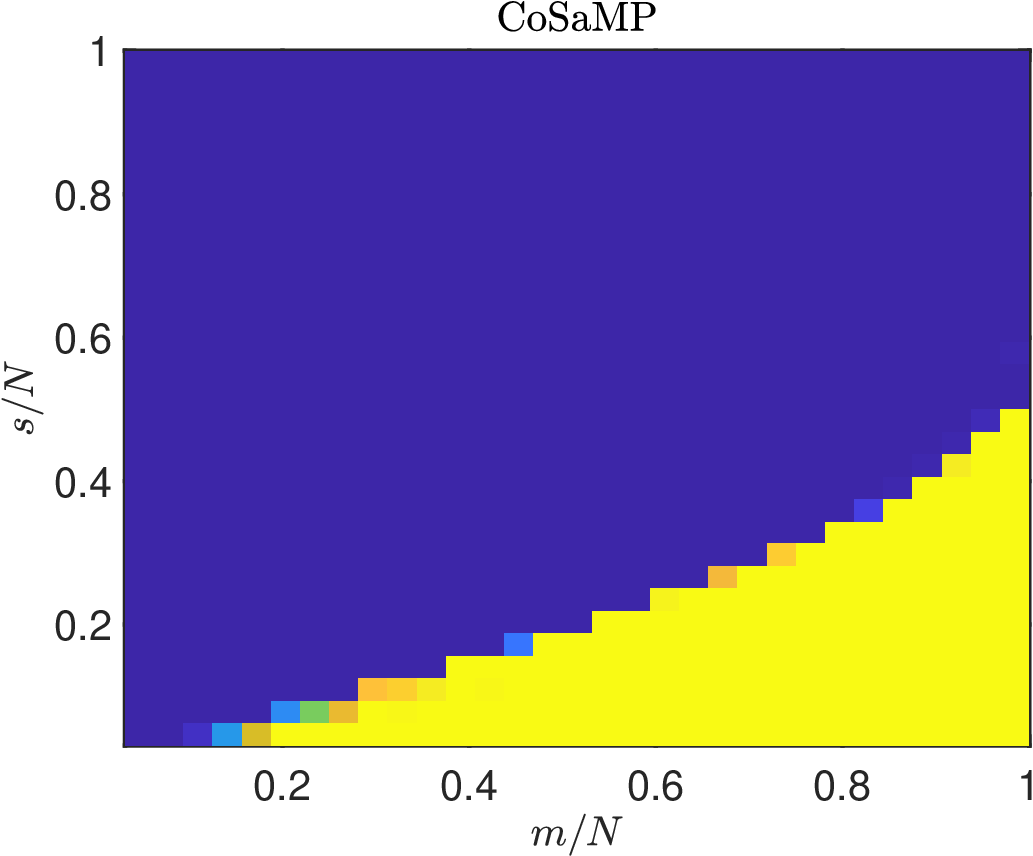} &
\includegraphics[width=\figsize]{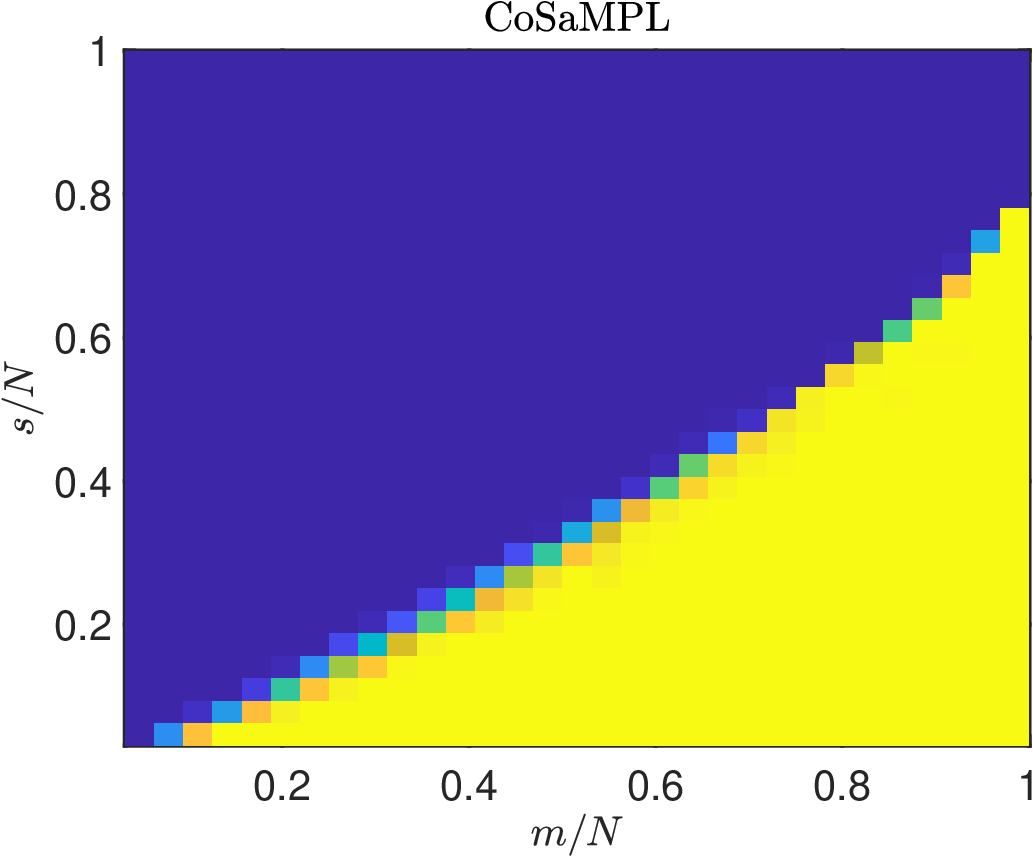}\\
\includegraphics[width=\figsize]{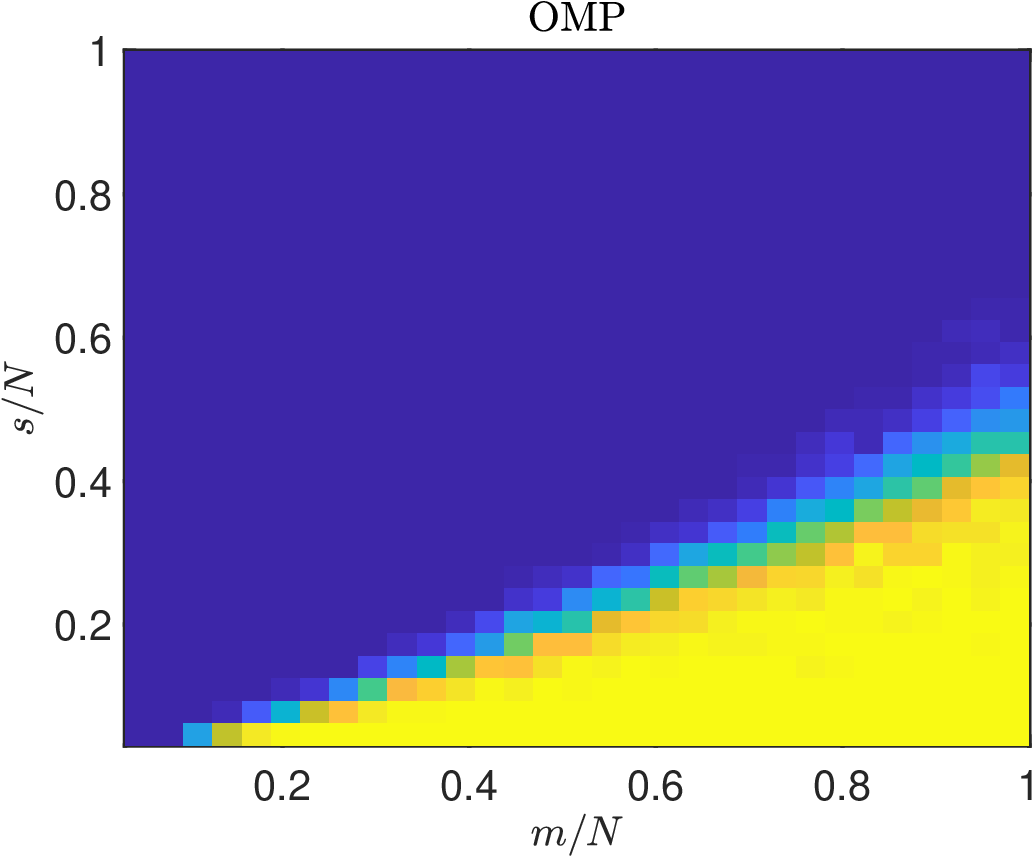} &
\includegraphics[width=\figsize]{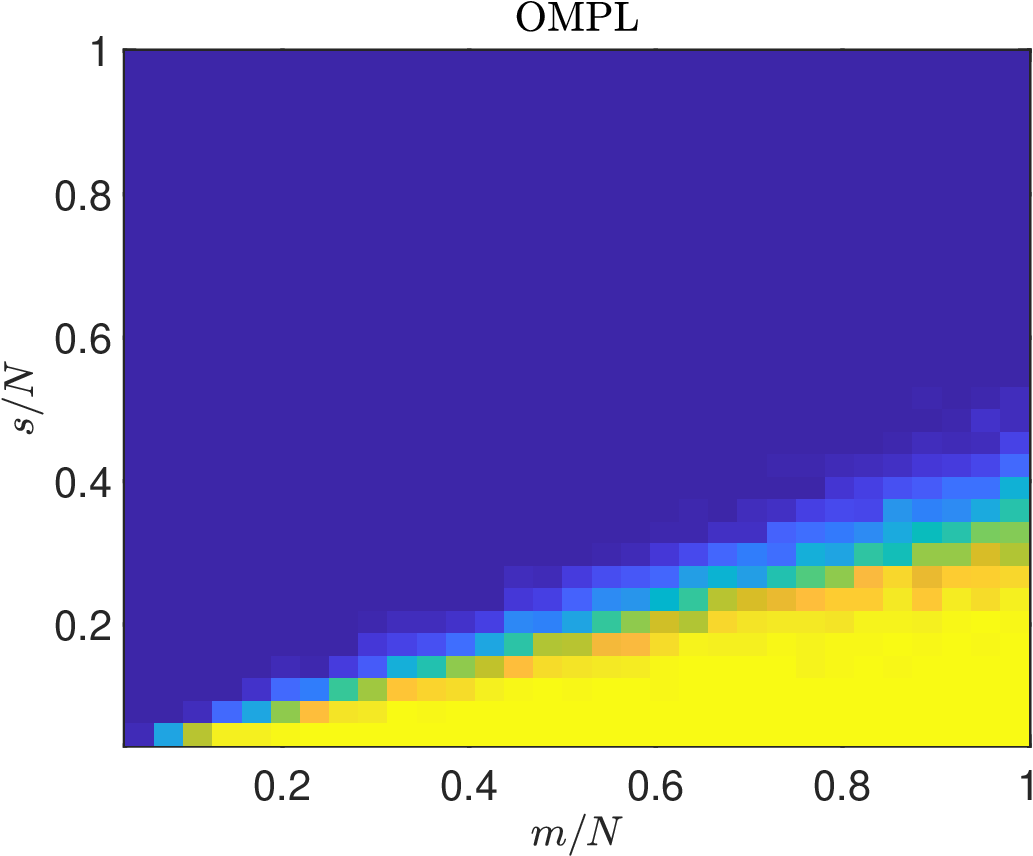}
\end{tabular}
\end{center}
\caption[example] 
{ \label{fig:2levelunif} 
Phase transition plots comparing the standard sparse decoders of IHT, NIHT, CoSaMP, and OMP against the levels-based generalizations for $N = 256$. Here, the underlying vector is $\bm{s} = (3s/4, s/4)$ sparse in two levels $\bm{M} = (3s/4,N)$.

}
\end{figure} 

\section{Proofs} \label{sec:proofs}

\subsection{Outline}
As this series of proofs is quite lengthy, we begin by outlining the main steps. Both \cref{lemma:616Generalized} and \cref{lemma:620generalized} are direct generalizations of standard sparse results, giving useful bounds involving the RICL constant. Using these, we prove an extremely key result, \cref{thm:WeightedTheorem}, which gives conditions on any vector $x'$ and RICL constant $\delta_{\bm{s},\bm{M}}$ to guarantee the true solution $x$ and $x'$ are sufficiently close. Using this result, the overall argument for both IHTL and CoSaMPL is similar. In either case we use \cref{lemma:616Generalized} and \cref{lemma:620generalized}, along with careful tracking of index sets, to show that $x^{(n)} = x'$ satisfies the assumptions of \cref{thm:WeightedTheorem}. From this, the final results follow immediately. This style of argument is extended from the sparse case contained in \cite{FoucartRauhutCSbook}.

\subsection{Preliminary Lemmas}

The following two results are based on \cite[Lemma 6.16]{FoucartRauhutCSbook}, and \cite[Lemma 6.20]{FoucartRauhutCSbook} respectively.

\begin{lemma}  \label{lemma:616Generalized}
Let $u,v \in \mathbb{C}^{N}$ be $(\bm{s}',\bm{M})$-sparse and $(\bm{s}'',\bm{M})$-sparse respectively, and $\Delta \in D_{\bm{s},\bm{M}}$ be arbitrary. Then for any matrix $A \in \mathbb{C}^{m \times N}$,
\begin{align*}
    (i)& \qquad \vert \langle u, (I -A^{*}A)v \rangle \vert \leq \delta_{\bm{s}'+\bm{s}'',\bm{M}} \Vert u \Vert_{\ell^{2}} \Vert v \Vert_{\ell^{2}}\\
    (ii)& \qquad \Vert P_{\Delta}(I-A^{*}A)v \Vert_{\ell^{2}} \leq \delta_{\bm{s} + \bm{s}'',\bm{M}} \Vert v \Vert_{\ell^{2}} 
\end{align*}
\end{lemma}
\begin{proof}
 
To show $(i)$ we expand the inner product
\begin{align*}
    \vert \langle u, (I -A^{*}A)v \rangle \vert = \vert \langle u,v \rangle - \langle Au, Av \rangle \vert 
\end{align*}
and define $\Xi= \text{supp}(u) \cup \text{supp}(v) \in D_{\bm{s}'+\bm{s}'',\bm{M}}$. Then the above may be written as 
\begin{align}
    \vert \langle u,v \rangle - \langle Au, Av \rangle \vert  &= \vert \langle P_{\Xi} u, P_{\Xi} v \rangle - \langle (AP_{\Xi}) P_{\Xi}u, (AP_{\Xi}) P_{\Xi} v \rangle \vert \nonumber \\
    &=     \vert \langle P_{\Xi}u, (P_{\Xi} -(P_{\Xi}A^{*}AP_{\Xi})P_{\Xi})v \rangle \vert \nonumber \\
    &\leq  \Vert P_{\Xi}u \Vert_{\ell^{2}} \Vert P_{\Xi} -(P_{\Xi}A^{*}AP_{\Xi})\Vert_{\ell^{2}} \Vert P_{\Xi}v \Vert_{\ell^{2}} \label{UsingAlternateDefnOfRICL}
\end{align}
As $P_{\Xi}v$ is $(\bm{s}'',\bm{M})$-sparse and thus $(\bm{s}'+\bm{s}'',\bm{M})$-sparse, we use the fact that
$\delta_{\bm{s}'+\bm{s}'',\bm{M}} = \displaystyle \sup_{\Xi \in D_{\bm{s}'+\bm{s}'',\bm{M}}} \Vert P_{\Xi} - P_{\Xi}A^{*}AP_{\Xi} \Vert_{\ell^{2}}$
to obtain that the right-hand side of~\cref{UsingAlternateDefnOfRICL} may be written as
\begin{align*}
   & \Vert P_{\Xi}u \Vert_{\ell^{2}} \Vert P_{\Xi} -(P_{\Xi}A^{*}AP_{\Xi})\Vert_{\ell^{2}} \Vert P_{\Xi}v \Vert_{\ell^{2}} \\
   & \quad \leq \delta_{\bm{s}' +\bm{s}'',\bm{M}} \Vert P_{\Xi} u \Vert_{\ell^{2}} \Vert  P_{\Xi}v \Vert_{\ell^{2}} 
    = \delta_{\bm{s}' +\bm{s}'',\bm{M}} \Vert u \Vert_{\ell^{2}} \Vert  v \Vert_{\ell^{2}},
\end{align*}
which gives $(i)$. 

For $(ii)$, we note that
\begin{equation*}
    \Vert P_{\Delta}(I-A^{*}A)v \Vert_{\ell^{2}}^{2} = \vert \langle P_{\Delta}(I-A^{*}A)v, (I-A^{*}A)v \rangle \vert 
\end{equation*}
and apply $(i)$ with $u = P_{\Delta}(I-A^{*}A)v$, giving
\begin{equation*}
     \Vert P_{\Delta}(I-A^{*}A)v \Vert_{\ell^{2}}^{2} \leq \delta_{\bm{s} + \bm{s}'',\bm{M}} \Vert P_{\Delta}(I-A^{*}A)v \Vert_{\ell^{2}}\Vert v \Vert_{\ell^{2}}
\end{equation*}
and dividing through by $ \Vert P_{\Delta}(I-A^{*}A)v \Vert_{\ell^{2}}$ gives the desired result.
\end{proof}

\begin{lemma} \label{lemma:620generalized}
Let $e \in \mathbb{C}^{m}$, $A \in \mathbb{C}^{m\times N}$ with RICL  $\delta_{\bm{s},\bm{M}}$ and $\Delta \in D_{\bm{s},\bm{M}}$. Then
$
    \Vert P_{\Delta}A^{*}e \Vert_{\ell^{2}} \leq \sqrt{1 + \delta_{\bm{s},\bm{M}}}\Vert e \Vert_{\ell^{2}}.
$
\end{lemma}

\begin{proof}
We compute
\begin{align*}
    \Vert P_{\Delta}A^{*}e \Vert_{\ell^{2}}^{2} 
    & = \langle A^{*}e, P_{\Delta}A^{*}e \rangle = \langle e, A P_{\Delta}A^{*}e \rangle \\
    & \leq \Vert e \Vert_{\ell^{2}} \Vert AP_{\Delta}A^{*}e \Vert_{\ell^{2}}.
\end{align*}
But as $P_{\Delta}A^{*}e$ is $(\bm{s},\bm{M})$-sparse we have
\begin{equation*}
     \Vert e \Vert_{\ell^{2}} \Vert AP_{\Delta}A^{*}e \Vert_{\ell^{2}} \leq \Vert e \Vert_{\ell^{2}} \sqrt{1+\delta_{\bm{s},\bm{M}}} \Vert P_{\Delta}A^{*}e \Vert_{\ell^{2}}
\end{equation*}
and dividing through by $ \Vert P_{\Delta}A^{*}e \Vert_{\ell^{2}}$ gives the desired result.
\end{proof}

With these is hand, we prove a key result. This theorem is directly extended from the sparse case in \cite[Lemma 6.23]{FoucartRauhutCSbook}.

\begin{theorem} \label{thm:WeightedTheorem} Suppose $A \in \mathbb{C}^{m \times N}$ satisfies the RIPL of order $(\bm{s},\bm{M})$ and has RICL $\delta_{\bm{s},\bm{M}} < 1$. Let $\kappa, \tau > 0, \lambda \geq 0$ and $e \in \mathbb{C}^{m}$ be given, and $w \in \mathbb{R}^N$, with $w >0$, be a set of weights constant on each level, such that $w_{i} = w^{(k)}$, for $M_{k-1} < i \leq M_k$ and  $1 \leq  k \leq r$. Suppose we have $x,x' \in \mathbb{C}^{N}$ such that
\begin{equation*}
    x' \in D_{\kappa \bm{s},\bm{M}} , \, \text{ and } \, \Vert P_{\Xi}x - x' \Vert_{\ell^{2}} \leq \tau \Vert AP_{\Xi^{c}}x + e \Vert_{\ell^{2}} + \lambda,
\end{equation*} 
where $\Xi = \Xi_{1} \cup \ldots \cup \Xi_{r}$, and $\Xi_{i}$ is the index set of the largest $2s_{i}$ entries of $P_{M_{i}}^{M_{i-1}}x$. Then, there exist constants $C_{\kappa,\tau}, D_{\kappa,\tau}, F_{\kappa,\tau} >0$ depending only on $\kappa$ and $\tau$ and $E_{\kappa}>0$ depending only on $\kappa$ such that
\begin{align*}
\Vert x - x' \Vert_{\ell^{1}_{w}} 
&\leq C_{\kappa,\tau} \frac{\sqrt{\zeta}}{\sqrt{\xi}}\sigma_{\bm{s,M}}(x)_{\ell_{w}^{1}} 
+ D_{\kappa, \tau} \sqrt{\zeta} \Vert e \Vert_{\ell^{2}} 
+ {E_{\kappa}} \sqrt{\zeta} \lambda, \\
\Vert x - x' \Vert_{\ell^{2}} &\leq \frac{F_{\tau}}{\sqrt{\xi}}\sigma_{\bm{s,\bm{M}}}(x)_{\ell^{1}_{w}} + \tau \Vert e \Vert_{\ell^{2}} + \lambda,
\end{align*}
where $\zeta$ and $\xi$ are as in \eqref{zetaxidefinition}.
\end{theorem}

\begin{proof}
Let us consider some fixed level $i$, and $\Xi$ defined as above. We consider the case of the weighted 1-norm first. Projecting onto level $i$ gives
\begin{align*}
    & \Vert P_{M_{i}}^{M_{i-1}}(x - x') \Vert_{\ell^{1}_{w}} \\
    & \qquad\leq w^{(i)} \Vert P_{\Xi^{c}_{i}}x \Vert_{\ell^{1}} + w^{(i)} \Vert P_{\Xi_{i}}x - P_{M_{i}}^{M_{i-1}}x' \Vert_{\ell^{1}},
\end{align*}
and where $\Xi^c_i = \{M_{i-1}+1,\ldots,M_i\} \setminus \Xi_i$  is the relative complement of $\Xi_i$ with respect to the level $i$. We bound the latter term by noting that $ P_{\Xi_{i}}x - P_{M_{i}}^{M_{i-1}}x'$ is $(2+\kappa)s_{i}$-sparse, so that
\begin{align*}
    & w^{(i)} \Vert P_{\Xi_{i}}x - P_{M_{i}}^{M_{i-1}}x' \Vert_{\ell^{1}} \\
    &\qquad\leq \sqrt{(2+\kappa)s_{i}(w^{(i)})^{2} } \Vert P_{\Xi_{i}}x - P_{M_{i}}^{M_{i-1}}x' \Vert_{\ell^{2}}.
\end{align*}
Further defining $\Delta \in D_{\bm{s},\bm{M}}$ to be the index set of a best $(\bm{s},\bm{M})$-term approximation to $x$, we bound the former term by   $w^{(i)}\Vert P_{\Xi_{i}^{c}} x \Vert_{\ell^{1}} \leq w^{(i)}\Vert P_{\Delta_{i}^{c}} x \Vert_{\ell^{1}}$, where $\Delta_{i} = \Delta \cap \{ M_{i-1}+1 , \ldots, M_{i} \}$ and $\Delta_{i}^{c} = \{ M_{i-1}+1 , \ldots, M_{i} \} \backslash \Delta_{i}$. Summing over all levels $i = 1, \ldots, r$ and using the Cauchy-Schwarz inequality gives
\begin{equation*}
    \Vert x - x' \Vert_{\ell^{1}_{w}} \leq \sigma_{\bm{s},\bm{M}}(x)_{\ell_{w}^{1}} + \sqrt{(2+\kappa)\zeta}\Vert P_{\Xi}x - x' \Vert_{\ell^{2}}.
\end{equation*}
By supposition we have then
\begin{equation}
    \Vert x - x' \Vert_{\ell^{1}_{w}} \leq \sigma_{\bm{s},\bm{M}}(x)_{\ell_{w}^{1}} + \sqrt{(2+\kappa)\zeta}(\tau \Vert AP_{\Xi^{c}}x + e \Vert_{\ell^{2}} + \lambda). \label{weightedIntermediateStep}
\end{equation}
We now perform a particular decomposition of $\Delta^{c}=\{1,\ldots,N\}\setminus \Delta$, letting
\begin{equation*}
    \Delta_{i}^{c} = \Xi_{i,1} \cup \Xi_{i,2} \cup \ldots,
\end{equation*}
where $\Xi_{i,1}$ is index set of the $s_{i}$ largest entries of $P_{\Delta_{i}^{c}}x$, $\Xi_{i,2}$ is the index set of the largest $s_{i}$ entries of $P_{(\Delta_{i} \cup \Xi_{i,1})^{c}}x$, and so on. (Note that $\Delta_i^c$, $(\Delta_{i} \cup \Xi_{i,1})^{c}$, etc.\ are relative complements with respect to the level $i$). This allows us to define the collection $\Xi^{(k)}$ for $k=1,2,\ldots$
\begin{equation*}
    \Xi^{(k)} = \bigcup_{i=1}^{r} \Xi_{k,i}, \quad \text{ where by construction } \Xi^{(k)} \in D_{\bm{s,M}},
\end{equation*}
and furthermore $ \Xi^{c} = \bigcup_{k \geq 2} \Xi^{(k)}.$ Using this decomposition and the RIPL assumption we have
\begin{align*}
    & \Vert AP_{\Xi^{c}}x + e \Vert_{\ell^{2}} 
    \leq \sum_{k \geq 2} \sqrt{1 + \delta_{\bm{s},\bm{M}}} \Vert P_{\Xi^{(k)}} x \Vert_{\ell^{2}} + \Vert e \Vert_{\ell^{2}} \nonumber\\ 
    & \qquad\leq\sqrt{2}  \sum_{k \geq 1} \sqrt{\sum_{i=1}^{r} \frac{1}{s_{i}} \Vert P_{\Xi_{i,k}} x  \Vert_{\ell^{1}}^{2} } + \Vert e \Vert_{\ell^{2}} \\
    &\qquad\leq \sqrt{2} \frac{1}{\sqrt{\xi}} \sum_{k \geq 1} \sqrt{\sum_{i=1}^{r}\Vert P_{\Xi_{i,k}} x  \Vert_{\ell^{1}_{w}}^{2} } + \Vert e \Vert_{\ell^{2}} \\
    & \qquad\leq \tfrac{\sqrt{2}}{\sqrt{\xi}} \Vert P_{\Delta^{c}} x \Vert_{\ell_{w}^{1}}  + \Vert e \Vert_{\ell^{2}} 
     = \tfrac{\sqrt{2}}{\sqrt{\xi}}\sigma_{\bm{s,M}}(x)_{\ell_{w}^{1}} + \Vert e \Vert_{\ell^{2}}.
\end{align*}
Combining this result with \cref{weightedIntermediateStep}, we have 
\begin{align*}
    & \Vert x- x' \Vert_{\ell_{w}^{1}} 
    \leq \left[1+ \tfrac{\sqrt{(4 + 2\kappa)\zeta}\tau}{\sqrt{\xi}} \right]\sigma_{\bm{s,M}}(x)_{\ell_{w}^{1}} \\
    & \qquad \qquad \qquad + \sqrt{(2+\kappa)\zeta} \tau \Vert  e \Vert_{\ell^{2}} 
     + \sqrt{(2+\kappa)\zeta} \lambda \\
    &\qquad = {C_{\kappa,\tau}}\tfrac{\sqrt{\zeta}}{\sqrt{\xi}}\sigma_{\bm{s,M}}(x)_{\ell_{w}^{1}} + D_{\kappa , \tau} \sqrt{\zeta} \Vert e \Vert_{\ell^{2}} + {E_{\kappa}} \sqrt{\zeta} \lambda,
\end{align*}
as was to be shown.

For the 2-norm case we again focus on particular level $i$. Using the definition of $\Xi_i$ and Stechkin's inequality (see, e.g., \cite[Proposition 2.3]{FoucartRauhutCSbook}), we see that
\begin{align*}
    &\Vert P_{M_{i}}^{M_{i-1}}(x-x')\Vert_{\ell^{2}}^{2} 
    = \Vert P_{\Xi_{i}^{c}} x \Vert_{\ell^{2}}^{2} + \Vert P_{\Xi_{i}}x -P_{M_{i}}^{M_{i-1}}x' \Vert_{\ell^{2}}^{2} \\
    &\qquad\qquad\qquad \leq \tfrac{1}{(w^{(i)})^{2}s_{i}} \Vert P_{\Xi_{i}^{c}} x \Vert_{\ell_{w}^{1}}^{2} + \Vert P_{\Xi_{i}}x -P_{M_{i}}^{M_{i-1}}x' \Vert_{\ell^{2}}^{2}.
\end{align*}
Summing over all levels $i=1, \ldots, r$ we have 
\begin{align*}
    \Vert x-x' \Vert_{\ell^{2}}^{2} 
    &\leq \tfrac{1}{\xi} \Vert P_{\Xi^{c}} x \Vert_{\ell_{w}^{1}}^{2} + \Vert P_{\Xi}x -x' \Vert_{\ell^{2}}^{2} \\
    &\leq \tfrac{1}{\xi}\sigma_{\bm{s,\bm{M}}}(x)_{\ell^{1}_{w}}^{2} + (\tau \Vert AP_{\Xi^{c}}x + e \Vert_{\ell^{2}} + \lambda)^{2},
\end{align*}
where we have applied the definitions of $\xi,\Xi,$ and the assumptions of the theorem. As a result we also have
\begin{equation*}
     \Vert x-x' \Vert_{\ell^{2}} 
     \leq \tfrac{1}{\sqrt{\xi}}\sigma_{\bm{s,\bm{M}}}(x)_{\ell^{1}_{w}} + \tau \Vert AP_{\Xi^{c}}x + e \Vert_{\ell^{2}} + \lambda.
\end{equation*}
As we have already bounded this second term, we have
\begin{align*}
\Vert x-x' \Vert_{\ell^{2}} 
&\leq \tfrac{1 + \sqrt{2}\tau}{\sqrt{\xi}}\sigma_{\bm{s,\bm{M}}}(x)_{\ell^{1}_{w}} + \tau \Vert e \Vert_{\ell^{2}} + \lambda \\
& = \tfrac{{F_{\tau}}}{\sqrt{\xi}}\sigma_{\bm{s,\bm{M}}}(x)_{\ell^{1}_{w}} + \tau \Vert e \Vert_{\ell^{2}} + \lambda,
\end{align*}
thus completing the proof.
\end{proof}

\subsection{Proving \cref{thm:IHTFinalTheorem} on IHT}

The following theorem is based on \cite[{Theorem 6.18}]{FoucartRauhutCSbook}.

\begin{theorem} \label{thm:IHTIntermediateBound}
Suppose $x \in \mathbb{C}^{N}$ is $(\bm{s},\bm{M})$-sparse in levels, with the RIPL constant satisfying
\begin{equation*}
    \delta_{3\bm{s},\bm{M}} < \tfrac{1}{\sqrt{3}}.
\end{equation*}
Then, for all $x \in \mathbb{C}^{N}$, $e \in \mathbb{C}^{m}$ and $\Delta \in D_{\bm{s},\bm{M}}$, the sequence $(x^{(n)})_{n \geq0}$ defined by the algorithm $\mathrm{IHTL}(A,y,\bm{s},\bm{M})$ for $y = Ax + e$ satisfies 
\begin{equation*}
\Vert x^{(n)} - P_{\Delta}x \Vert_{\ell^{2}} \leq \rho^{n} \Vert x^{(0)} - P_{\Delta}x\Vert_{\ell^{2}} + \tau \Vert AP_{\Delta^{c}}x + e \Vert_{\ell^{2}}.
\end{equation*}
where $\rho = \sqrt{3}\delta_{3\bm{s},\bm{M}} < 1$ and $\tau > 0$ only depends on $\rho$ and $\delta_{3 \bm{s},\bm{M}}$, with $\tau \leq 2.18/(1-\rho)$.
\end{theorem}

\begin{proof} We firstly define $\Delta_{i} = \Delta \cap \{ M_{i-1}+1 , \ldots, M_{i} \}$ and $\Delta_{i}^{c} = \{ M_{i-1}+1 , \ldots, M_{i} \} \backslash \Delta_{i}$. Analogously to the proof of Theorem~\ref{thm:WeightedTheorem}, it will prove to be convenient to decompose 
\begin{equation*}
    \{ 1, \ldots, N \} = \Delta \cup \Delta^{c} = \bigcup_{i=1}^{r} \Delta_{i} \cup \Delta_{i}^{c}.
\end{equation*}
Similarly we define $\Delta_{i}^{n+1}$  as the index set of the largest $s_{i}$ entries of $x^{(n+1)}$ in the band $\{ M_{i-1}+1, \ldots M_{i}\}$. With this decomposition, we may use techniques near-identical to those in \cite[Theorem 6.18]{FoucartRauhutCSbook}, and thus we give a brief treatment where possible. By definition, for any $\Delta_{i}$,
\begin{align*}
    & \Vert P_{\Delta_{i}}( x^{(n)} + A^{*}(y-Ax^{(n)})) \Vert_{\ell^{2}} \\
    & \qquad \qquad \qquad \qquad \leq \Vert P_{\Delta^{n+1}_{i}}( x^{(n)} + A^{*}(y-Ax^{(n)})) \Vert_{\ell^{2}}.
\end{align*}
Then we cancel any shared contribution on the set $\Delta_{i} \cap \Delta_{i}^{n+1}$,
\begin{align}
     & \Vert P_{\Delta_{i} \backslash \Delta_{i}^{n+1}}( x^{(n)} + A^{*}(y-Ax^{(n)})) \Vert_{\ell^{2}} \nonumber \\
     & \qquad \qquad \leq \Vert P_{\Delta^{n+1}_{i} \backslash \Delta_{i}}( x^{(n)} + A^{*}(y-Ax^{(n)})) \Vert_{\ell^{2}}. \label{removedIntIneq}
\end{align}
Here, making the observation that $P_{\Delta_{i}}x  = 0$ on $\Delta^{n+1}_{i} \backslash \Delta_{i}$ and 
$P_{M_{i}}^{M_{i-1}}x^{(n+1)}  = 0$ on  $\Delta_{i} \backslash \Delta^{(n+1)}_{i}$, we write the right-hand side of \cref{removedIntIneq} as 
\begin{align*}
    \Vert P_{\Delta^{n+1}_{i} \backslash \Delta_{i}}( x^{(n)} - P_{\Delta_{i}}x+ A^{*}(y-Ax^{(n)})) \Vert_{\ell^{2}},
\end{align*}
and we bound the left-hand side of \cref{removedIntIneq} from below as 
\begin{align*}
    & \Vert P_{\Delta_{i} \backslash \Delta_{i}^{n+1}}( x^{(n)} + A^{*}(y-Ax^{(n)})) \Vert_{\ell^{2}}\\
    &\qquad \quad \geq \Vert P_{\Delta_{i} \backslash \Delta_{i}^{n+1}}( P_{\Delta_{i}}x - P_{M_{i}}^{M_{i-1}}x^{(n+1)}) \Vert_{\ell^{2}}  \\
    & \qquad \qquad - \Vert P_{\Delta_{i} \backslash \Delta_{i}^{n+1}}(x^{(n)} -  P_{\Delta_{i}}x + A^{*}(y-Ax^{(n)})) \Vert_{\ell^{2}}.
\end{align*}
Combining these both into \cref{removedIntIneq} we find that
\begin{align}
    & \Vert P_{\Delta_{i} \backslash \Delta_{i}^{n+1}}( P_{\Delta_{i}}x - P_{M_{i}}^{M_{i-1}}x^{(n+1)}) \Vert_{\ell^{2}} \label{boundSymettricDifferenceChain} \\
    &\qquad \leq \sqrt{2} \Vert P_{\Delta^{n+1}_{i} \ominus \Delta_{i}}( x^{(n)} - P_{\Delta_{i}}x+ A^{*}(y-Ax^{(n)})) \Vert_{\ell^{2}}, \nonumber
\end{align}
where $\Delta^{n+1}_{i} \ominus \Delta_{i} =  (\Delta^{n+1}_{i} \backslash \Delta_{i} )\cup( \Delta_{i} \backslash \Delta_{i}^{n+1}  )$ is the symmetric difference. We now seek to bound the left-hand side further from below. To do so, we decompose
\begin{align*}
    &\Vert  P_{M_{i}}^{M_{i-1}}x^{(n+1)} - P_{\Delta_{i}}x \Vert_{\ell^{2}}^{2}\\ 
    & \qquad \qquad = \Vert P_{\Delta^{n+1}_{i}}(P_{\Delta_{i}}x - P_{M_{i}}^{M_{i-1}}x^{(n+1)}) \Vert_{\ell^{2}}^{2} \\
    & \qquad \qquad \qquad + \Vert P_{(\Delta^{n+1}_{i})^{c}}( P_{\Delta_{i}}x - P_{M_{i}}^{M_{i-1}}x^{(n+1)}) \Vert_{\ell^{2}}^{2} \\
    & \qquad \qquad = \Vert P_{\Delta_{i}^{n+1}}(x^{(n)} - P_{\Delta_{i}}x + A^{*}(y - Ax^{(n)})) \Vert_{\ell^{2}}^{2} \\
    & \qquad \qquad \qquad  +\Vert P_{(\Delta^{n+1}_{i})^{c}}( P_{\Delta_{i}}x - P_{M_{i}}^{M_{i-1}}x^{(n+1)}) \Vert_{\ell^{2}}^{2}.
\end{align*}
Further observing that $P_{M_{i}}^{M_{i-1}}x^{(n+1)} = 0$ on $(\Delta_{i}^{n+1})^{c}$, and $P_{\Delta_{i}}x =0$ on $\Delta_{i}^{c}$, we can write
\begin{align*}
    &  \Vert  P_{M_{i}}^{M_{i-1}}x^{(n+1)} - P_{\Delta_{i}}x \Vert_{\ell^{2}}^{2} \\
    & \qquad \qquad =  \Vert P_{\Delta_{i}^{n+1}}(x^{(n)} - P_{\Delta_{i}}x + A^{*}(y - Ax^{(n)})) \Vert_{\ell^{2}}^{2} \\
     & \qquad \qquad \qquad +\Vert P_{\Delta_{i} \backslash \Delta_{i}^{n+1}}( P_{\Delta_{i}}x - P_{M_{i}}^{M_{i-1}}x^{(n+1)})  \Vert_{\ell^{2}}^{2}.
\end{align*}
Combining this argument with the previous bound \cref{boundSymettricDifferenceChain} we have, in summary,
\begin{align*}
& \Vert P_{M_{i}}^{M_{i-1}}x^{(n+1)} - P_{\Delta_{i}}x \Vert_{\ell^{2}}^{2} \\
& \qquad \leq \Vert P_{\Delta_{i}^{n+1}}(x^{(n)} - P_{\Delta_{i}}x+ A^{*}(y - Ax^{(n)})) \Vert_{\ell^{2}}^{2} \\
    &\qquad \qquad + 2 \Vert P_{\Delta^{n+1}_{i} \ominus \Delta_{i}}( x^{(n)} - P_{\Delta_{i}}x+ A^{*}(y-Ax^{(n)})) \Vert_{\ell^{2}}^{2} \\
    &\qquad \leq 3 \Vert P_{\Delta^{n+1}_{i} \cup \Delta_{i}}( x^{(n)} - P_{\Delta_{i}}x+ A^{*}(y-Ax^{(n)})) \Vert_{\ell^{2}}^{2}.
\end{align*}
Summing this over all levels $i=1, \ldots, r$, we have that
\begin{align*}
    & \Vert x^{(n+1)} - P_{\Delta}x \Vert_{\ell^{2}}^{2} \\
    & \qquad \leq 3 \Vert P_{\Delta^{n+1} \cup \Delta} ( x^{(n)} - P_{\Delta}x+ A^{*}(y-Ax^{(n)})) \Vert_{\ell^{2}}^{2}.
\end{align*}
By redefining $y = Ax + e = AP_{\Delta}x + e'$, with $e' = e + AP_{\Delta^c}x$, we may further bound this from above as
\begin{align}
& \Vert x^{(n+1)} - P_{\Delta}x \Vert_{\ell^{2}}\nonumber\\
&\quad \leq \sqrt{3} \left[\Vert P_{\Delta^{n+1} \cup \Delta} ( x^{(n)} - P_{\Delta}x+ A^{*}A(P_{\Delta}x - x^{(n)}) \Vert_{\ell^{2}}\right. \nonumber\\
& \qquad \qquad \qquad + \Vert  P_{\Delta^{n+1} \cup \Delta} A^{*}e' \Vert_{\ell^{2}} \nonumber \Big]\nonumber\\
&\quad \leq \sqrt{3} \left[\Vert P_{\Delta^{n+1} \cup \Delta}(I - A^{*}A)(x^{(n)} - P_{\Delta}x) \Vert_{\ell^{2}} \right.\nonumber\\
& \qquad \qquad \qquad +  \Vert  P_{\Delta^{n+1} \cup \Delta} A^{*}e' \Vert_{\ell^{2}}\Big]
\label{aboutToUSeRICL}.
\end{align}
Here we note that $\text{supp}(x^{(n)} - P_{\Delta}(x)) \subset \Delta \cup \Delta^{n}$, and $ (\Delta \cup \Delta^{n}) \cup (\Delta^{n+1} \cup \Delta) \in D_{3\bm{s},\bm{M}}$. These observations allow us to apply \cref{lemma:616Generalized}(ii) on the first term, 
and \cref{lemma:620generalized} on the second term, giving
\begin{align}
   & \Vert x^{(n+1)} - P_{\Delta}x \Vert_{\ell^{2}}\label{finalResultBeforeInduction}\\  
   & \qquad \leq \sqrt{3}\left[ \delta_{3\bm{s},\bm{M}} \Vert x^{(n)} - P_{\Delta}x\Vert_{\ell^{2}} + \sqrt{1 + \delta_{2\bm{s},\bm{M}}} \Vert e' \Vert_{\ell^{2}} \right].\nonumber
\end{align}
Finally by examining this inequality, we set
\begin{equation*}
    \rho = \sqrt{3}\delta_{3\bm{s},\bm{M}}, \qquad (1-\rho)\tau = \sqrt{3}\sqrt{1 + \delta_{2\bm{s},\bm{M}}}.
\end{equation*}
Recalling that $e' =  AP_{\Delta^c}x + e$, we have
\begin{equation*}
    \Vert x^{(n+1)} - P_{\Delta} x \Vert_{\ell^{2}} \leq \rho \Vert x^{(n)} - P_{\Delta} x\Vert_{\ell^{2}} + (1-\rho)\tau \Vert AP_{\Delta^{c}}x + e \Vert_{\ell^{2}},
\end{equation*}
which, by induction on $n$, gives 
\begin{equation*}
    \Vert x^{(n)} - P_{\Delta}x \Vert_{\ell^{2}} \leq \rho^{n} \Vert x^{(0)} - P_{\Delta}x\Vert_{\ell^{2}} + \tau \Vert AP_{\Delta^{c}}x + e \Vert_{\ell^{2}}.
\end{equation*}
This was precisely the result to be shown, noting that
$\rho < 1$ if and only if $\delta_{3\bm{s},\bm{M}} < \frac{1}{\sqrt{3}}$ and so
\begin{equation*}
\tau = \tfrac{\sqrt{3}\sqrt{1+\delta_{3\bm{s},\bm{M}}}}{1- \rho} < \tfrac{\sqrt{3 +\sqrt{3}}}{1-\rho} < \tfrac{2.18}{1-\rho}.
\end{equation*}
This concludes the proof.
\end{proof}

\begin{proof} (Of \cref{thm:IHTFinalTheorem})
Using \cref{thm:IHTIntermediateBound} with $(2\bm{s},\bm{M})$ instead of $(\bm{s},\bm{M})$, there exist constants $\rho \in (0,1)$, $\tau > 0$ depending on $\delta_{6\bm{s},\bm{M}}$ such that
\begin{equation*}
\Vert x^{(n)} - P_{\Xi}x \Vert_{\ell^{2}} \leq \rho^{n} \Vert x^{(0)} - P_{\Xi}x\Vert_{\ell^{2}} + \tau \Vert AP_{\Xi^{c}}x + e \Vert_{\ell^{2}}.
\end{equation*}
where $\Xi = \Xi^{(0)} \cup \ldots \cup \Xi_{n}$ and $\Xi_i$ is the index set of the largest $2s_i$ entries of $P^{M_{i-1}}_{M_i}x$ (note that we applied Theorem~\ref{thm:IHTIntermediateBound} with $(2\bm{s},\bm{M})$ since $\Xi \in D_{2\bm{s},\bm{M}}$). Then, by letting $x' = x^{(n)}$ and $\lambda = \rho^{n}\Vert P_{\Xi}x \Vert_{\ell^{2}}$ (recall that $x^{(0)} = 0$) we may apply \cref{thm:WeightedTheorem} with $\kappa = 2$ to assert 
\begin{align*}
    \Vert x - x^{(n)} \Vert_{\ell^{1}_{w}} 
    &\leq C\tfrac{\sqrt{\zeta}}{\sqrt{\xi}}\sigma_{\bm{s,M}}(x)_{\ell_{w}^{1}} 
    + D \sqrt{\zeta} \Vert e \Vert_{\ell^{2}} 
    + 2 \sqrt{\zeta} \rho^{n} \Vert x \Vert_{\ell^{2}}, \\
\Vert x - x^{(n)} \Vert_{\ell^{2}} 
&\leq \tfrac{E}{\sqrt{\xi}}\sigma_{\bm{s,\bm{M}}}(x)_{\ell^{1}_{w}} + \tau \Vert e \Vert_{\ell^{2}} + \rho^{n} \Vert x \Vert_{\ell^{2}}
\end{align*}
where $C,D,E>0$ depend on $\tau, \rho, \kappa$ and thus only on $\delta_{6\bm{s}},\bm{M}$.
\end{proof}

\subsection{Proving \cref{thm:CoSamPlMainTheorem} on CoSaMP}
The following is based on \cite[Theorem 6.27]{FoucartRauhutCSbook}.

\begin{theorem} \label{CoSaMPLIntermediateTheorem} 
Suppose the $(4\bm{s},\bm{M})$-th RICL constant of the matrix $A \in \mathbb{C}^{m \times N}$ satisfies 
\begin{equation*}
    \delta_{4\bm{s},\bm{M}} < \tfrac{\sqrt{\sqrt{\frac{11}{3}} - 1}}{2}
\end{equation*}
Then for $x \in \mathbb{C}^{N}$, $e \in \mathbb{C}^{m}$ and index set $\Delta
\in D_{\bm{s},\bm{M}}$, the sequence $(x^{(n)})_{n \geq0}$ defined by $\mathrm{CoSaMPL}(A,y,\bm{s},\bm{M})$ with $y=Ax+e$ satisfies
\begin{equation*}
    \Vert x^{(n)} - P_{\Delta}x \Vert_{\ell^{2}} \leq \rho^{n} \Vert x^{(0)} - P_{\Delta}x \Vert_{\ell^{2}} + \tau \Vert AP_{\Delta^{c}}x + e \Vert_{\ell^{2}},
\end{equation*}
where $\rho \in (0,1)$ and $\tau >0$ are constants only depending on $\delta_{4\bm{s},\bm{M}}$.
\end{theorem}
\begin{proof} As before, with correct treatment of our index sets, many of the algebraic manipulations follow near-identically from \cite[Theorem 6.27]{FoucartRauhutCSbook}. We have
\begin{align*}
    \Vert P_{U^{(n+1)}}(P_{\Delta}x - x^{(n+1)}) \Vert_{\ell^{2}}
    &\leq \Vert u^{(n+1)} - x^{(n+1)} \Vert_{\ell^{2}} \\
    & \quad + \Vert u^{(n+1)} -  P_{U^{(n+1)} \cap \Delta}x \Vert_{\ell^{2}}. 
\end{align*}
Further as $x^{(n+1)} = H_{\bm{s},\bm{M}}(u^{(n+1)})$ we bound $\Vert u^{(n+1)} - x^{(n+1)} \Vert_{\ell^{2}} \leq  \Vert u^{(n+1)} -  P_{U^{(n+1)} \cap \Delta}x \Vert_{\ell^{2}}$. This result, combined with the fact that $P_{(U^{(n+1)})^{c} }x^{(n+1)} = P_{(U^{(n+1)})^{c} }u^{(n+1)} = 0$, 
asserts
\begin{align}
    \Vert P_{\Delta}x - x^{(n+1)} \Vert_{\ell^{2}}^{2} 
    &= \Vert P_{ (U^{(n+1)})^{c} }(P_{\Delta}x - x^{(n+1)}) \Vert_{\ell^{2}}^{2} \nonumber\\
    & \quad + \Vert  P_{U^{(n+1)}}(P_{\Delta}x - x^{(n+1)}) \Vert_{\ell^{2}}^{2} \nonumber \\
    &\leq \Vert P_{(U^{(n+1)})^{c} }(P_{\Delta}x - u^{(n+1)}) \Vert_{\ell^{2}}^{2} \nonumber \\
    & \quad + 4\Vert  P_{U^{(n+1)}}(P_{\Delta}x - u^{(n+1)}) \Vert_{\ell^{2}}^{2}. \label{CoSamPLFinalStep1}
\end{align}
We will us this bound later, but we now examine the latter term more closely. 

We first make the observation that $P_{U^{(n+1)}}A^{*}(y-Au^{(n+1)}) = 0,$
as $u^{(n+1)}$ satisfies the normal equations when restricted to its support. Thus we may write $P_{U^{(n+1)}}A^{*}A(P_{\Delta}x - u^{(n+1)}) = -P_{U^{(n+1)}}A^{*}e',$
where $e'=AP_{\Delta^{c}}x + e$. We use this to write
\begin{align*}
    & \Vert P_{U^{(n+1)}}(P_{\Delta}x - u^{(n+1)}) \Vert_{\ell^{2}} \\
    & \qquad \leq \Vert (I - A^{*}A)(P_{\Delta}x - u^{(n+1)}) \Vert_{\ell^{2}} 
    + \Vert P_{U^{(n+1)}}A^{*}e' \Vert_{\ell^{2}}.
\end{align*}
Now as $\Delta \in D_{\bm{s},\bm{M}}$ and $U^{(n+1)} \in D_{3\bm{s},\bm{M}}$, we have their union is in $D_{4\bm{s},\bm{M}}$. Thus using \cref{lemma:616Generalized} (ii) gives $\Vert (I - A^{*}A)(P_{\Delta}x - u^{(n+1)}) \Vert_{\ell^{2}} \leq \delta_{4\bm{s},\bm{m}} \Vert P_{\Delta}x - u^{(n+1)} \Vert_{\ell^{2}}$, and so
\begin{align}
   &   \Vert P_{U^{(n+1)}}(P_{\Delta}x - u^{(n+1)}) \Vert_{\ell^{2}} \nonumber \\
   & \qquad \leq \delta_{4\bm{s},\bm{m}} \Vert P_{\Delta}x - u^{(n+1)} \Vert_{\ell^{2}} + \Vert P_{U^{(n+1)}}A^{*}e' \Vert_{\ell^{2}}. \label{Step2CoSaMPIndexSets}
\end{align}
From here denoting $\delta_{4\bm{s},\bm{M}} = \delta$, we wish to derive the inequality
\begin{align}
   &   \Vert P_{U^{(n+1)}}(P_{\Delta}x - u^{(n+1)}) \Vert_{\ell^{2}} \label{CoSaMplFinalStep2}\\
   & \quad \leq \frac{\delta\Vert P_{(U^{(n+1))^{c} }}(P_{\Delta}x - u^{(n+1)} )\Vert_{\ell^{2}}}{\sqrt{1-\delta^{2}}} 
    + \frac{\Vert P_{U^{(n+1)}}A^{*}e' \Vert_{\ell^{2}}}{1-\delta} . \nonumber
\end{align}
For the sake of brevity, we assert the desired inequality follows from purely algebraic manipulations, and uses no sparsity properties. These steps for the sparse case can be found in detail in \cite[Theorem~6.28, p.~166]{FoucartRauhutCSbook}. As before, we save this bound for later use, and switch to a final argument.

We first recall that CoSaMPL defines $S^{(n)} = \text{supp}(x^{({n})})$, and that $S^{(n)} \subset U^{(n+1)}$.  Further we define $T^{(n+1)} = L_{2\bm{s},\bm{M}}(A^{*}(y-Ax^{(n)}))$. As this is the index set of the largest $(2\bm{s},\bm{M})$ entries of $A^{*}(y-Ax^{(n)})$, and $\Delta \cup S^{(n)} \in D_{2\bm{s},\bm{M}} $ we have
\begin{equation*}
    \Vert P_{\Delta \cup S^{(n)}}A^{*}(y-Ax^{(n)}) \Vert_{\ell^{2}} 
    \leq \Vert P_{T^{(n+1)}} A^{*}(y-Ax^{(n)}) \Vert_{\ell^{2}}.
\end{equation*}
In turn, eliminating the shared contribution on $(\Delta \cup S^{(n)}) \cap T^{(n+1)}$ we find
\begin{align*}
    & \Vert P_{(\Delta \cup S^{(n)}) \backslash T^{(n+1)}} A^{*}(y-Ax^{(n)}) \Vert_{\ell^{2}} \\
    & \qquad \qquad \qquad \leq \Vert P_{T^{(n+1)} \backslash (\Delta \cup S^{(n)})} A^{*}(y-Ax^{(n)}) \Vert_{\ell^{2}}.
\end{align*}
Now as $P_{\Delta}x - x^{(n)} = 0$ on $T^{(n+1)} \backslash (\Delta \cup S^{(n)})$ we may write the right-hand side of the above as
\begin{equation*}
    \Vert P_{T^{(n+1)} \backslash (\Delta \cup S^{(n)})}(P_{\Delta}x - x^{(n)} + A^{*}(y-Ax^{(n)})) \Vert_{\ell^{2}},
\end{equation*}
whereas for the left-hand side we apply a reverse triangle inequality
\begin{align*}
    & \Vert P_{(\Delta \cup S^{(n)}) \backslash T^{(n+1)}} A^{*}(y-Ax^{(n)}) \Vert_{\ell^{2}}\\
    &\quad \geq \Vert P_{(\Delta \cup S^{(n)}) \backslash T^{(n+1)}}(P_{\Delta}x - x^{(n)} )\Vert_{\ell^{2}} \\
    & \quad \quad - \Vert P_{(\Delta \cup S^{(n)}) \backslash T^{(n+1)}}( P_{\Delta}x - x^{(n)} +A^{*}(y-Ax^{(n)}) )\Vert_{\ell^{2}} \\
    &\quad = \Vert P_{(T^{(n+1)})^{c}}(P_{\Delta}x - x^{(n)} )\Vert_{\ell^{2}} \\
    & \quad\quad  - \Vert P_{(\Delta \cup S^{(n)}) \backslash T^{(n+1)}}( P_{\Delta}x - x^{(n)} +A^{*}(y-Ax^{(n)}) )\Vert_{\ell^{2}}.
\end{align*}
Combining these two observations and rearranging gives
\begin{align*}
    &\Vert P_{(T^{(n+1)})^{c}}(P_{\Delta}x - x^{(n)} )\Vert_{\ell^{2}}\\
    &\quad\leq \sqrt{2} \Vert P_{T^{(n+1)} \ominus (\Delta \cup S^{(n)})}(P_{\Delta}x - x^{(n)} + A^{*}(y-Ax^{(n)})) \Vert_{\ell^{2}} \\
    &\quad\leq \sqrt{2}\Vert  P_{T^{(n+1)} \ominus (\Delta \cup S^{(n)})}(I - A^{*}A)(x^{(n)} - P_{\Delta}x) \Vert_{\ell^{2}} \\
    & \quad \qquad \qquad + \sqrt{2} \Vert  P_{T^{(n+1)} \ominus (\Delta \cup S^{(n)})}A^{*}e' \Vert_{\ell^{2}},
\end{align*}
where $\ominus$ denotes the symmetric difference, and $y = AP_{\Delta}x+e'$ is as before. Now, as $T^{(n+1)} \subset U^{(n+1)}$ and $S^{(n)} \subset U^{(n+1)}$ by the definition of CoSaMPL, we may bound the left-hand side of the above equation from below by
\begin{align*}
     \Vert P_{(T^{(n+1)})^{c}}(P_{\Delta}x - x^{(n)} )\Vert_{\ell^{2}} 
     & \geq  \Vert P_{(U^{(n+1)})^{c}}(P_{\Delta}x - x^{(n)} )\Vert_{\ell^{2}}\\
      & = \Vert P_{(U^{(n+1)})^{c}} P_{\Delta}x \Vert_{\ell^{2}} \\
     & = \Vert P_{ (U^{(n+1)})^{c} }(P_{\Delta}x - u^{(n+1)}) \Vert_{\ell^{2}}.
\end{align*}
With this lower bound in hand, we note that $\Delta,S^{(n)} \in D_{\bm{s},\bm{M}}$ and $T^{(n+1)} \in D_{2\bm{s},\bm{M}}$ so that we may apply \cref{lemma:620generalized} (ii) with $T^{(n+1)} \ominus (\Delta \cup S^{(n)}) \subset T^{(n+1)} \cup (\Delta \cup S^{(n)}) \in D_{4\bm{s},\bm{M}}$ on the term $\Vert  P_{T^{(n+1)} \ominus (\Delta \cup S^{(n)})}(I - A^{*}A)(x^{n} - P_{\Delta}x) \Vert_{\ell^{2}}$.  Combining this series of observations gives
\begin{align}
    & \Vert P_{{ (U^{(n+1)})^{c} }}(P_{\Delta}x - u^{(n+1)}) \Vert_{\ell^{2}} \label{CoSaMPLFinalStep3} \\
    &  \leq \sqrt{2} \delta_{4\bm{s},\bm{M}} \Vert x^{n} - P_{\Delta}x \Vert_{\ell^{2}} 
    + \sqrt{2} \Vert  P_{T^{(n+1)} \ominus (\Delta \cup S^{(n)})}(A^{*}e') \Vert_{\ell^{2}}. \nonumber
\end{align}

To conclude our argument, it remains to combine the three distinct results of equations \cref{CoSamPLFinalStep1}, \cref{CoSaMplFinalStep2} and \cref{CoSaMPLFinalStep3}. Again, this is near identical to the sparse case in \cite[Theorem 6.27]{FoucartRauhutCSbook}, and contains purely algebraic manipulations. This leads to the inequality
\begin{align*}
    \Vert P_{\Delta}x - x^{(n+1)} \Vert_{\ell^{2}} 
    & \leq \sqrt{\tfrac{ 2\delta^{2}(1 +3\delta^{2})}{1-\delta^{2}} } \Vert x^{(n)} - P_{\Delta}x \Vert_{\ell^{2}} \\
    & \quad + \sqrt{ \tfrac{2 (1 +3\delta^{2})}{1-\delta^{2}} } \Vert  P_{T^{(n+1)} \ominus (\Delta \cup S^{(n)})} A^{*}e' \Vert_{\ell^{2}} \\
    & \quad  + \tfrac{2}{1-\delta}\Vert P_{U^{(n+1)}}A^{*}e' \Vert_{\ell^{2}}.
\end{align*}
Now using \cref{lemma:620generalized} on the sets $T^{(n+1)} \ominus (\Delta \cup S^{(n)}) \in D_{4\bm{s},\bm{M}}$ and $U^{(n+1)} \in D_{3\bm{s},\bm{M}} \subset D_{4\bm{s},\bm{M}}$ we find
\begin{align*}
    & \Vert P_{\Delta}x - x^{(n+1)} \Vert_{\ell^{2}} \leq \sqrt{ \tfrac{2\delta^{2}(1 +3\delta^{2})}{1-\delta^{2}}} \Vert x^{(n)} - P_{\Delta}x \Vert_{\ell^{2}} \\
    & \quad \qquad \qquad + \left( \sqrt{ \tfrac{2(1+\delta)(1 +3\delta^{2})}{1-\delta^{2}}} + \tfrac{2\sqrt{1+\delta}}{1-\delta}\right) \Vert  e'\Vert_{\ell^{2}},
\end{align*}
which is exactly
\begin{equation*}
    \Vert P_{\Delta}x - x^{(n+1)} \Vert_{\ell^{2}} \leq \rho \Vert x^{(n)} - P_{\Delta}x \Vert_{\ell^{2}} + \tau \Vert AP_{\Delta^{c}}x + e \Vert_{\ell^{2}},
\end{equation*}
for suitable $\rho, \tau >0$ depending only on $\delta$. Then by a simple induction we have
\begin{equation*}
     \Vert x^{(n)} - P_{\Delta}x \Vert_{\ell^{2}} \leq \rho^{n} \Vert x^{(0)} - P_{\Delta}x \Vert_{\ell^{2}} + \tau \Vert AP_{\Delta^{c}}x + e \Vert_{\ell^{2}}.
\end{equation*}
Which is precisely the result to be shown. Notably, the constant $\rho < 1$ only if 
$
    \sqrt{ \frac{2\delta^{2}(1 +3\delta^{2})}{1-\delta^{2}}} < 1 \Leftrightarrow 6\delta^{4} + 3 \delta^{2} -1 <0,
$
which by solving this quadratic in $\delta^{2}$ for its largest root gives us that we require $
    \delta^{2} < \frac{\sqrt{\frac{11}{3}} - 1}{4}$
as was assumed.
\end{proof}

\begin{proof} (Of \cref{thm:CoSamPlMainTheorem}) Under the hypotheses of the theorem, let us denote $\Xi = L_{2\bm{s},\bm{M}}(x)$ to be the index set corresponding to the largest $(2\bm{s},\bm{M})$ entries of $x$. First, we may apply \cref{CoSaMPLIntermediateTheorem} to assert there exist $\rho \in (0,1)$ and $\tau > 0$ depending only on $\delta_{8\bm{s},\bm{M}}$ such that, for any $n\geq 0$,
\begin{equation*}
     \Vert x^{(n)} - P_{\Xi}x \Vert_{\ell^{2}} \leq \rho^{n} \Vert P_{\Xi}x \Vert_{\ell^{2}} + \tau \Vert AP_{\Xi^{c}}x + e \Vert_{\ell^{2}}.
\end{equation*}
Then, we may apply \cref{thm:WeightedTheorem} with $x' = x^{(n)}$ and $\lambda = \rho^{n} \Vert P_{\Xi}x \Vert_{\ell^{2}} \leq \rho^{n} \Vert x \Vert_{\ell^{2}}$ to give us that
\begin{align*}
    \Vert x - x^{(n)} \Vert_{\ell^{1}_{w}} 
    &\leq C\tfrac{\sqrt{\zeta}}{\sqrt{\xi}}\sigma_{\bm{s,M}}(x)_{\ell_{w}^{1}} 
    + D \sqrt{\zeta} \Vert e \Vert_{\ell^{2}} 
    + 2 \sqrt{\zeta} \rho^{n} \Vert x \Vert_{\ell^{2}}, \\
\Vert x - x^{(n)} \Vert_{\ell^{2}} 
&\leq \tfrac{E}{\sqrt{\xi}}\sigma_{\bm{s,\bm{M}}}(x)_{\ell^{1}_{w}} + \tau \Vert e \Vert_{\ell^{2}} + \rho^{n} \Vert x \Vert_{\ell^{2}}
\end{align*}
where $C,D,E >0$ depend only on $\tau,\rho,\delta_{8\bm{s},\bm{M}}$ and thus only on $\delta_{8\bm{s},\bm{M}}$. This is exactly the result that was to be shown.
\end{proof}

\subsection{Proving Theorem \ref{thm:simplerQCBPresult} on QCBP}
\label{ss:simplerQCBPresultproof}

\begin{proof}[Proof of Theorem \ref{thm:simplerQCBPresult}]
A matrix $A \in \mathbb{C}^{m \times N}$ has the weighted robust null space property in levels of order $(\bm{s},\bm{M})$ with constants $0 < \rho < 1$ and $\gamma > 0$ if
$$
\Vert P_{\Delta}x \Vert_{2} \leq \frac{\rho \Vert P_{\Delta^{c}}x \Vert_{\ell^{1}_{w}}}{\sqrt{\zeta } } + \gamma \Vert Ax \Vert_{2},
$$
for all $x \in \mathbb{C}^{n}$ and $\Delta \in D_{\bm{s},\bm{M}}$, see \cite[Defn.\ 5.1]{AAHWalshWavelet}. Modifying the proof of Theorem 5.4 of \cite{AAHWalshWavelet} in a minor way, it can be shown that if $A$ has this property, then any minimizer $\hat{x}$ of \eqref{QCBP} with $y = A x  +e$, where $x \in \mathbb{C}^N$ and $\Vert e \Vert_{\ell^2} \leq \eta$, satisfies
\begin{align*}
\Vert x - \hat{x} \Vert_{\ell^{1}_{w}} 
    &\leq C \sigma_{\bm{s,M}}(x)_{\ell_{w}^{1}} 
    + D \gamma \sqrt{\zeta} \eta, 
    \\
\Vert x - \hat{x} \Vert_{\ell^{2}} &\leq \left ( 1 + (\zeta / \xi)^{1/4} \right ) \left ( \frac{E}{\sqrt{\xi}}\sigma_{\bm{s,\bm{M}}}(x)_{\ell^{1}_{w}} + F \gamma \eta \right ),
\end{align*}
where $C,D,E,F$ depend on $\rho$ only. Therefore, it suffices to show that if $A$ has the RIPL of order $(2 \bm{s},\bm{M})$ with constant $\delta_{2 \bm{s} , \bm{M}} < \frac{1}{\sqrt{\frac{2\zeta}{\xi}} +1}$ then it also has the weighted robust null space property in levels. 

We begin by performing a very similar decomposition as in the proof of Theorem~\ref{thm:IHTIntermediateBound}. For $l = 1, \ldots r$, let $\Xi_{0,l}$ be the index set to the largest $s_{l}$ entries of $P_{M_{l}}^{M_{l-1}}x$ in absolute value. Then, define $\Xi^{(0)}= \Xi_{0,1} \cup \ldots \cup \Xi_{0,r}$. For such index set, we decompose $(\Xi^{(0)})^{c}=\{1,\ldots,N\}\setminus \Xi^{(0)}$, letting
\begin{equation*}
    \Xi_{0,l}^{c} = \Xi_{1,l} \cup \Xi_{2,l} \cup \ldots,
\end{equation*}
where $\Xi_{1,l}$ is index set of the $s_{l}$ largest entries of $P_{\Xi_{0,l}^{c}}x$, $\Xi_{2,l}$ is the index set of the largest $s_{l}$ entries of $P_{(\Xi_{0,l} \cup \Xi_{1,l})^{c}}x$, and so on, letting $\Xi_{i,l} = \emptyset$ as needed for sufficiently large $i$. (Note that $\Xi_{0,l}^c$, $(\Xi_{0,l} \cup \Xi_{1,l})^{c}$, etc. are relative complements with respect to the level $l$). Finally we define $\Xi^{(i)} = \Xi_{i,1} \cup \ldots \cup \Xi_{i,r}$ for each $i = 1,2, \ldots$. Then
\begin{equation}
    \Vert P_{\Xi^{(0)} \cup \Xi^{(1)}} x\Vert_{\ell^{2}}^{2} \leq \tfrac{1}{1- \delta_{2\bm{s},\bm{M}} } \Vert A  P_{\Xi^{(0)} \cup \Xi^{(1)}} x \Vert_{\ell^{2}}^{2}, \label{eqn:wRNSP1}
\end{equation}
as by assumption $A$ has the RIPL of order $(2\bm{s},\bm{M})$. Then expanding according to the partition and using the RIPL again we obtain
\begin{align}
    \Vert A  P_{\Xi^{(0)} \cup \Xi^{(1)}}x \Vert_{\ell^{2}}^{2}  
    & \leq \sqrt{1 + \delta_{2\bm{s},\bm{M}}}  \Vert P_{\Xi^{(0)} \cup \Xi^{(1)}}x \Vert_{\ell^{2}} \Vert Ax \Vert_{\ell^{2}} \nonumber \\
    & \quad + \sum_{i \geq 2} \vert \langle AP_{\Xi^{(0)} \cup \Xi^{(1)}}x , AP_{\Xi^{(i)}}x \rangle \vert. \label{eqn:wRNSP2}
\end{align}
Now, for $i \geq 2$, using that $\vert \langle P_{\Xi^{(0)} \cup \Xi^{(1)}}x , P_{\Xi^{(i)}}x \rangle \vert = 0$ in tandem with Lemma~\ref{lemma:616Generalized} (i) we see that
\begin{align*}
    & \vert \langle AP_{\Xi^{(0)} \cup \Xi^{(1)}}x , AP_{\Xi^{(i)}}x \rangle \vert \\
    &\qquad \qquad \leq \delta_{2\bm{s},\bm{M}} (\Vert P_{\Xi^{(0)}} x \Vert_{\ell^{2}} + \Vert P_{\Xi^{(1)}}x \Vert_{\ell^{2}} ) \Vert P_{\Xi^{(i)}}x \Vert_{\ell^{2}} \\
 &\qquad\qquad  \leq  \sqrt{2}\delta_{2\bm{s},\bm{M}}  \Vert P_{\Xi^{(0)} \cup \Xi^{(1)}}x \Vert_{\ell^{2}} \Vert P_{\Xi^{(i)}}x \Vert_{\ell^{2}}.
\end{align*}
Furthermore using bounds \eqref{eqn:wRNSP1} and \eqref{eqn:wRNSP2} and the RIPL we have 
\begin{align}
    & \Vert P_{\Xi^{(0)} \cup \Xi^{(1)}}x \Vert_{\ell^{2}} \nonumber \\
    & \quad \leq \tfrac{\sqrt{1 + \delta_{2\bm{s},\bm{M}}}}{1 - \delta_{2\bm{s},\bm{M}}} \Vert Ax \Vert_{\ell^{2}} + \sqrt{2} \tfrac{\delta_{2\bm{s},\bm{M}}}{1- \delta_{2\bm{s},\bm{M}}} \sum_{i \geq 2} \Vert P_{\Xi^{(i)}}x \Vert_{\ell^{2}}. \label{eqn:wRNSPL3}
\end{align}
Recalling our goal is the show the weighted robust null space property, we need to relate the summation in the latter term to $\Vert P_{\Delta}^{c}x \Vert_{\ell^{1}_{w}}$, where $\Delta = \Xi^{(0)}$. But by construction
\begin{align*}
    \Vert P_{\Xi_{i,l}}x \Vert_{\ell^{2}} 
    & \leq \sqrt{s_{l}} \Vert P_{\Xi_{i,l}}x \Vert_{\ell^{\infty}} 
    \leq \sqrt{s_{l}} \min_{j \in \Xi_{i-1,l}} \vert x_{j} \vert \\
    & \leq\tfrac{  \Vert P_{\Xi_{i-1,l}}x \Vert_{\ell^{1}}}{\sqrt{s_{l}}} 
    =  \tfrac{\Vert P_{\Xi_{i-1,l}}x \Vert_{\ell^{1}_{w}}}{w_{l}\sqrt{s_{l}}}.
\end{align*}
Thus overall, we have
\begin{align*}
    \Vert P_{\Xi^{(i)}} \Vert_{\ell^{2}}^{2} &= \sum_{l=1}^{r} \Vert P_{\Xi_{i,l}}x \Vert_{\ell^{2}}^{2} \leq \sum_{l=1}^{r} \left( \tfrac{\Vert P_{\Xi_{i-1,l}}x \Vert_{\ell^{1}_{w}}}{w_{l}\sqrt{s_{l}}} \right)^{2} \\
    & \leq \tfrac{1}{\xi} \Vert P_{\Xi^{(i-1)}}x \Vert_{\ell^{1}_{w}}^{2}.
\end{align*}
And hence,
\begin{equation*}
    \sum_{i\geq 2}  \Vert P_{\Xi_{i,l}}x \Vert_{\ell^{2}} 
    \leq \tfrac{1}{\sqrt\xi} \sum_{i\geq 2} \vert P_{\Xi^{(i-1)}}x \Vert_{\ell^{1}_{w}}  
    = \tfrac{1}{\sqrt{\xi}} \Vert P_{\Delta^{c}} x \Vert_{\ell^{1}_{w}}.
\end{equation*}
Combining this with \eqref{eqn:wRNSPL3} gives that
\begin{align*}
    &\Vert P_{\Delta}x \Vert_{\ell^{2}} 
    \leq \Vert P_{\Xi^{(0)} \cup \Xi^{(1)}}x \Vert_{\ell^{2}} \\
     & \quad\leq \tfrac{\sqrt{1 + \delta_{2\bm{s},\bm{M}}}}{1 - \delta_{2\bm{s},\bm{M}}} \Vert Ax \Vert_{\ell^{2}} + \sqrt{2} \tfrac{\delta_{2\bm{s},\bm{M}}}{1- \delta_{2\bm{s},\bm{M}}} \tfrac{1}{\sqrt\xi} \Vert P_{\Delta^{c}} x \Vert_{\ell^{1}_{w}} \\
    & \quad= \tfrac{\sqrt{1 + \delta_{2\bm{s},\bm{M}}}}{1 - \delta_{2\bm{s},\bm{M}}} \Vert Ax \Vert_{\ell^{2}} + \sqrt{2} \tfrac{\delta_{2\bm{s},\bm{M}}}{1- \delta_{2\bm{s},\bm{M}}} \tfrac{\sqrt{\zeta}}{\sqrt\xi} \frac{\Vert P_{\Delta^{c}} x \Vert_{\ell^{1}_{w}}}{\sqrt{\zeta}}.
\end{align*}
Hence, $A$ has the wrNSPL provided
\begin{equation*}
    \sqrt{2} \tfrac{\delta_{2\bm{s},\bm{M}}}{1- \delta_{2\bm{s},\bm{M}}} \tfrac{\sqrt{\zeta}}{\sqrt\xi}  < 1,
\end{equation*}
or namely $\delta_{2\bm{s},\bm{M}} < \frac{1}{\sqrt{\frac{2\zeta}{\xi}} +1}$, as required.
\end{proof}

\section{Open problems\label{sec:conclusion}}

We conclude this work by discussing directions of future interest and open questions. In the context of our numerical experiments, we have observed that the levels based version of normalized IHT  \cite{BlumensathDavies2010} consistently outperforms IHTL. Therefore, a natural question is whether the theoretical guarantees for normalized IHT proved in \cite{BlumensathDavies2010} can be generalized to the levels based setting.


Another open problem directly following this work is whether other greedy algorithms can be generalized, with OMP being a first candidate. Any algorithm that permits a result of the style of \cref{thm:IHTIntermediateBound} or \cref{CoSaMPLIntermediateTheorem} would allow application of \cref{thm:WeightedTheorem} for a stability estimate, but more sophisticated techniques may be required in general. This is of particular interest in the case of OMP, which when generalized correctly performs well numerically in some cases. However OMPL is not a uniform improvement over OMP in our experiments, contrasting with IHTL and CoSaMPL. This provokes the natural question of whether the formulation of OMPL here is the best possible - and if another variant would improve further. Of similar importance are the Matching Pursuit \cite{mallatzhang1993matching} and Subspace Pursuit algorithms \cite{dai2009subspace}, which have not yet been studied in this context, but are well studied in the sparse setting \cite{mallatzhang1993matching,PatiOrthogonalMatchingPursuit}. These have the potential to perform well when correctly generalized to the sparse in levels class, lending further foundation for practical use. 

Finally, applying these results to optimal function approximation is an open problem. Thus far, encoder-decoder pairs have been optimization programs without guaranteed computational cost \cite{BASBMKRCSwavelet}. These iterative approaches, with computational guarantees, may serve to replace optimization programs in these problems - allowing for known computational time \text{a priori}. This would be of particular use in, for example, imaging problems, where the optimization approaches already have been shown to perform well both in theory and practice~\cite{BASBMKRCSwavelet}.

\section*{Acknowledgment}

BA and MKR acknowledge the support of the PIMS CRG ``High-dimensional Data Analysis'', SFU's Big Data Initiative ``Next Big Question'' Fund and by NSERC through grant R611675. SB acknowledges NSERC through grant RGPIN-2020-06766 and the Faculty of Arts and Science of Concordia University. The authors would like to thank the anonymous reviewers whose valuable comments led to a significant improvement of the paper.

\bibliographystyle{plain}
\bibliography{references}

\begin{thebibliography}{10}

\bibitem{BAEtAlCorruptions}
B.~Adcock, A.~Bao, J.~D. Jakeman, and A.~Narayan.
\newblock Compressed sensing with sparse corruptions: {F}ault-tolerant sparse
  collocation approximations.
\newblock {\em SIAM/ASA J. Uncertain. Quantif.}, 6(4):1424--1453, 2018.

\bibitem{BASBMKRCSwavelet}
B.~Adcock, S.~Brugiapaglia, and M.~King-Roskamp.
\newblock {Do log factors matter? On optimal wavelet approximation and the
  foundations of compressed sensing}.
\newblock {\em Found. Comput. Math.}, pages 1--61, 2021.

\bibitem{AHPRBreaking}
B.~Adcock, A.~C. Hansen, C.~Poon, and B.~Roman.
\newblock Breaking the coherence barrier: A new theory for compressed sensing.
\newblock {\em Forum Math. Sigma}, 5, 2017.

\bibitem{OptimalSamplingQuest}
B.~Adcock, A.~C. Hansen, and B.~Roman.
\newblock The quest for optimal sampling: computationally efficient,
  structure-exploiting measurements for compressed sensing.
\newblock In {\em Compressed Sensing and Its Applications}. Birkh\"auser, 2015.

\bibitem{adcock2019iterative}
Ben Adcock, Simone Brugiapaglia, and Matthew King-Roskamp.
\newblock Iterative and greedy algorithms for the sparsity in levels model in
  compressed sensing.
\newblock In {\em Wavelets and Sparsity XVIII}, volume 11138, page 1113809.
  International Society for Optics and Photonics, 2019.

\bibitem{AAHWalshWavelet}
V.~Antun, B.~Adcock, and A.~C. Hansen.
\newblock Uniform recovery in infinite-dimensional compressed sensing and
  applications to structured binary sampling.
\newblock {\em arXiv:1905.00126}, 2019.

\bibitem{bach2012structured}
F.~Bach, R.~Jenatton, J.~Mairal, and G.~Obozinski.
\newblock Structured sparsity through convex optimization.
\newblock {\em Statist. Sci.}, 27(4):450--468, 2012.

\bibitem{baraniuk2010model}
R.~G. Baraniuk, V.~Cevher, M.~F. Duarte, and C.~Hedge.
\newblock Model-based compressive sensing.
\newblock {\em IEEE Trans. Inform. Theory}, 56(4):1982--2001, 2010.

\bibitem{BastounisHansen}
A.~Bastounis and A.~C. Hansen.
\newblock On the absence of uniform recovery in many real-world applications of
  compressed sensing and the restricted isometry property and nullspace
  property in levels.
\newblock {\em SIAM J. Imaging Sci.}, 10(1):335--371, 2017.

\bibitem{blumensath2009sampling}
M.~E. Blumensath, T.and~Davies.
\newblock Sampling theorems for signals from the union of finite-dimensional
  linear subspaces.
\newblock {\em IEEE Trans. Inform. Theory}, 55(4):1872--1882, 2009.

\bibitem{Blumensath2012}
T.~Blumensath.
\newblock Accelerated iterative hard thresholding.
\newblock {\em Signal Process.}, 92:752--756, 2012.

\bibitem{BlumensathDavies2008}
T.~Blumensath and M.~E. Davies.
\newblock Iterative thresholding for sparse approximations.
\newblock {\em J. Fourier Anal. Appl.}, 14:629--654, 2008.

\bibitem{BlumensathDavies2009}
T.~Blumensath and M.~E. Davies.
\newblock Iterative hard thresholding for compressed sensing.
\newblock {\em Appl. Comput. Harmon. Anal.}, 27:265?274, 2009.

\bibitem{BlumensathDavies2010}
T.~Blumensath and M.~E. Davies.
\newblock Normalized iterative hard thresholding: Guaranteed stability and
  performance.
\newblock {\em IEEE J. Sel. Top. Signal Process.}, 4(2), 2010.

\bibitem{BoyerBlockStructured}
C.~Boyer, J.~Bigot, and P.~Weiss.
\newblock Compressed sensing with structured sparsity and structured
  acquisition.
\newblock {\em Appl. Comput. Harmon. Anal.}, 46(2):312--350, 2019.

\bibitem{brugiapaglia2018robustness}
Simone Brugiapaglia and Ben Adcock.
\newblock Robustness to unknown error in sparse regularization.
\newblock {\em IEEE Trans. Inform. Theory}, 64(10):6638--6661, 2018.

\bibitem{CaiZhangRIP2}
T.~Cai and A.~Zhang.
\newblock {Sharp RIP bound for sparse signal and low-rank matrix recovery}.
\newblock {\em Appl. Comput. Harmon. Anal.}, 35(1):74--93, 2013.

\bibitem{candes2008restricted}
E.J. Candes.
\newblock The restricted isometry property and its implications for compressed
  sensing.
\newblock {\em C. R. Math. Acad. Sci. Paris}, 346(9-10):589--592, 2008.

\bibitem{ChunAdcock16ITW}
I.~Chun and B.~Adcock.
\newblock Optimal sparse recovery for multi-sensor measurements.
\newblock In {\em IEEE Inf. Theory Workshop (ITW) 2016}, 2016.

\bibitem{AdcockChunParallel}
I.-Y. Chun and B.~Adcock.
\newblock Compressed sensing and parallel acquisition.
\newblock {\em IEEE Trans. Inform. Theory}, 63(8):4860--4882, 2017.

\bibitem{dai2009subspace}
Wei Dai and Olgica Milenkovic.
\newblock Subspace pursuit for compressive sensing signal reconstruction.
\newblock {\em IEEE Trans. Inform. Theory}, 55(5):2230--2249, 2009.

\bibitem{dirksen2016dimensionality}
S.~Dirksen.
\newblock Dimensionality reduction with subgaussian matrices: a unified theory.
\newblock {\em Found. Comput. Math.}, 16(5):1367--1396, 2016.

\bibitem{donato2020structured}
J.~S. Donato and H.~W. Levinson.
\newblock {Structured Iterative Hard Thresholding with Off-Grid Applications}.
\newblock {\em arXiv preprint arXiv:2012.12783}, 2020.

\bibitem{Dorsch2016}
D.~Dorsch and H.~Rauhut.
\newblock Refined analysis of sparse mimo radar.
\newblock {\em J. Fourier Anal. Appl.}, pages 1--45, 2016.

\bibitem{DuarteEldarStructuredCS}
M.~F. Duarte and Y.~C. Eldar.
\newblock Structured compressed sensing: from theory to applications.
\newblock {\em IEEE Trans. Signal Process.}, 59(9):4053--4085, 2011.

\bibitem{duarte2011structured}
M.~F. Duarte and Y.~C. Eldar.
\newblock Structured compressed sensing: From theory to applications.
\newblock {\em IEEE Trans. Signal Process.}, 59(9):4053--4085, 2011.

\bibitem{elad2007wide}
M.~Elad, B.~Matalon, J.~Shtok, and M.~Zibulevsky.
\newblock A wide-angle view at iterated shrinkage algorithms.
\newblock In {\em Wavelets XII}, volume 6701, page 670102. International
  Society for Optics and Photonics, 2007.

\bibitem{eldar2009robust}
Y.~C. Eldar and M.~Mishali.
\newblock Robust recovery of signals from a structured union of subspaces.
\newblock {\em IEEE Trans. Inform. Theory}, 55(11):5302--5316, 2009.

\bibitem{FoucartRauhutCSbook}
S.~Foucart and H.~Rauhut.
\newblock {\em A Mathematical Introduction to Compressive Sensing}.
\newblock Birkhauser, 2013.

\bibitem{HedgeIndykSchmidt2015}
C.~Hegde, P.~Indyk, and L.~Schmidt.
\newblock Approximation algorithms for model-based compressive sensing.
\newblock {\em IEEE Trans. Inform. Theory}, 61(9), 2015.

\bibitem{junge2017generalized}
M.~Junge and K.~Lee.
\newblock Generalized notions of sparsity and restricted isometry property.
  part i: A unified framework.
\newblock {\em arXiv preprint arXiv:1706.09410}, 2017.

\bibitem{LiAdcockRIP}
C.~Li and B.~Adcock.
\newblock Compressed sensing with local structure: uniform recovery guarantees
  for the sparsity in levels class.
\newblock {\em Appl. Comput. Harmon. Anal.}, 46:453---477, 2019.

\bibitem{LiCorruptionsConstrApprox}
X.~Li.
\newblock Compressed sensing and matrix completion with a constant proportion
  of corruptions.
\newblock {\em Constr. Approx.}, 37:73--99, 2013.

\bibitem{lu2008theory}
Y.~M. Lu and M.~N. Do.
\newblock A theory for sampling signals from a union of subspaces.
\newblock {\em IEEE Trans. Signal Process.}, 56(6):2334--2345, 2008.

\bibitem{mallatzhang1993matching}
St{\'e}phane~G Mallat and Zhifeng Zhang.
\newblock Matching pursuits with time-frequency dictionaries.
\newblock {\em IEEE Trans. Signal Process.}, 41(12):3397--3415, 1993.

\bibitem{micchelli2013regularizers}
C.~A. Micchelli, J.~M. Morales, and M.~Pontil.
\newblock Regularizers for structured sparsity.
\newblock {\em Adv. Comput. Math.}, 38(3):455--489, 2013.

\bibitem{NeedellTropp2008}
D.~Needell and J.~Tropp.
\newblock Cosamp: Iterative signal recovery from incomplete and inaccurate
  samples.
\newblock {\em Appl. Comput. Harmon. Anal.}, 26(3):301--321, 2008.

\bibitem{PatiOrthogonalMatchingPursuit}
Y.~C. Pati, R.~Rezaiifar, and P.~S. Krishnaprasad.
\newblock Orthogonal matching pursuit: Recursive function approximation with
  applications to wavelet decomposition.
\newblock In {\em in Conference Record of The Twenty-Seventh Asilomar
  Conference on Signals, Systems and Computers}, pages 1--3, 1993.

\bibitem{AsymptoticCS}
B.~Roman, A.~C. Hansen, and B.~Adcock.
\newblock On asymptotic structure in compressed sensing.
\newblock {\em arXiv:1406.4178}, 2014.

\bibitem{TraonmilinGribonvalRIP}
Y.~Traonmilin and R.~Gribonval.
\newblock Stable recovery of low-dimensional cones in {H}ilbert spaces: {O}ne
  {RIP} to rule them all.
\newblock {\em Appl. Comput. Harm. Anal.}, 45(1):170--205, 2018.

\bibitem{tropp2007signal}
Joel~A Tropp and Anna~C Gilbert.
\newblock Signal recovery from random measurements via orthogonal matching
  pursuit.
\newblock {\em IEEE Trans. Inform. Theory}, 53(12):4655--4666, 2007.

\bibitem{yu2011solving}
G.~Yu, G.~Sapiro, and S.~Mallat.
\newblock Solving inverse problems with piecewise linear estimators: From
  gaussian mixture models to structured sparsity.
\newblock {\em IEEE Trans. Image Process.}, 21(5):2481--2499, 2011.

\end{thebibliography}

%

%
%
%





\end{document}